%% file: Eprint-final-submit.tex
\documentclass[runningheads]{llncs}
\makeatletter
\newcommand{\printdate}{%
  \ifx\@date\@empty\else
    \begin{center}\small \@date\end{center}%
  \fi}
\makeatother
\makeatletter
\let\oldmaketitle\maketitle
\renewcommand{\maketitle}{\oldmaketitle\printdate}
\makeatother
\makeatletter
\renewenvironment{remark}[1][]{%
  \refstepcounter{theorem}\par\medskip
  \noindent\textbf{Remark~\thetheorem
    \ifx\relax#1\relax\else\ (#1)\fi.} \itshape
}{\par\medskip}
\makeatother
\makeatletter
\renewenvironment{remark}[1][]{%
  \refstepcounter{theorem}\par\medskip
  \noindent\textbf{Claim~\thetheorem
    \ifx\relax#1\relax\else\ (#1)\fi.} \itshape
}{\par\medskip}
\makeatother

\usepackage{cite}
\usepackage{lipsum}
\usepackage{amssymb,amsmath}
\usepackage{algorithm}
\usepackage{algpseudocode}
\usepackage{graphicx}
\usepackage{enumitem}
\usepackage{framed}
\usepackage{float}
\usepackage{wrapfig}
\usepackage{multirow}
\usepackage{url}
\usepackage{enumitem}
\usepackage{stackengine}
\usepackage{cancel}
\usepackage{array}
\usepackage{hyperref}
\usepackage{thmtools,thm-restate}
\usepackage{pifont}
\usepackage{mathtools}
\usepackage{xcolor}
\usepackage{bm}
\usepackage{ragged2e}
\usepackage{braket}
\usepackage{tcolorbox}
\usepackage{caption}
\usepackage[normalem]{ulem}
\usepackage{authblk}
\newcommand{\authcount}[1]{} 
\renewcommand*{\contentsname}{Contents}

\usepackage{color}
\usepackage[usenames,dvipsnames,svgnames,table]{xcolor}
\usepackage{hyperref}
\hypersetup{
     colorlinks   = true,
     citecolor    = red
}

\usepackage{tocloft}

\usepackage{booktabs}
\usepackage{colortbl}
\usepackage{pdflscape}
\usepackage{stmaryrd}
\usepackage{tikz}
\usetikzlibrary{calc}
\usetikzlibrary{arrows.meta,decorations.markings}
\usetikzlibrary{positioning,shapes,arrows.meta,decorations.pathreplacing,decorations.markings,fit}
\usetikzlibrary{decorations.pathmorphing, decorations.markings, decorations.shapes, arrows.meta, shapes.geometric}
\usetikzlibrary{patterns}
\usetikzlibrary{shapes.gates.logic.US, shapes}
\tikzset{
  arrowmiddle/.style={
    thick,
    postaction={
      decorate,
      decoration={
        markings,
        mark=at position 0.5 with {\arrow{>}}
      }
    }
  }
}

\tikzset{
  arrowright/.style={
    thick,
    postaction={
      decorate,
      decoration={
        markings,
        mark=at position 0.85 with {\arrow{>}}
      }
    }
  }
}
\tikzset{
  arrowleft/.style={
    thick,
    postaction={
      decorate,
      decoration={
        markings,
        mark=at position 0.15 with {\arrow{>}}
      }
    }
  }
}
\tikzset{
  arrow/.style={
      thick,
      postaction={decorate},
      decoration={markings, mark=at position 0.7 with {\arrow{Latex[length=3mm,width=2mm]}}}
    },
    doublearrow/.style={thick, {Latex[length=3mm,width=2mm]}-{Latex[length=3mm,width=2mm]}}
}

\usepackage[T1]{fontenc}
\usepackage{lmodern}

\allowdisplaybreaks[1]

\urldef{\mailsa}\path|{alfred.hofmann, ursula.barth, ingrid.haas,
frank.holzwarth,| \urldef{\mailsb}\path|anna.kramer, leonie.kunz,
christine.reiss, nicole.sator,|
\urldef{\mailsc}\path|erika.siebert-cole, peter.strasser,
lncs}@springer.com|
\title{Registered Attribute-Based Encryption with Publicly Verifiable Certified Deletion, Everlasting Security, and More}
\titlerunning{Registered Attribute-Based Encryption}
\author{Shayeef Murshid, Ramprasad Sarkar, Mriganka Mandal}
\authorrunning{Murshid et al.}
\institute{Cryptology and Security Research Unit\\ R. C. Bose Centre for Cryptology and Security\\ Indian Statistical Institute\\ 203 B.T. Road, Kolkata 700108, India\\
\email{shayeef.murshid91@gmail.com, \{rpsarkar\_p,mriganka\}@isical.ac.in}
}
\begin{document}
\let\oldaddcontentsline\addcontentsline
\def\addcontentsline#1#2#3{}
\maketitle
\def\addcontentsline#1#2#3{\oldaddcontentsline{#1}{#2}{#3}}
\input{0_abstract}
\clearpage
\hypersetup{linkcolor=blue}
\tableofcontents

\newpage
\input{1_intro}
\input{1.1_Tech-overview}
\input{2_prel}

\input{3_RABE_CD-CED-Framework}

\input{4_Shadow-RABE}

\input{5-6_RABE-CD}

\input{7-8_RABE-CED}
\bibliographystyle{alpha}
\bibliography{reference-url}
\end{document}

%% file: 0_abstract.tex
\begin{abstract}
Certified deletion ensures that encrypted data can be irreversibly deleted, preventing future recovery even if decryption keys are later exposed. Although existing works have achieved certified deletion across various cryptographic primitives, they rely on central authorities, leading to inherent escrow vulnerabilities. This raises the question of whether certified deletion can be achieved in decentralized frameworks such as Registered Attribute-Based Encryption (\textsf{RABE}) that combines fine-grained access control with user-controlled key registration. This paper presents the \emph{first} \textsf{RABE} schemes supporting certified deletion and certified everlasting security. Specifically, we obtain the following: \vspace{0.1cm}

\begin{itemize}\setlength\itemsep{0.2em}
    \item[$-$] \justifying We first design a privately verifiable \textsf{RABE} with Certified Deletion (\textsf{RABE-CD}) scheme by combining our newly proposed shadow registered \textsf{ABE} (\textsf{Shad-RABE}) with one-time symmetric key encryption with certified deletion. 

    \item[$-$] \justifying We then construct a publicly verifiable \textsf{RABE-CD} scheme using \textsf{Shad-RABE}, witness encryption, and one-shot signatures, allowing any party to validate deletion certificates without accessing secret keys.  

    \item[$-$] \justifying We also extend to privately verifiable \textsf{RABE} with Certified Everlasting Deletion (\textsf{RABE-CED}) scheme, integrating quantum-secure \textsf{RABE} with the certified everlasting lemma. Once a certificate is produced, message privacy becomes information-theoretic even against unbounded adversaries.  

    \item[$-$] \justifying We finally realize a publicly verifiable \textsf{RABE-CED} scheme by employing digital signatures for the BB84 states, allowing universal verification while ensuring that deletion irreversibly destroys information relevant to decryption.  
\end{itemize}
 
\end{abstract}

%% file: 1_intro.tex
\section{Introduction}\label{intro}
%
Certified deletion ensures that encrypted data, once deleted, becomes permanently and irretrievably unrecoverable. Achieving this property is a fundamental challenge at the intersection of classical cryptography and quantum information. In purely classical settings, this problem is intractable: digital information can be copied arbitrarily and stored indefinitely, making it impossible to ensure the destruction of every copy. An adversary may retain hidden replicas and later exploit leaked decryption keys or future advances in cryptanalysis to recover the underlying plaintext. Thus, within the classical cryptographic paradigm, certified deletion is widely regarded as unattainable.\vspace{0.2cm}

\noindent \textbf{Quantum Certified Deletion.} Quantum mechanics provides a framework for addressing this limitation through the no-cloning theorem \cite{wootters2009no}, which forbids the duplication of arbitrary quantum states. This intrinsic unclonability sharply distinguishes quantum information from its classical counterpart and enables cryptographic guarantees unattainable in classical settings. Exploiting this property, in 2020, Broadbent \textit{et al.} \cite{broadbent2020quantum} first introduced the quantum encryption with certified deletion protocol, wherein classical messages are encoded into quantum ciphertexts that cannot be copied. In this setting, a recipient of a quantum ciphertext can either decrypt it to recover the message or perform a deletion procedure that produces a verifiable certificate for deletion. Once deleted, the message becomes irretrievable, even if the decryption key is later revealed. This paradigm continues a broader line of quantum-enabled cryptographic primitives that leverage physical unclonability, including quantum key distribution \cite{bennett2014quantum}, quantum money \cite{wiesner1983conjugate}, and quantum secret sharing \cite{hillery1999quantum}.\vspace{0.20cm}
%

\noindent Since its inception, the notion of certified deletion has progressively been extended to a diverse set of cryptographic primitives \cite{hiroka2021quantum,bartusek2023obfuscation,bartusek2023cryptography,bartusek2024secret}. In particular, Bartusek \textit{et al.} \cite{bartusek2023cryptography} showed that for any encryption $\textsf{X} \in$ \{public-key, attribute-based, fully homomorphic, witness, timed release\}, a generic compiler can transform any (post-quantum) $\textsf{X}$-encryption scheme into one supporting certified deletion. In addition, they presented a compiler that upgrades statistically binding commitments into statistically binding commitments with certified everlasting hiding. Collectively, these results underscore the increasing generality of certified deletion and highlight its emerging role as a foundational primitive in hybrid classical-quantum cryptographic constructions.\vspace{0.2cm}

\noindent \textbf{Challenges in Centralization.} Despite these advances, most existing general-purpose frameworks for certified deletion, including those of \cite{hiroka2021quantum,bartusek2023obfuscation,bartusek2023cryptography,bartusek2024secret}, rely crucially on \emph{centralized trust}. In nearly all known constructions, a designated authority is entrusted with generating and distributing cryptographic keys. Such a trust model inherently introduces key-escrow vulnerabilities: the security of the entire system collapses if the authority is compromised, reintroducing the same single point of failure that traditional cryptography has long struggled to eliminate. Eliminating this reliance on centralized trust is therefore not just an optimization in terms of security but a prerequisite for deploying certified deletion in realistic large-scale systems.\vspace{0.2cm}

\noindent A notable exception to this centralized approach is the foundational work of Broadbent \textit{et al.} \cite{broadbent2020quantum}, whose protocol operates without a central authority. However, their construction is inherently restricted to a one-to-one communication model and cannot be readily extended to multi-user or public-key settings without a central authority. This limitation exposes a fundamental gap in the current landscape and motivates the following central question:

\begin{center}
    \textit{Can certified deletion be achieved in advanced encryption schemes without relying on a trusted authority?}
\end{center}

\vspace{0.2cm}
\noindent \textbf{Registered Attribute-Based Encryption.} Attribute-Based Encryption (\textsf{ABE}) \cite{sahai2005fuzzy} extends the traditional public-key encryption paradigm that supports fine-grained access control over encrypted data. In a ciphertext-policy \textsf{ABE} (\textsf{CP-ABE}) scheme, users are issued secret keys associated with attribute sets, while ciphertexts are generated with embedded access policies. A user can decrypt a ciphertext if and only if the attributes bound to their secret key satisfy the ciphertext policy. The dual construction, referred to as key-policy \textsf{ABE} (\textsf{KP-ABE}), interchanges the roles of attributes and access policies. 
However, both the conventional \textsf{ABE} schemes \cite{sahai2005fuzzy,goyal2006attribute,gorbunov2015attribute} are inherently centralized: a single authority issues all user secret keys, thereby introducing the same key-escrow vulnerability. Registration-Based Encryption (\textsf{RBE}) \cite{garg2018registration} overcomes this limitation by decentralizing key generation, users generate their own key pairs and register only their public keys, while a curator aggregates them into a master public key without holding any private information. This model eliminates the problem of central key escrow while preserving global encryptability. Building on this foundation, Registered Attribute-Based Encryption (\textsf{RABE}) \cite{hohenberger2023registered} extends decentralization to attribute-controlled settings: users independently register both their keys and attributes, and encryption policies are evaluated relative to the registered attributes. Recent constructions achieve practicality under bilinear pairings \cite{hohenberger2023registered,zhu2023registered} and post-quantum security under lattice assumptions \cite{champion2025registered,zhu2025black,pal2025registered,wee2025unbounded}.\vspace{0.2cm}

\noindent The convergence of certified deletion with decentralized cryptography naturally raises the question of whether these two paradigms can be integrated without compromising security. Extending certified deletion to \textsf{RABE} would provide an unprecedented combination of properties: decentralized trust (without the key-escrow problem), fine-grained access control, and provable guarantees of irreversible data deletion. This raises the following research question:

\begin{center}
\textit{Can certified deletion be incorporated into the registered framework?}
\end{center}

\noindent However, this integration is technically challenging. In \textsf{RABE}, users control their own secret keys, while the curator manages and updates the helper secret key for users. In contrast, certified deletion must ensure that all the decryption capabilities are irreversibly destroyed. Coordinating these independent components to guarantee complete information destruction is difficult.

\vspace{0.2cm}
\noindent \textbf{Certified Everlasting Security.} An extension of certified deletion security is \emph{certified everlasting security}, which guarantees that once deletion certificates are produced, security becomes information-theoretic: even adversaries with unlimited computational power cannot recover the information after deletion. Recent works have demonstrated certified everlasting security across diverse primitives, including commitments \cite{hiroka2022certified}, zero-knowledge proofs \cite{hiroka2022certified}, garbling schemes \cite{hiroka2024certified}, non-committing encryption \cite{hiroka2024certified}, predicate encryption \cite{hiroka2024certified}, fully-homomorphic encryption \cite{bartusek2023cryptography}, time-release encryption \cite{bartusek2023cryptography}, witness encryption \cite{bartusek2023cryptography}, and functional encryption \cite{hiroka2024certified}. These results suggest that certified everlasting security forms a unifying paradigm for quantum cryptography.

\vspace{0.2cm}
\noindent Taken together, these developments lead to the following fundamental research question:
\begin{center}
    \textit{Can we construct a registered attribute-based encryption protocol with certified everlasting security?}
\end{center}
\vspace{0.2cm}

\subsection{Our Results}

We provide affirmative answers to all the above questions. To the best of our knowledge, we are the first to present the lattice-based key-policy registered attribute-based encryption schemes that support both certified deletion and certified everlasting deletion. We summarize our main findings as follows: \vspace{0.2cm}

\noindent We first introduce a new primitive  (key-policy) shadow-registered attribute-based encryption (\textsf{Shad-RABE}) {with receiver non-committing security}, constructed from an indistinguishability obfuscation ($i\mathcal{O}$), zero-knowledge argument (\textsf{ZKA}), and registered attribute-based encryption (\textsf{RABE}) . The \textsf{Shad-RABE} primitive serves as the foundational building block for subsequent constructions. In addition, we show that lattice-based instantiations of \textsf{Shad-RABE} can be obtained by combining the \textsf{RABE} scheme of Champion \textit{et al.} \cite{champion2025registered}, the \textsf{ZKA} of Bitansky \textit{et al.} \cite{bitansky2020post}, and the $i\mathcal{O}$ framework of Cini \textit{et al.} \cite{cini2025lattice}. A detailed description appears in Section~\ref{section:shad-RABE}.  \vspace{0.2cm}

\noindent We also design the Registered Attribute-Based Encryption with Certified Deletion (\textsf{RABE-CD}) under two different verification models. The first construction, referred to as \textsf{RABE-PriVCD}, achieves privately verifiable deletion by combining our proposed \textsf{Shad-RABE} with One-Time Symmetric Key Encryption with Certified Deletion (\textsf{OTSKE-CD}) \cite{broadbent2020quantum}. This combination requires both the underlying \textsf{RABE} functionality and the simulation features required for certified deletion. We also demonstrate a lattice-based instantiation of \textsf{RABE-PriVCD} by integrating our generic construction with a lattice-based \textsf{Shad-RABE} secure under $\ell$-succinct \textsf{LWE}, plain \textsf{LWE} and equivocal \textsf{LWE} assumptions (\textit{cf.} Section \ref{Inst:Shadow-RABE}) along with an unconditionally secure \textsf{OTSKE-CD} scheme \cite{broadbent2020quantum}. The resulting protocol achieves certified deletion security, which is discussed in detail in Section \ref{subsection-RABE-PriVCD}. The second construction, referred to as \textsf{RABE-PubVCD}, which integrates our proposed \textsf{Shad-RABE}, witness encryption, and one-shot signatures, allowing deletion certificates to be publicly validated by any party without access to secret keys. We instantiate \textsf{RABE-PubVCD} over lattices by combining three central building blocks: the lattice-based \textsf{Shad-RABE} protocol secure under $\ell$-succinct \textsf{LWE}, plain \textsf{LWE} and equivocal \textsf{LWE} (\textit{cf.} Section \ref{Inst:Shadow-RABE}); extractable witness encryption instantiated from the \textsf{LWE} assumption and its variants \cite{vaikuntanathan2022witness,tsabary2022candidate}; and one-shot signatures \cite{amos2020one,shmueli2025one}. The resulting scheme presents the lattice-based publicly verifiable \textsf{RABE-CD} design, as shown in Section~\ref{subsection-RABE-PubVCD}.\vspace{0.2cm}

\noindent We then further extend our framework to \textsf{RABE} with Certified Everlasting Deletion (\textsf{RABE-CED}), which guarantees message privacy even against computationally unbounded adversaries once a valid deletion certificate is issued. 
In the privately verifiable setting, our construction combines quantum-secure \textsf{RABE} with the certified everlasting lemma \cite{bartusek2023cryptography}, employing quantum states prepared in random bases together with classical encryption that conceals the basis information. Since generating a deletion certificate requires measurement of these quantum states, the necessary decryption information is irreversibly destroyed, ensuring the indistinguishability of ciphertexts even against unbounded adversaries. The detailed description of this construction is presented in Section~\ref{sec:RABE-PrivCED}. In the publicly verifiable setting, we utilize one-time unforgeable signatures for BB84 states, which allow universal verification of deletion certificates. The certified everlasting lemma guarantees that BB84 states $\lvert x \rangle_\theta$ cannot be reconstructed without the bases $\theta$, and any incorrect measurement permanently erases the information about $x$. This construction is also analyzed in depth in Section~\ref{sec:RABE-PubVCED}.

\subsection{Related Work}\label{subsection:related-works}

\noindent\textbf{Threshold and Multi-Authority Approaches.}  
The key escrow problem in identity-based encryption \cite{boneh2001identity,paterson2008security} and attribute-based encryption \cite{sahai2005fuzzy,goyal2006attribute,gorbunov2015attribute} motivated new approaches to reduce centralized trust. An early line of work \cite{boneh2001identity,chen2002applications,paterson2008security,kate2010distributed} applied threshold cryptography, where the master secret key was divided among several authorities, and the generation of keys required their collective participation. Although effective in reducing single point of failure, these schemes remained vulnerable to collusion between authorities. Another direction was multi-authority \textsf{ABE} \cite{chase2007multi,lin2008secure,muller2008distributed,chase2009improving,lewko2011decentralizing,rouselakis2015efficient,datta2021decentralized-a,datta2021decentralized-b,waters2022multi}, where separate attribute domains were controlled by different entities. This distribution limited the power of individual authorities, but did not fundamentally eliminate the systemic risk posed by the compromise of trusted authorities.

\vspace{0.2cm}  
\noindent\textbf{Registration-Based Encryption.}  
In 2018, a paradigm shift occurred with the introduction of Registration-Based Encryption (\textsf{RBE}) \cite{garg2018registration}, which fundamentally restructured the trust model rather than redistributing existing vulnerabilities. This line of work gave rise to Registered Attribute-Based Encryption (\textsf{RABE}) \cite{hohenberger2023registered}, where users generate their own secret/public key pair and register the corresponding public keys along with attribute information through a curator. The curator aggregates these public keys into a succinct master public key, while users employ helper secret keys to enable decryption. Early \textsf{RABE} schemes \cite{hohenberger2023registered,zhu2023registered} relied on bilinear pairings over composite order groups and supported policies that can be represented through linear secret sharing, laying the groundwork for decentralized attribute-based access control.  

\vspace{0.2cm}  
\noindent\textbf{Revocation in Registered \textsf{ABE}.}  
Practical deployment prompted the study of revocation mechanisms within \textsf{RABE}. The first formal analysis introduced Deletable \textsf{RABE} (\textsf{DRABE}) and Directly Revocable \textsf{RABE} (\textsf{RRABE}) \cite{asano2025key}, addressing temporal access control in dynamic systems. Later work \cite{li2025revocable} expanded this to fine-grained revocation, supporting both file-specific access removal and permanent user deregistration. These advances enhanced the flexibility of the administrator while preserving the decentralized trust foundation of registered frameworks.  

\vspace{0.2cm}  
\noindent\textbf{Post-Quantum Secure \textsf{RABE}.}  
Significant theoretical progress has extended \textsf{RABE} into the post-quantum setting. A scheme~\cite{zhu2025black} supporting unbounded users with a transparent setup was built from the private-coin evasive \textsf{LWE} assumption. It avoids random oracles but relies on complex assumptions that are now under scrutiny~\cite{vaikuntanathan2022witness,branco2024pseudorandom,brzuska2024evasive}. In 2025, Champion \textit{et al.} \cite{champion2025registered} proposed a key-policy \textsf{RABE} for bounded-depth circuits under the $\ell$-succinct \textsf{LWE} assumption in the random oracle model, trading stronger expressiveness for reliance on heuristic models. Recently, Pal \textit{et al.}~\cite{pal2025registered} achieved post-quantum secure ciphertext-policy and key-policy \textsf{RABE} schemes and registered predicate encryption, supporting unbounded circuit depth. In parallel, multi-authority extensions \cite{lu2025multi} introduced independent curators for separate attribute domains, enabling policies across multiple authorities while maintaining the registration-based structure.  

\vspace{0.2cm}  
\noindent\textbf{Quantum Encryption with Certified Deletion.}  
In parallel, quantum cryptography introduced mechanisms for verifiable data deletion. Quantum encryption with certified deletion \cite{broadbent2020quantum} exploits the no-cloning property of quantum states, ensuring that deleted information cannot be recovered, even by unbounded adversaries. This principle has been extended to public-key encryption, \textsf{ABE}, fully homomorphic encryption, and witness encryption \cite{hiroka2021quantum,bartusek2023obfuscation,bartusek2023cryptography,bartusek2024secret}, highlighting its wide applicability.  

\vspace{0.2cm}  
\noindent\textbf{Certified Everlasting Security.}  
The concept of certified deletion has further evolved into certified everlasting security, which ensures that once deletion occurs, privacy persists indefinitely. Examples include certified everlasting commitments and zero-knowledge proofs for \textsf{QMA} \cite{hiroka2022certified}, as well as frameworks for all-or-nothing encryption based on BB84 states \cite{bartusek2023cryptography}. These results extend naturally to functional encryption, garbled circuits, and non-committing encryption \cite{hiroka2024certified}, demonstrating the breadth of everlasting security guarantees.  

\vspace{0.2cm}  
\noindent While many certified deletion systems are limited to private verifiability, where verification of deletion is restricted to the designated parties, practical applications often require publicly verifiable guarantees. Addressing this need, recent works have introduced publicly verifiable certified deletion for public-key encryption, \textsf{ABE}, and fully homomorphic encryption \cite{poremba2022quantum,kitagawa2023publicly}, thereby enabling universal auditability and supporting regulatory compliance in decentralized settings.

%% file: 1.1_Tech-overview.tex
\section{Technical Overview}\label{tech-overview}
This section presents a high-level technical overview of our constructions for registered attribute-based encryption with certified deletion and certified everlasting deletion, considering both privately and publicly verifiable frameworks, depicted in Fig. \ref{tech-overview-fig}. \vspace{0.2cm}

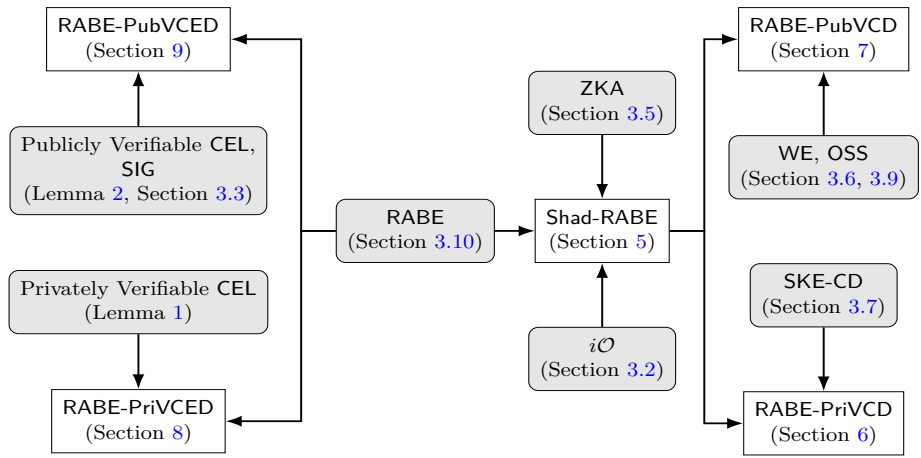
\begin{figure}
\centering
\resizebox{1\linewidth}{0.85\height}{%
\begin{tikzpicture}[
  node/.style={draw, minimum size=8mm, rectangle},
  arr/.style={-Latex, thick},
  special/.style={draw, rounded corners, fill=gray!20, minimum size=8mm, rectangle},
  box/.style={draw, rounded corners, inner sep=10pt}
]

\node[node] (C) at (0.5,-1) {\begin{tabular}{c}\textsf{Shad-RABE}\\ (Section~\ref{section:shad-RABE})\end{tabular}};
\node[special] (RABE) at (-2.2,-1) {\begin{tabular}{c}\textsf{RABE}\\ (Section~\ref{model:RABE})\end{tabular}};

\node[special] (Top) [above=of C] {{\begin{tabular}{c}\textsf{ZKA}\\ (Section \ref{def:zka})\end{tabular}}};
\node[special] (Bottom) [below=of C] {{\begin{tabular}{c}$i\mathcal{O}$\\ (Section~\ref{def:iO-framework})\end{tabular}}};

\draw[arr] (Top.south) -- (C.north);   
\draw[arr] (Bottom.north) -- (C.south); 

\node[node] (L1) at (-6.2,2) {{\begin{tabular}{c}\textsf{RABE-PubVCED}\\ (Section~\ref{sec:RABE-PubVCED})\end{tabular}}};
\node[special] (L2) [below= 0.86 cm of L1] {{\begin{tabular}{c} Publicly Verifiable \textsf{CEL},\\ \textsf{SIG}\\(Lemma~\ref{lem:pv-certified-everlasting}, Section~\ref{def:sig-framework})\end{tabular}}};
\node[special] (L3) [below=0.86 cm of L2] {{\begin{tabular}{c} Privately Verifiable \textsf{CEL}\\ (Lemma~\ref{lem:certified-everlasting-bk})\end{tabular}}};
\node[node] (L4) [below=0.86 cm of L3] {{\begin{tabular}{c}\textsf{RABE-PriVCED}\\ (Section~\ref{sec:RABE-PrivCED})\end{tabular}}};

\draw[arr] (L2) -- (L1); 
\draw[arr] (L3) -- (L4); 
\draw[arr] (RABE.east) -- (C.west);   

\draw[arr] (RABE.west) -- ++(-0.5,0) |- (L1.east);
\draw[arr] (RABE.west) -- ++(-0.5,0) |- (L4.east);

\node[node] (R1) at (3.7,2) {{\begin{tabular}{c}\textsf{RABE-PubVCD}\\ (Section~\ref{subsection-RABE-PubVCD})\end{tabular}}};
\node[special] (R2) [below=of R1] {{\begin{tabular}{c}\textsf{WE}, \textsf{OSS}\\ (Section~\ref{def:we-framework}, \ref{def:oss-framework})\end{tabular}}};
\node[special] (R3) [below=of R2] {{\begin{tabular}{c}\textsf{SKE-CD}\\ (Section~\ref{subsection-enc-cd})\end{tabular}}};
\node[node] (R4) [below=of R3] {{\begin{tabular}{c}\textsf{RABE-PriVCD}\\ (Section~\ref{subsection-RABE-PriVCD})\end{tabular}}};

\draw[arr] (R2) -- (R1); 
\draw[arr] (R3) -- (R4); 

\draw[arr] (C.east) -- ++(0.5,0) |- (R1.west);
\draw[arr] (C.east) -- ++(0.5,0) |- (R4.west);

\end{tikzpicture}
}
\vspace{-3mm}

    \caption{High-level Overview of Our Constructions}
    \label{tech-overview-fig}
\end{figure}

\noindent\textbf{Notations.} We begin by introducing some basic notation that will be used throughout this paper. The notation $x \leftarrow X$ denotes that $x$ is selected uniformly at random from a finite set $X$. If $A$ is an algorithm (which may be classical, probabilistic, or quantum), then $y \leftarrow A(x)$ represents the output of $A$ on input $x$, and we write $y \leftarrow A(x; r)$ when $A$ uses explicit randomness $r$. $\lambda$ and $\tau$ denote a security parameter and a policy-family parameter respectively. For a distribution $D$, we write $x \leftarrow D$ to denote sampling according to $D$. The assignment $y:= z$ indicates that the variable $y$ is set or defined to be equal to $z$. For a positive integer $n$, we denote the set $\{1, 2, \ldots, n\}$ by $[n] := \{1, \ldots, n\}$, and $[N]_p$ represents the set of all $p$-element subsets of $\{1, 2, \ldots, N\}$. The symbol $\lambda$ denotes the security parameter. For classical strings $x$ and $y$, their concatenation is written as $x \| y$. For a bit string $s \in \{0,1\}^n$, both $s_i$ and $s[i]$ refer to its $i$-th bit. The acronym \textsf{QPT} stands for quantum polynomial-time algorithms, while \textsf{PPT} refers to probabilistic polynomial-time algorithms using classical computation. A function $f: \mathbb{N} \to \mathbb{R}$ is said to be negligible if, for every constant $c > 0$, there exists $\lambda_0 \in \mathbb{N}$ such that $f(\lambda) < \lambda^{-c}$ for all $\lambda > \lambda_0$. We use the notation $f(\lambda) \leq \mathsf{negl}(\lambda)$ to express that $f$ is a negligible function. For a bit string $x$, we write $x_j$ to mean the $j$-th bit of $x$. For bit strings $x, \theta \in \{0, 1\}^{\ell}$, we write $\ket{x}_{\theta}$ to mean the BB84 state $\bigotimes_{j \in [\ell]} H^{\theta_j} \ket{x_j}$ where $H$ is the Hadamard operator. The trace distance, denoted by $\mathsf{TD}(\cdot,\cdot)$, between two states $\rho$ and $\sigma$ is defined as
\(
\mathsf{TD}(\rho,\sigma) = \|\rho - \sigma\|_{\mathrm{tr}},
\) where \(
\|A\|_{\mathrm{tr}} \coloneqq \operatorname{Tr}\!\big(\sqrt{A^\dagger A}\big).
\) We use the symbol $\stackrel{c}{\approx}$ to denote computational indistinguishability between distributions.

\vspace{-0.2cm}

\subsection{Towards Registered \textsf{ABE-CD} from \textsf{ABE-CD}}
In a registered attribute-based encryption scheme, users independently generate key pairs $(\mathsf{pk},\mathsf{sk})$ using local randomness, ensuring that no trusted party is ever exposed to their secret information. The users then submit only their public keys, together with their policies, to a curator. The curator performs an administrative role: aggregating submitted keys into an evolving master public key $\mathsf{mpk}$ and producing helper secret key $\mathsf{hsk}$ for registered users. Crucially, the curator never learns the secret keys of users, preserving security even against a malicious or compromised curator. \vspace{0.2cm}

\noindent Verifiable deletion for the ciphertext further augments privacy: a recipient can produce a classical certificate that proves the irreversible destruction of the quantum ciphertext. {A valid certificate of deletion guarantees that no adversary can recover the message, even in possession of the secret keys}. This verifiable destruction mechanism was introduced in attribute-based encryption with certified deletion (\textsf{ABE-CD}) of Hiroka \textit{et al.}~\cite{hiroka2021quantum}. However, integrating such a feature into a registered \textsf{ABE} framework is highly non-trivial, as the curator does not hold any secret information, such as the master secret key \textsf{msk}, in the registered \textsf{ABE} system. Before developing our construction, we recall the definition of \textsf{ABE-CD} \cite{hiroka2021quantum}: 

\begin{itemize}\setlength\itemsep{0.4em}
    \item[$-$] \textbf{\textsf{Setup}}: Given the security parameter $1^\lambda,$ this algorithm outputs a public key $\mathsf{pk}$ and a master secret key $\mathsf{msk}$.
       
    \item[$-$] \textbf{\textsf{KeyGen}}: This algorithm takes as input $\mathsf{msk}$ and a policy $P$ to return a secret key $\mathsf{sk}_P$. 
    
    \item[$-$] \textbf{\textsf{Encryption}}: Taking as input $\mathsf{pk}$, an attribute $X$, and a message $\mu$, this algorithm outputs a verification key $\mathsf{vk}$ and a ciphertext $\mathsf{CT}$.

    \item[$-$] \textbf{\textsf{Decryption}}: On input the secret key $\mathsf{sk}_P$ and a ciphertext $\mathsf{CT}$, the decryption algorithm outputs a message $\mu'$ or $\bot$.
    
    \item[$-$] \textbf{\textsf{Delete}}: The delete algorithm takes as input a ciphertext $\mathsf{CT}$ and outputs a certification $\mathsf{cert}$.
 
    \item[$-$] \textbf{\textsf{Verify}}:  Given the verification key $\mathsf{vk}$ and the certificate $\mathsf{cert}$,  this algorithm outputs $\top$ or $\bot$.
\end{itemize}

\noindent \textbf{Starting Point.} As a warm-up to our actual scheme, in the following, we first describe the attribute-based encryption with certified deletion (\textsf{ABE-CD}) scheme proposed by Hiroka \textit{et al.} \cite{hiroka2021quantum}. Their framework is built upon receiver non-committing attribute-based encryption (\textsf{RNCABE}), which serves as a building block for achieving \textsf{ABE-CD} scheme. The main challenge here is to reconcile the fine-grained access control provided by attribute-based encryption with the irreversible deletion guarantees required by certified deletion. To address this, Hiroka \textit{et al.} combine \textsf{RNCABE} with a one-time secret key encryption protocol that supports certified deletion, thus embedding deletion capabilities directly into the attribute-based paradigm. \vspace{0.2cm}

\noindent Understanding their construction requires first grasping the underlying framework of \textsf{RNCABE}, which provides the essential simulation-based properties needed for security in the presence of deletion functionality. Therefore, we briefly review the formal syntax of \textsf{RNCABE}, defined over a message space $\mathcal{M}$, an attribute space $\mathcal{X}$, and a policy space $\mathcal{P}$.

\begin{itemize}\setlength\itemsep{0.5em}
    \item[$-$] \textbf{\textsf{Setup}}: Given the security parameter $1^\lambda$, the setup algorithm outputs a master public key $\mathsf{mpk}$ and a master secret key $\mathsf{msk}$.
    
    \item[$-$] \textbf{\textsf{KeyGen}}: On input of the master secret key $\mathsf{msk}$ and a policy $P$, the key generation algorithm generates the corresponding secret key $\mathsf{sk}_P$.

    \item[$-$] \textbf{\textsf{Encryption}}: Using the master public key $\mathsf{mpk}$, an attribute $X$, and a message $m$, the encryption algorithm produces a ciphertext $\mathsf{CT}$.

    \item[$-$] \textbf{\textsf{Decryption}}: The decryption algorithm recovers the message $m'$ from $\mathsf{CT}$ using a secret key $\mathsf{sk}$, or outputs $\perp$ in case of failure.
    
    \item[$-$] \textbf{\textsf{FakeSetup}}: The fake setup algorithm produces a master public key $\mathsf{mpk}$ together with auxiliary information $\mathsf{aux}$, serving as an alternative to the real setup.
    
    \item[$-$] \textbf{\textsf{FakeCT}}: Given the master public key $\mathsf{mpk}$, an auxiliary information $\mathsf{aux}$, and an attribute $X$, the fake encryption algorithm outputs a fake ciphertext $\widetilde{\mathsf{CT}}$.
    
    \item[$-$] \textbf{\textsf{FakeSK}}: The fake key generation algorithm generates a fake secret key $\widetilde{\mathsf{sk}}$ corresponding to policy $P$.
    
    \item[$-$] \textbf{\textsf{Reveal}}: On input the master public key $\mathsf{mpk}$, an auxiliary information $\mathsf{aux}$, a fake ciphertext $\widetilde{\mathsf{CT}}$, and a message $m$, the reveal algorithm outputs a fake master secret key $\widetilde{\mathsf{msk}}$.
\end{itemize}

\noindent\textbf{First Modification: Towards Registered \textsf{ABE} with Certified Deletion.} As an initial step towards our goal, we focus on constructing a registered attribute-based encryption scheme with certified deletion, without yet addressing its full-fledged security guarantees. In a registered attribute-based setting, users must generate their own public and secret key pairs independently, without relying on any trusted central authority. Consequently, to extend the existing \textsf{ABE-CD} framework into a registered setting, we need to introduce new algorithms and modify existing ones from the protocol of Hiroka \textit{et al.} \cite{hiroka2021quantum}. In particular, the following changes are required:

\begin{itemize}\setlength\itemsep{0.5em}
    \item[$-$] \textbf{\textsf{Setup}}: Taking as input the length of the security parameter $\lambda$ and the size of the attribute universe $\mathcal{U}$, the setup algorithm outputs a common reference string $\mathsf{crs}$.

    \item[$-$] \textbf{\textsf{KeyGen}}: Given the common reference string $\mathsf{crs}$, an auxiliary state $\mathsf{aux}$, and a policy $P$, each user independently runs the key generation algorithm to derive its own key pair $(\mathsf{pk}, \mathsf{sk})$ before registration.
    
    \item[$-$] \textbf{\textsf{RegPK}}: On input of the common reference string \textsf{crs}, an auxiliary state \textsf{aux}, a user's public key $\textsf{pk}$, and a policy $P$, the registration algorithm updates the master public key \textsf{mpk}, and the auxiliary state \textsf{aux} to incorporate the registration of the new user.
        
    \item[$-$] \textbf{\textsf{Update}}: The update algorithm, given the common reference string \textsf{crs}, auxiliary state \textsf{aux}, and a registered public key \textsf{pk}, outputs a helper secret key $\mathsf{hsk}$ that allows the associated user to decrypt ciphertexts for which they are authorized.
\end{itemize}

\noindent In the registered setting, the functionality of the trusted authority is replaced by a transparent public entity, referred to as the curator, which executes the \textsf{RegPK} and \textsf{Update} algorithms. Before registering a public key \textsf{pk}, the curator can invoke a validation procedure to ensure its correctness, as the key may have been maliciously generated.\vspace{0.2cm}

\noindent \textbf{Second Modification: Architecture of \textsf{Shad-RABE}.}
{Building on the model of Hiroka \textit{et al.}~\cite{hiroka2021quantum}, we next aim to extend the notion of receiver non-committing attribute-based encryption (\textsf{RNCABE}) to a decentralized setting.} This extension introduces several conceptual challenges, primarily due to the need for specialized simulation algorithms that preserve the registration framework. Conceptually, this can be viewed as adapting \textsf{RNCABE} to the registered framework.\vspace{0.2cm}

\noindent In the \textsf{RNCABE} framework of Hiroka \textit{et al.} \cite{hiroka2021quantum}, the trusted authority holds a master secret key $\mathsf{msk}$, which is used to generate and distribute user secret keys. In contrast, within the registered setting, the curator is a transparent and publicly verifiable entity, and hence cannot hold or access any secret information. This fundamental restriction makes a direct adaptation of \textsf{RNCABE} infeasible.\vspace{0.2cm}

\noindent {Alternatively, one may view the incorporation of registration mechanisms into \textsf{RNCABE} as yielding a registered \textsf{ABE} scheme with non-committing security}. Since both key generation and the distribution of helper secret keys in the \textsf{Update} procedure are already managed by the curator, embedding simulation algorithms within the registered \textsf{ABE} framework provides a natural pathway toward our goal. {In particular, the simulation algorithms are used only within the security experiment, executed by the challenger, and have no bearing on the scheme’s correctness.} We refer to this extended framework as \emph{Shadow Registered Attribute-Based Encryption} (\textsf{Shad-RABE}), which serves as the core building block for our subsequent constructions. Specifically, \textsf{Shad-RABE} augments the classical \textsf{RABE} algorithms with five additional simulation procedures: \vspace{-0.2cm}

\begin{itemize}\setlength\itemsep{0.5em}
    \item[$-$] \textbf{\textsf{SimKeyGen}}: Given the common reference string $\mathsf{crs}$, an auxiliary state $\mathsf{aux}$, an access policy $P$, and an ancillary dictionary $\mathsf{B}$ (possibly empty) indexed by public keys, the simulator runs the key-generation procedure to produce a simulated public key $\widetilde{\mathsf{pk}}$ along with an updated dictionary $\mathsf{B}$.

    \item[$-$] \textbf{\textsf{SimRegPK}}: The simulated registration procedure takes as input the common reference string $\mathsf{crs}$, an auxiliary state $\mathsf{aux}$, a public key $\mathsf{pk}$, an access policy $P$, and the ancillary dictionary $\mathsf{B}$, and outputs a simulated master public key $\widetilde{\mathsf{mpk}}$ together with an updated auxiliary state $\widetilde{\mathsf{aux}}$.

    \item[$-$] \textbf{\textsf{SimCorrupt}}: The simulated corrupt algorithm, on input the common reference string $\mathsf{crs}$, a public key $\mathsf{pk}$, and the ancillary dictionary $\mathsf{B}$, returns a simulated secret key $\widetilde{\mathsf{sk}}$ corresponding to $\mathsf{pk}$.

    \item[$-$] \textbf{\textsf{SimCT}}: Given the simulated master public key $\widetilde{\mathsf{mpk}}$, the dictionary $\mathsf{B}$, {and an attribute $X$, the ciphertext simulation algorithm produces a simulated ciphertext} $\widetilde{\mathsf{ct}}$.

    \item[$-$] \textbf{\textsf{Reveal}}: Given a public key $\widetilde{\mathsf{pk}},$ the dictionary $\mathsf{B}$, a simulated ciphertext $\widetilde{\mathsf{ct}}$, and a target message $\mu$, the reveal algorithm outputs a simulated secret key $\widetilde{\mathsf{sk}}$ that decrypts $\widetilde{\mathsf{ct}}$ to $\mu$.

\end{itemize}

\noindent\textbf{Security for \textsf{Shad-RABE}.} The security foundation of \textsf{Shad-RABE} builds upon a receiver non-committing security framework that combines indistinguishability obfuscation with zero-knowledge arguments possessing statistical soundness and computational zero-knowledge properties. The security analysis employs a sequence of hybrid transformations that systematically transition from real-world operations to ideal-world simulations, ensuring that no polynomial-time adversary can distinguish between these worlds with non-negligible probability.\vspace{0.2cm}

\noindent 
The proof methodology begins with the original \textsf{Shad-RABE} security experiment, where all cryptographic operations are performed according to their original algorithms. The analysis then progressively introduces simulation components through computationally indistinguishable transformations. The first transformation replaces left-or-right decryption circuits embedded within obfuscated programs with left-only variants, thereby eliminating the scheme's dependence on the random selector string $z$, which determines the bit representation to use for each message position. This change relies on the security properties of indistinguishability obfuscation, which ensures that functionally equivalent circuits remain computationally indistinguishable when obfuscated.\vspace{0.2cm}

\noindent Next, to integrate zero-knowledge arguments within the registered framework, we consider constructions that rely on a common reference string \textsf{crs}. In such settings, each new user registration would necessitate an update of the \textsf{crs}. Our initial attempt was to realize \textsf{Shad-RABE} using existing \textsf{NIZK} schemes with updatable common reference strings \cite{groth2018updatable,abdolmaleki2020lift}. However, these schemes either fail to achieve statistical soundness or depend on pairing-based assumptions that are potentially insecure against quantum polynomial-time (\textsf{QPT}) adversaries.\vspace{0.2cm}

\noindent To overcome these limitations, we employ a zero-knowledge argument system that does not require any common reference string, while still guaranteeing soundness and zero-knowledge properties against \textsf{QPT} adversaries. The analysis transitions to simulated \textsf{ZKA} components, where challenge proofs are generated using simulation algorithms: $\widetilde{\pi} \leftarrow \langle{Sim}\rangle (d^*,y)$, where $d^*, y$ are the statement and an element of $\{0,1\}^*$, respectively. This transformation leverages the computational zero-knowledge property of the \textsf{ZKA} protocol to ensure that the simulated proofs are indistinguishable from the real ones. This makes the proof components independent of the actual challenge message. \vspace{0.2cm}

\noindent Another hybrid transformation introduces asymmetric encryption patterns that break the symmetry in how different message bits are encrypted. For each bit position $i$, the construction encrypts different values depending on the selector bit: $\mathsf{CT}^*_{i,z[i]} \leftarrow \mathsf{RABE}.\mathsf{Encrypt}(\widetilde{\mathsf{mpk}}_{i,z[i]}, X^*, \mu[i])$ encrypts the actual message bit, while $\mathsf{CT}^*_{i,1-z[i]} \leftarrow \mathsf{RABE}.\mathsf{Encrypt}(\widetilde{\mathsf{mpk}}_{i,1-z[i]}, X^*, 1-\mu[i])$ encrypts its complement. This asymmetric approach is validated through a position-by-position hybrid argument that relies on the semantic security of the underlying \textsf{RABE} scheme.\vspace{0.2cm}

\noindent Finally, the analysis achieves complete independence from the target message $\mu$ employing a new random string $z^* \leftarrow \{0,1\}^{\ell_m}$ and encrypting fixed values: $\mathsf{CT}^*_{i,z^*[i]} \leftarrow$ $\mathsf{RABE}.\mathsf{Encrypt}(\widetilde{\mathsf{mpk}}_{i,z^*[i]}, X^*, 0)$ and $\mathsf{CT}^*_{i,1-z^*[i]} \leftarrow$ $\mathsf{RABE}.\mathsf{Encrypt}$ $(\widetilde{\mathsf{mpk}}_{i,1-z^*[i]},$ $ X^*, 1)$. This systematic elimination of dependencies, from the random selector string $z$, through real \textsf{ZKA} proofs, to the target message $\mu$ itself, combined with the computational indistinguishability of consecutive transformations, establishes that no polynomial-time adversary can distinguish between real and ideal executions with non-negligible probability. The resulting \textsf{Shad-RABE} primitive provides the simulation capabilities necessary for constructing secure \textsf{RABE} schemes with certified deletion, enabling security proofs that remain valid even when adversaries can adaptively corrupt users and request deletion certificates. \vspace{-0.2cm}

\subsection{{Design of} \textsf{RABE-PriVCD} under \textsf{LWE} Assumptions}
We propose a construction of Registered Attribute-Based Encryption with Privately Verifiable Certified Deletion (\textsf{RABE-PriVCD}), designed to jointly realize decentralized access control and quantum-secure data deletion. The framework follows a hybrid encryption strategy, which distinctly separates the role of registered \textsf{ABE} from the quantum mechanisms required certified deletion.
\vspace{0.15cm}

\noindent At the core of our design lies the observation that the two functionalities, registered \textsf{ABE} and certified deletion, can be combined in a modular manner within a hybrid structure. To this end, we integrate two specialized primitives: 
\begin{itemize}\setlength\itemsep{0.5em}
    \item \textbf{\textsf{Shad-RABE}}, which manages access control, decentralized key registration, and fine-grained policy enforcement.

    \item \textbf{\textsf{SKE-CD}}, symmetric key encryption with certified deletion, which leverages quantum principles to guarantee that deleted data remains irretrievable, even against unbounded adversaries.
\end{itemize}

\noindent\textbf{Two-Layer Encryption.} The construction operates through a two-layer encryption approach. During encryption, a fresh symmetric key is generated and encrypted using \textsf{Shad-RABE} under the specified attribute set, while the actual message is encrypted using \textsf{SKE-CD} with the symmetric key. This design ensures that attributes are enforced through the registered \textsf{ABE} mechanism, while the certified deletion security is derived from the quantum properties of the symmetric scheme. The core encryption algorithm \textsf{Encrypt}$(\mathsf{mpk}, X, \mu)$ demonstrates the hybrid approach:\vspace{-0.2cm}

\begin{itemize}\setlength\itemsep{0.5em}
  \item[$-$] Generate symmetric key $\mathsf{ske.sk} \leftarrow \mathsf{SKE\text{-}CD}.\mathsf{KeyGen}(1^\lambda)$
  \item[$-$] Compute \textsf{Shad-RABE} ciphertext $\mathsf{srabe.ct} \leftarrow \mathsf{Shad\text{-}RABE}.\mathsf{Encrypt}(\mathsf{mpk},$ $ X,$ $\mathsf{ske.sk})$
  \item[$-$] Compute symmetric ciphertext $\mathsf{ske.ct} \leftarrow \mathsf{SKE\text{-}CD}.\mathsf{Enc}(\mathsf{ske.sk}, \mu)$
  \item[$-$] Output ciphertext $\mathsf{ct} := (\mathsf{srabe.ct}, \mathsf{ske.ct})$ and verification key $\mathsf{vk} := \mathsf{ske.sk}$
\end{itemize}

\noindent The security guarantees derive from the interaction between both components:
\begin{itemize}\setlength\itemsep{0.5em}
    \item \textbf{Access Control Security:} \textsf{Shad-RABE} ensures that only users with policies $P$ satisfying $P(X)=1$ can recover the symmetric key $k$. This holds even for adaptive corruptions, dynamic user registration, and key update operations.

    \item \textbf{Certified Deletion Security:} {Once a valid deletion certificate is produced, the \textsf{SKE-CD} scheme guarantees that the message $\mu$ can no longer be recovered by any quantum polynomial-time (\textsf{QPT}) adversary.}

    \item \textbf{Composition Security:} The hybrid approach preserves both security properties. Unauthorized users cannot access the information because they cannot recover $k$. After deletion, even authorized users cannot recover the data due to the information-theoretic deletion guarantees.
\end{itemize}

\noindent\textbf{Post-Quantum Realization.} Our construction achieves post-quantum security through a modular framework. Specifically, the \textsf{Shad-RABE} component is realized from lattice-based building blocks, namely: a registered \textsf{ABE} scheme under the $\ell$-succinct \textsf{LWE} assumption \cite{champion2025registered}, a zero-knowledge argument based on plain \textsf{LWE} \cite{bitansky2020post}, and an indistinguishability obfuscation scheme relying on equivocal \textsf{LWE} \cite{cini2025lattice}. These primitives jointly ensure resistance against quantum adversaries and enable simulation techniques required for certified deletion proofs. Complementing this, the \textsf{SKE-CD} component employs an unconditionally secure one-time symmetric encryption with certified deletion \cite{broadbent2020quantum}, thus providing information-theoretic guarantees beyond standard computational assumptions. This separation allows independent analysis of \textsf{ABE} functionality and deletion mechanisms, while their integration delivers decentralized key management, fine-grained access control, and verifiable data deletion with private verification. \vspace{-0.2cm}

\subsection{{Design of \textsf{RABE-PubVCD} under \textsf{LWE} Assumptions}}
In the \textsf{RABE-PriVCD} framework, the verification key coincides with the symmetric key $\mathsf{ske.sk}$ generated by the \textsf{SKE-CD} scheme. Since $\mathsf{ske.sk}$ is used directly in the two-layer encryption of message $\mu$, verification remains inherently private; public verification cannot be achieved in this setting. Consequently, to construct \textsf{RABE-PubVCD}, we cannot rely on \textsf{SKE-CD} protocol. \vspace{0.2cm}

\noindent\textbf{Core Idea.}  
To overcome this limitation, we employ the notion of {one-shot signatures} \cite{amos2020one}. A one-shot signature scheme allows the generation of a classical public key $\mathsf{oss.pk}$ along with a quantum secret key $\mathsf{oss.sk}$. The secret key can be used to sign message $0$ (producing $\sigma_0$) or message $1$ (producing $\sigma_1$), but never both. Importantly, signatures can be verified publicly, making them well-suited for publicly verifiable deletion.  \vspace{0.2cm}

\noindent Our construction integrates one-shot signatures with extractable witness encryption. The encryption of a message $\mu$ proceeds by encoding $\mu$ within a witness encryption instance, where the statement corresponds to the existence of a valid signature on $0$. The deletion certificate, on the other hand, is realized as a one-shot signature on $1$. Once such a certificate is issued, the one-shot property ensures that no valid signature on $0$ can ever be generated, thereby guaranteeing certified deletion. Public verifiability follows immediately, as any third party can check signatures. In order to prevent an adversary from decrypting the ciphertext before issuing the deletion certificate, we add an additional layer of encryption, for which we implement \textsf{Shad-RABE}. \vspace{0.2cm}

\noindent\textbf{Protocol Execution.}  
The protocol unfolds as follows: the sender first generates one-shot signature parameters \textsf{oss.crs} and shares them with the receiver. The receiver uses these to create a key pair $(\mathsf{oss.pk}, \mathsf{oss.sk})$, returns the public key to the sender, and keeps the signing key secret. With $\mathsf{oss.pk}$, the sender defines a witness statement that asserts the existence of a valid signature on $0$, encrypts $\mu$ under this statement by witness encryption, and then encapsulates the resulting ciphertext using \textsf{Shad-RABE} under the chosen attribute set $X$. The final ciphertext is transmitted to the receiver. At the end of this process, the sender publishes the public verification key $\mathsf{vk} = (\mathsf{oss.crs}, \mathsf{oss.pk})$, while the receiver obtains the ciphertext $\mathsf{ct} = (\mathsf{srabe.ct}, \mathsf{oss.sk})$ containing both the encrypted data and the deletion capability. \vspace{0.2cm}

\noindent\textbf{Security Guarantees.}  
The one-shot property guarantees that no adversary can produce both $\sigma_0$ and $\sigma_1$. Thus, once a deletion certificate $\sigma_1$ is issued, the witness encryption becomes undecryptable, certifying irreversible deletion. Since one-shot signatures are publicly verifiable, any party can confirm deletion by checking  
\[
\mathsf{OSS}.\mathsf{Verify}(\mathsf{oss.crs}, \mathsf{oss.pk}, \sigma, 1) = \top.
\]  
Only the signing key $\mathsf{oss.sk}$ involves quantum information and is generated locally by the receiver; all remaining components, including the ciphertext $\mathsf{srabe.ct}$, are entirely classical. Consequently, the interaction between the sender and receiver can proceed over a classical channel. In contrast to prior frameworks that rely on quantum communication or trusted verification authorities, our construction operates entirely over classical communication and supports fully public verifiability. \vspace{-0.2cm}

\subsection{Achieving Certified Everlasting Deletion Security}

We present two constructions that achieve certified everlasting security, ensuring that once deletion is executed, the underlying message can never be recovered, even by adversaries with unlimited computational power in the future. Our methodology leverages fundamental quantum mechanical properties, particularly the uncertainty principle, to design schemes that support both private and public verification.\vspace{0.2cm}

\noindent \textbf{Quantum One-Time Pad.} Our constructions are based on a quantum one-time pad implemented with BB84 states. To encrypt a single message bit $b$, the encryptor samples two uniformly random strings, $x, \theta \leftarrow \{0,1\}^{\lambda}$. It then prepares the quantum state $\ket{x}_\theta = \bigotimes_{j \in [\lambda]} H^{\theta_j} \ket{x_j}$, where each qubit is encoded in either the computational basis (if $\theta_j = 0$) or the Hadamard basis (if $\theta_j = 1$).\vspace{0.2cm}

\noindent The message bit $b$ is masked using the bits of $x$ that were encoded in the Hadamard basis. The classical part of the ciphertext will contain an encryption of the basis string $\theta$ and the masked message, $b \oplus \bigoplus_{j: \theta_j = 1} x_j$. The quantum state $\ket{x}_\theta$ itself serves as the quantum component of the ciphertext. An authorized decryptor, upon recovering $\theta$, can measure the quantum state in the correct bases to learn $x$ and subsequently unmask the message.\vspace{0.2cm}

\noindent \textbf{Privately Verifiable Certified Everlasting Deletion.}  
In our first construction, \textsf{RABE} with Privately Verifiable Certified Everlasting Deletion (\textsf{RABE-PriVCED}), the encryption procedure begins by generating $(x, \theta)$ and preparing the quantum state $\ket{x}_{\theta}$. 
A bitwise concealment value is then computed as $m^{*} := b \oplus \bigoplus_{i:\theta_i=1} x_i$, after which $(\theta, m^{*})$ is encrypted using \textsf{RABE} under the attribute $X$. 
The resulting ciphertext therefore consists of a classical component and a quantum register. During decryption, the algorithm first recovers $(\theta, m^{*})$ from the classical ciphertext, measures $\ket{x}_{\theta}$ in the specified bases to reconstruct $x$, and subsequently retrieves the original message. Deletion is achieved by measuring the quantum register in the Hadamard basis, producing outcomes that serve as a deletion certificate. Since verifying this certificate requires knowledge of $\theta$, which remains secret within the decryption process, the deletion in this scheme is only privately verifiable.   \vspace{0.15cm}

\noindent \textbf{Publicly Verifiable Certified Everlasting Deletion.}  
To achieve public verifiability, we upgrade our construction using a technique inspired by Kitagawa \textit{et al.} \cite{kitagawa2023publicly} to authenticate the quantum state. Our idea is to generate a signature for the classical string $x$ by coherently running a signing algorithm on the quantum state itself. This creates a ``coherently signed'' BB84 state, which embeds a quantum signature that is verifiable using a classical public key. Here, the encryptor generates a signature key pair $(\mathsf{sigk},\mathsf{vk})$ and includes $\mathsf{sigk}$ along with $(\theta,m^*)$ in the classical component of the ciphertext. The quantum component is a coherently signed BB84 state of the form $U_{\mathsf{sign}}\ket{x}_\theta\ket{0\ldots 0} = \ket{x}_\theta\ket{\mathsf{Sign}(\mathsf{sigk},x)}$, with $\mathsf{vk}$ made public. During decryption, the signature register can be “uncomputed” using $\mathsf{sigk}$ to recover the pure BB84 state, after which the procedure mirrors that of the privately verifiable scheme. Deletion is performed by measuring the entire quantum state in the computational basis, yielding both a string $x'$ and a classical signature $\sigma$. The pair $(x',\sigma)$ forms the deletion certificate, which any third party can validate using $\mathsf{vk}$. This makes deletion publicly verifiable without requiring access to hidden basis information.\vspace{0.2cm}

\noindent \textbf{Security Foundation.}
The security relies on fundamental quantum mechanical principles formalized through the Certified Everlasting Lemma (\textsf{CEL}) \cite{bartusek2023cryptography}. {A key distinction from the notion of Certified Deletion lies in the adversary's post-deletion capabilities. In certified deletion security, the challenger reveals all honest secret keys after successful deletion verification, and security relies on the computational assumption that the adversary cannot distinguish the encrypted messages despite having these keys. In contrast, Certified Everlasting security does not grant access to honest secret keys but instead guarantees information-theoretic indistinguishability of the residual state.} When BB84 states prepared on random bases are measured incorrectly during deletion, the resulting probability distributions become computationally indistinguishable for any encrypted messages. The proof employs sophisticated quantum information techniques, including \textsf{EPR} pairs, deferred measurements, and projective measurement analysis. Once deletion occurs through incorrect basis measurement, the quantum superposition irreversibly collapses, making message recovery impossible even with unlimited computational resources. Unlike classical cryptography, which relies on computational assumptions potentially vulnerable to quantum computers or algorithmic advances, our everlasting deletion guarantees remain secure indefinitely. After valid deletion, the privacy protection persists regardless of future technological developments.\vspace{0.2cm}

%% file: 2_prel.tex
\section{Preliminaries}\label{preli}



\noindent This section provides an overview of the cryptographic tools used in this work.
\subsection{Public Key Encryption}\label{def:pke-syntax}
We present the formal definition of the public key encryption (\textsf{PKE}) scheme, along with its correctness and security framework.\vspace{0.2cm}

\noindent\textbf{Syntax.} Let $\lambda$ be a security parameter, and let $p$, $q,$ and $r$ be polynomial functions in $\lambda$. A \textsf{PKE} scheme consists of a tuple of algorithms, $\prod_{\textsf{PKE}}=(\mathsf{KeyGen},$ $\mathsf{Encrypt},$ $\mathsf{Decrypt})$ with message space $\mathcal{M}:= \{0, 1\}^{n}$, ciphertext space $\mathcal{C}:= \{0, 1\}^{p(\lambda)}$, public key space $\mathcal{PK}:= \{0, 1\}^{q(\lambda)},$ and secret key space $\mathcal{SK}:= \{0, 1\}^{r(\lambda)}$, where $n$ is the length of the message. The algorithms are specified as follows:

\begin{description}\setlength\itemsep{0.5em}
    \item$-$~$(\mathsf{pk}, \mathsf{sk})$ $\leftarrow$ \textsf{KeyGen}$(1^{\lambda}).$ The key generation algorithm takes as input the length of the security parameter $1^{\lambda}$ to output a public key $\mathsf{pk} \in \mathcal{PK}$ and a secret key $\mathsf{sk} \in \mathcal{SK}$.
    
    \item$-$ $(\mathsf{CT})$ $\leftarrow$ \textsf{Encrypt}$(\mathsf{pk}, m).$ The encryption algorithm takes as input $\mathsf{pk}$ and a message $m \in \mathcal{M}$ to output a ciphertext $\mathsf{CT} \in \mathcal{C}$.
    
    \item$-$ $(m'~\vee\perp)$ $\leftarrow$ \textsf{Decrypt}$(\mathsf{sk}, \mathsf{CT}).$ The decryption algorithm takes as input $\mathsf{sk}$ and $\mathsf{CT}$ to output either the message $m'$ or $\perp,$ indicating decryption failure.
\end{description}
\noindent We require that a \textsf{PKE} scheme $\prod_{\textsf{PKE}}=(\mathsf{KeyGen},$ $\mathsf{Encrypt},$ $\mathsf{Decrypt})$ satisfies the following properties.
\begin{description}
\item$\bullet$ \textbf{Correctness.}\label{def:pke-correctness}
There exists a negligible function $\mathsf{negl}$ such that for any $\lambda \in \mathbb{N}$, $m \in \mathcal{M}$, it holds that
\[
\Pr\left[
  \mathsf{Decrypt}(\mathsf{sk}, \mathsf{CT}) \neq m 
  \,\bigg|\, 
  \begin{array}{l}
    (\mathsf{pk}, \mathsf{sk}) \leftarrow \mathsf{KeyGen}(1^{\lambda}) \\
    \mathsf{CT} \leftarrow \mathsf{Encrypt}(\mathsf{pk}, m)
  \end{array}
\right] \leq \mathsf{negl}(\lambda).
\]

\item$\bullet$ \textbf{\textsf{IND-CPA} Security.}\label{def:pke-ind-cpa} We define indistinguishability under chosen-plaintext attack (\textsf{IND-CPA}) security through the following experiment, denoted $\mathsf{Exp}^{\mathsf{IND-CPA}}_{\textsf{PKE}, \mathcal{A}}$ $(\lambda,$ $b)$, for a \textsf{PKE} scheme and a quantum polynomial-time adversary (\textsf{QPT}) $\mathcal{A}$.\vspace{0.2cm}

\begin{itemize}\setlength\itemsep{0.3em}
    \item[$-$] The challenger samples a public/secret key pair $(\mathsf{pk}, \mathsf{sk}) \leftarrow \mathsf{KeyGen}(1^\lambda)$ and sends the public key $\mathsf{pk}$ to $\mathcal{A}$.
    
    \item[$-$] The adversary $\mathcal{A}$ selects two messages $(m_0, m_1) \in \mathcal{M}^2$ and submits them to the challenger.

    \item[$-$] The challenger selects a bit $b \in \{0,1\}$ uniformly at random, computes $\mathsf{CT}_b \leftarrow \mathsf{Encrypt}(\mathsf{pk}, m_b)$, and sends $\mathsf{CT}_b$ to $\mathcal{A}$.

    \item[$-$] The adversary $\mathcal{A}$ outputs a guess $b' \in \{0,1\}$ for bit $b$.
\end{itemize}
\end{description}\vspace{0.2cm}

\noindent The advantage of $\mathcal{A}$ in the above experiment is defined by
\begin{align*}
    &\mathsf{Adv}^{\mathsf{IND-CPA}}_{\textsf{PKE}, \mathcal{A}}(\lambda) := \left|\Pr[\mathsf{Exp}^{\mathsf{IND-CPA}}_{\textsf{PKE}, \mathcal{A}}(\lambda, 0) = 1] - \Pr[\mathsf{Exp}^{\mathsf{IND-CPA}}_{\textsf{PKE}, \mathcal{A}}(\lambda, 1) = 1]\right|.
\end{align*}

\noindent We say that the $\prod_{\textsf{PKE}}$ protocol is \textsf{IND-CPA} secure if for all \textsf{QPT} adversaries $\mathcal{A}$, the advantage is negligible in $\lambda,$ \textit{i.e.,} the following holds:
\begin{align*}
    &\mathsf{Adv}^{\mathsf{IND-CPA}}_{\textsf{PKE}, \mathcal{A}}(\lambda) \leq \mathsf{negl}(\lambda).
\end{align*}

 \noindent A notable example is the $\mathsf{Regev}$ $\mathsf{PKE}$ scheme~\cite{regev2009lattices} that achieves $\mathsf{IND}$-$\mathsf{CPA}$ security under the Learning with Errors ($\mathsf{LWE}$) assumption against \textsf{QPT} adversaries. We refer the reader to~\cite{regev2009lattices,gentry2008trapdoors} for a more detailed treatments of the $\mathsf{LWE}$ assumption and additional constructions of the post-quantum secure \textsf{PKE} schemes.

\subsection{Indistinguishability Obfuscation}\label{def:iO-framework}

Let $\{\mathcal{C}_{\lambda}\}_{\lambda \in \mathbb{N}}$ be a family of circuits. A \textsf{PPT} algorithm $i\mathcal{O}$ is an indistinguishability obfuscator for the circuit class if it satisfies the following:

\begin{itemize}
    \item \textbf{Correctness:} For every security parameter $\lambda \in \mathbb{N}$, every circuit $C \in \mathcal{C}_{\lambda}$, and every input $x$, the obfuscated circuit preserves functionality:
    \[
    \Pr[C'(x) = C(x) \mid C' \leftarrow i\mathcal{O}(C)] = 1.
    \]

    \item \textbf{Indistinguishability:}  For any \textsf{QPT} distinguisher $\mathcal{D}$ and any pair of functionally equivalent circuits $C_0, C_1 \in \mathcal{C}_\lambda$ such that $C_0(x) = C_1(x)$ for all inputs $x$ and $|C_0| = |C_1|$, the obfuscations are computationally indistinguishable:
    \[
    \big|\Pr[\mathcal{D}(i\mathcal{O}(C_0)) = 1] - \Pr[\mathcal{D}(i\mathcal{O}(C_1)) = 1]\big| \leq \mathsf{negl}(\lambda).
    \]
\end{itemize}

\noindent In recent literature, candidates for indistinguishability obfuscation secure against \textsf{QPT} adversaries have been proposed \cite{gay2021indistinguishability,wee2021candidate,brakerski2020factoring}. Moreover, the work of \cite{cini2025lattice} proposes a post-quantum $i\mathcal{O}$ candidate based on the \textsf{NTRU} assumption and a novel Equivocal \textsf{LWE} problem, employing a trapdoor-based Gaussian hint-release mechanism that ensures correctness while preventing structural leakage, with security proven against \textsf{QPT} adversaries.

\subsection{Digital Signatures}\label{def:sig-framework}
We now formally define a digital signature scheme, including its syntax, correctness, and a specialized one-time unforgeability notion against quantum adversaries.\vspace{0.2cm}

\noindent\textbf{Syntax.} A digital signature (\textsf{SIG}) scheme consists of a tuple of algorithms $\prod_{\textsf{SIG}}=(\mathsf{Gen}, \mathsf{Sign}, \mathsf{Verify})$ defined over the message space $\mathcal{M}$, signature space $\mathcal{S}$, verification key space $\mathcal{VK}$, and signature key space $\mathcal{SK}$. The algorithms are described in the following:

\begin{description}\setlength\itemsep{0.5em}
    \item[$-$] $(\mathsf{vk}, \mathsf{sk}) \leftarrow \mathsf{Gen}(1^\lambda).$  
    The key generation algorithm takes as input the length of the security parameter $1^\lambda$ to output a key pair consisting of a verification key $\mathsf{vk} \in \mathcal{VK},$ and a signing key $\mathsf{sk} \in \mathcal{SK}$.
    
    \item[$-$] $(\sigma) \leftarrow \mathsf{Sign}(\mathsf{sk}, m).$ Taking as input the signing key $\mathsf{sk}$ and a message $m \in \mathcal{M}$, the signing algorithm outputs a signature $\sigma \in \mathcal{S}$.
    
    \item[$-$] $\{0,1\} \leftarrow \mathsf{Verify}(\mathsf{vk}, m, \sigma).$ On the input of the verification key $\mathsf{vk}$, a message $m$ and a signature $\sigma$, the verification algorithm returns $1$ (accept) if the signature is valid or $0$ (reject) otherwise.
\end{description}

\noindent A digital signature scheme $\prod_{\textsf{SIG}}$ must satisfy the following properties.

\begin{description}
\item[$\bullet$] \textbf{Correctness.}\label{def:sig-correctness}  
There exists a negligible function $\mathsf{negl}(\cdot)$ such that for all security parameters $\lambda \in \mathbb{N}$ and for all messages $m \in \mathcal{M}$, we have
\[
\Pr\left[
  \mathsf{Verify}(\mathsf{vk}, m, \sigma) = 1
  \,\bigg|\,
  \begin{array}{l}
    (\mathsf{vk}, \mathsf{sk}) \leftarrow \mathsf{Gen}(1^{\lambda}) \\
    \sigma \leftarrow \mathsf{Sign}(\mathsf{sk}, m)
  \end{array}
\right] \geq 1 - \mathsf{negl}(\lambda).
\]
\end{description}

\vspace{0.2cm}

\begin{description}
\item[$\bullet$] \textbf{One-Time Unforgeability (\textsf{OT-UF}) for {BB84} States.}\label{def:sig-otuf-bb84}  
We define one-time unforgeability (\textsf{OT-UF}) security for digital signatures in the presence of quantum adversaries with access to BB84 quantum states. We say that a digital signature scheme with the tuples of algorithm $\prod_{\textsf{SIG}}=(\mathsf{Gen}, \mathsf{Sign}, \mathsf{Verify})$ satisfies one-time unforgeability against BB84 attacks if no \textsf{QPT} adversary $\mathcal{A}$ can win the following security game with non-negligible probability.\vspace{0.2cm}

\begin{itemize}\setlength\itemsep{0.3em}
    \item[$-$] The challenger generates a key pair $(\mathsf{vk}, \mathsf{sk}) \leftarrow \mathsf{Gen}(1^\lambda)$ and sends $\mathsf{vk}$ to $\mathcal{A}$.

    \item[$-$] $\mathcal{A}$ selects a message $m^* \in \mathcal{M}$ and receives a valid signature $\sigma^* \leftarrow \mathsf{Sign}(\mathsf{sk}, m^*)$.
    
    \item[$-$] Finally, $\mathcal{A}$ outputs a pair $(m', \sigma')$.
    
    \item[$-$] $\mathcal{A}$ wins if $m' \ne m^*$ and $\mathsf{Verify}(\mathsf{vk}, m', \sigma') = 1$.
\end{itemize}
\vspace{0.2cm}

\noindent The adversary may additionally have quantum access to BB84 states encoding the signed message. We say that a digital signature scheme is one-time unforgeable for BB84 states if this success probability is negligible in $\lambda$.
\end{description}

\noindent In the current literature, Kitagawa et al.~\cite{kitagawa2023publicly} constructed a digital signature scheme satisfying one-time unforgeability for BB84 states under the existence of one-way functions.

\subsection{\textsf{NP} Languages}\label{def:np-languages}
In computational complexity theory, \textsf{NP} represents a non-deterministic polynomial time machine \cite{turing1936computable,cook2000p} (\textit{i.e.,} a machine that can move more than one state from a given configuration). A language $\mathcal{L}$ over an alphabet $\Sigma$ is called an \textsf{NP} languages iff there exists a $k\in \mathbb{N}$ and a polynomial checking relation $R_{\mathcal{L}} \subseteq \Sigma^{*}\times \Sigma^{*}$ such that for all strings $x \in \Sigma^{*}$,
\begin{align*}
	x\in \mathcal{L} \iff \exists y~(|y|\le |x|^k~\text{and}~R_{\mathcal{L}}(x,y)),
\end{align*}
where $|x|$ and $|y|$ are the lengths of the strings $x$ and $y$, respectively.

\subsection{Zero-Knowledge Arguments}\label{def:zka}
We define a zero-knowledge argument system (\textsf{ZKA}) as an interactive protocol between two parties.
Let $\langle A(w), B(z) \rangle(x)$ denote the interaction between interactive machines $A$ and $B$ on a common input $x$, where $A$ receives auxiliary input $w$ and $B$ receives auxiliary input $z$. We refer to the protocol as $\prod_{\textsf{ZKA}} = \langle A, B \rangle$.


\begin{definition}[Classical Proof and Argument Systems for \textsf{NP}]
\label{def:proof-argument-system}
Let $(P, V)$ be a protocol with an honest \textsf{PPT} prover $P$ and an honest \textsf{PPT} verifier $V$ for a language $\mathcal{L} \in \mathsf{NP}$, satisfying:

\begin{enumerate}
    \item \textbf{Perfect Completeness:}  
    For any  $x \in \mathcal{L} $, $R_\mathcal{L}(x,w)$ and $z\in\{0,1\}^*$
    \begin{equation*}
        \Pr[\langle P(w), V(z) \rangle(x) = 1] = 1.
    \end{equation*}

    \item \textbf{Soundness:} For any quantum polynomial-size prover $P^* = \{P^*_{\lambda}, \rho_{\lambda}\}_{\lambda \in \mathbb{N}}$, there exists a negligible function $\mu(\cdot)$ such that for any security parameter $\lambda \in \mathbb{N}$ and any $x \in \{0,1\}^{\lambda} \setminus \mathcal{L}$ and $z\in\{0,1\}^*$
        \begin{equation*}
            \Pr[\langle P^*_{\lambda}(\rho_{\lambda}), V(z) \rangle(x) = 1] \leq \mu(\lambda).
        \end{equation*}
        A protocol with computational soundness is called an \emph{argument}.
        
 \item \textbf{Zero Knowledge:} There exists a quantum polynomial-time simulator $\mathsf{Sim}$ such that for any quantum polynomial-size verifier $V^* = \{V^*_{\lambda}, \rho_{\lambda}\}_{\lambda \in \mathbb{N}}$,
\begin{equation*}
    \big\{
        \langle P(w), V^*_{\lambda}(\rho_{\lambda}) \rangle(x)
    \big\}_{\lambda, x, w}
    \stackrel{c}{\approx}
    \big\{
        \mathsf{Sim}(x, V^*_{\lambda}, \rho_{\lambda})
    \big\}_{\lambda, x},
\end{equation*}
where $\lambda \in \mathbb{N}$, $x \in \mathcal{L} \cap \{0,1\}^{\lambda}$, and $ R_\mathcal{L}(x,w)$.

\noindent If $V^*$ is a classical circuit, then the simulator is computable by a classical polynomial-time algorithm.
    \[
\{ \langle P(w), V^*(z) \rangle (x) \}_{x, w , z }
\stackrel{c}{\approx}
\{ \langle \mathsf{Sim} \rangle(x, z) \}_{x , z }
\]
when the distinguishing gap is considered as a function of $  |x| $.

\end{enumerate}
\end{definition}

\noindent If the \textsf{LWE} assumption holds against quantum polynomial-size adversaries, then there exists a post-quantum zero-knowledge interactive argument system $(P, V)$ for all languages in $\mathsf{NP}$ \cite{bitansky2020post}.

\subsection{Witness Encryption}\label{def:we-framework}
Witness Encryption (\textsf{WE}), introduced by Garg \textit{et al.} \cite{stoc-13}, enables the encryption of a message $m$ with a statement $x$ from an \textsf{NP} language such that decryption is possible only with a corresponding valid witness $w$. We proceed to formally define the syntax of the \textsf{WE} scheme and outline its correctness and security guarantees, particularly in the context of quantum adversaries.\vspace{0.2cm}

\noindent\textbf{Syntax.} A witness encryption scheme for an \textsf{NP} language $\mathcal{L}$ consists of a pair of algorithms $\prod_{\textsf{WE}}=(\textsf{Encrypt}, \textsf{Decrypt})$ described below.

\begin{description}\setlength\itemsep{0.5em}
    \item[$-$] $(\mathsf{CT}) \leftarrow \textsf{Encrypt}(1^\lambda, x, m).$  
    On input a security parameter $1^\lambda$, an instance $x \in \mathcal{L}$, and a message $m \in \mathcal{M}$, the encryption algorithm returns a ciphertext $\mathsf{CT}$.

    \item[$-$] $(m\vee\bot) \leftarrow \textsf{Decrypt}(\mathsf{CT}, w).$  
    Taking as input a ciphertext $\mathsf{CT}$ and a witness $w$ such that $(x, w) \in R_\mathcal{L}$, the decryption algorithm outputs either the message $m$ or a distinguished failure symbol $\bot$.
\end{description}

\vspace{0.2cm}
\noindent A \textsf{WE} scheme $\prod_{\textsf{WE}}$ should satisfy the following correctness and security properties:

\begin{description}
\item[$\bullet$] \textbf{Correctness.}\label{def:we-correctness}  
We require that for all security parameters $\lambda \in \mathbb{N}$, all messages $m \in \mathcal{M}$, and all $(x,w) \in R_\mathcal{L}$, decryption correctly recovers the original message except with negligible probability:
\[
\Pr\left[
  \textsf{Decrypt}(\textsf{Encrypt}(1^\lambda, x, m), w) = m
\right] \ge 1 - \mathsf{negl}(\lambda).
\]

\item[$\bullet$] \textbf{Security.}\label{def:we-security}  
We consider two security attributes for witness encryption: {soundness} and {extractability}.\vspace{0.2cm}

\begin{description}\setlength\itemsep{0.5em}
    \item[$-$] \textbf{Soundness.}  
    A \textsf{WE} scheme satisfies soundness if for any $x \notin \mathcal{L}$, ciphertexts corresponding to different messages are computationally indistinguishable. More precisely, for any \textsf{QPT} adversary $\mathcal{A}$, and for all $m_0, m_1 \in \mathcal{M}$ and $x \notin \mathcal{L}$, we define the advantage:
    \[
    \mathsf{Adv}_\mathcal{A}^{\textsf{WE}}(\lambda) = 
    \Big|
        \Pr[\mathcal{A}(\textsf{Encrypt}(1^\lambda, x, m_0)) = 1] 
        - \Pr[\mathcal{A}(\textsf{Encrypt}(1^\lambda, x, m_1)) = 1]
    \Big|.
    \]
    
    We require that $\mathsf{Adv}_\mathcal{A}^{\textsf{WE}}(\lambda) \le \mathsf{negl}(\lambda)$.\vspace{0.25cm}

    We emphasize that the correctness property of \textsf{WE} scheme guarantees the successful decryption when $x \in \mathcal{L}$ and a valid witness $w$ for the relation $R$ is provided. In contrast, the soundness property ensures that if $x \notin \mathcal{L}$, no efficient adversary can distinguish between encryptions of different messages, and thus cannot recover meaningful information. We remark that the definition intentionally does not address the case where $x \in \mathcal{L}$ but no valid witness for $x$ is known. In such cases, the behaviour of the decryption algorithm is unspecified by the correctness and soundness definitions, and no guarantees are provided.
    
    \item[$-$] \textbf{Extractability}.  
    Extractability property \cite{goldwasser2013run,hiroka2021quantum} ensures that if an adversary can distinguish between encryptions of two messages with non-negligible advantage, then one can efficiently extract a valid witness for the instance. Formally, for any \textsf{QPT} adversary $\mathcal{A}$, any polynomial $p(\lambda)$, and any $m_0, m_1 \in \mathcal{M}$, there exists a \textsf{QPT} extractor $\mathcal{E}$ and polynomial $q(\lambda)$ such that for any auxiliary quantum state $\mathsf{aux}$ potentially entangled with external registers, if:
    \[
    \left|
        \Pr[\mathcal{A}(\mathsf{aux}, \textsf{Encrypt}(1^\lambda, x, m_0)) = 1] 
        - \Pr[\mathcal{A}(\mathsf{aux}, \textsf{Encrypt}(1^\lambda, x, m_1)) = 1]
    \right| \ge \frac{1}{p(\lambda)},
    \]
    then, the extractor succeeds in producing a valid witness with non-negligible probability:
    \[
    \Pr\left[
        (x, \mathcal{E}(1^\lambda, x, \mathsf{aux})) \in R_\mathcal{L}
    \right] \ge \frac{1}{q(\lambda)}.
    \]
\end{description}
\end{description}
\noindent {Extractable witness encryption was introduced in \cite{goldwasser2013run} as a syntactic modification of the witness encryption framework of \cite{garg2013witness}}.
\subsection{{Secret Key} Encryption with Certified Deletion}\label{subsection-enc-cd}
Broadbent and Islam introduced the concept of encryption with certified deletion in the context of secret key encryption (\textsf{SKE}) \cite{broadbent2020quantum}. Their framework considers a one-time usage model, where the secret key is employed for a single encryption instance. We now present the formal definition of a secret key encryption scheme with certified deletion (\textsf{SKE-CD}) and outline its associated security requirements.\vspace{0.2cm}

\noindent\textbf{Syntax of \textsf{SKE-CD}.}\label{def:ske-cd-syntax} A secret key encryption scheme with certified deletion (\textsf{SKE-CD}) is a tuple of \textsf{QPT} algorithms $\prod_{\textsf{SKE-CD}}=(\mathsf{KeyGen}, \mathsf{Encrypt},$ $\mathsf{Decrypt},$ $\mathsf{Deletion},$ $\mathsf{Verify})$ with plaintext space $\mathcal{M}$ and key space $\mathcal{K},$ with the following algorithms.

\begin{description}\setlength\itemsep{0.5em}
    \item[$-$] $(\mathsf{sk})\leftarrow\textsf{KeyGen}(1^{\lambda}).$ On input the length of the security parameter $1^{\lambda}$, the key generation algorithm outputs a secret key $\mathsf{sk} \in \mathcal{K}$.

    \item[$-$] $(\mathsf{CT}) \leftarrow \textsf{Encrypt}(\mathsf{sk}, m).$ The encryption algorithm takes as input $\mathsf{sk}$ and a plaintext $m \in \mathcal{M}$ to output a ciphertext $\mathsf{CT}$.

    \item[$-$] $(m'\vee\bot) \leftarrow\textsf{Decrypt}(\mathsf{sk}, \mathsf{CT}).$ Taking as input $\mathsf{sk}$ and $\mathsf{CT},$ the decryption algorithm outputs a plaintext $m' \in \mathcal{M}$ or $\bot$.

    \item[$-$] $(\mathsf{cert})\leftarrow\textsf{Deletion}(\mathsf{CT}).$ The deletion algorithm, on input the ciphertext $\mathsf{CT}$, produces a deletion certificate $\mathsf{cert}$.

    \item[$-$] $(\top \vee \bot)\leftarrow\textsf{Verify}(\mathsf{sk}, \mathsf{cert}).$ The verification algorithm takes as input the secret key $\mathsf{sk}$ and the certificate $\mathsf{cert}$ to return $\top$ {or} $\bot$.
\end{description}

\begin{description}\setlength\itemsep{0.5em}
    \item[$\bullet$] \textbf{Correctness.}\label{def:ske-cd-correctness} A \textsf{SKE-CD} scheme $\prod_{\textsf{SKE-CD}} = (\mathsf{KeyGen}, \mathsf{Encrypt},$ $ \mathsf{Decrypt},$ $ \mathsf{Deletion},$ $ \mathsf{Verify})$ 
    is said to be correct if it satisfies the following two properties.\vspace{0.2cm}
    
    \begin{description}\setlength\itemsep{0.5em}
        \item[$-$] \textbf{Decryption correctness:} For any $\lambda \in \mathbb{N}$, $m \in \mathcal{M}$, 
        \[
        \Pr\left[\mathsf{Decrypt}(\mathsf{sk}, \mathsf{CT}) \neq m \,\middle|\,
        \begin{array}{l}
            \mathsf{sk} \leftarrow \mathsf{KeyGen}(1^{\lambda}), \\
            \mathsf{CT} \leftarrow \mathsf{Encrypt}(\mathsf{sk}, m)
        \end{array}
        \right] \leq \mathsf{negl}(\lambda).
        \]

        \item[$-$] \textbf{Verification correctness:} For any $\lambda \in \mathbb{N}$, $m \in \mathcal{M}$, 
        \[
        \Pr\left[\mathsf{Verify}(\mathsf{sk}, \mathsf{cert}) = \bot \,\middle|\,
        \begin{array}{l}
            \mathsf{sk} \leftarrow \mathsf{KeyGen}(1^{\lambda}), \\
            \mathsf{CT} \leftarrow \mathsf{Encrypt}(\mathsf{sk}, m), \\
            \mathsf{cert} \leftarrow \mathsf{Deletion}(\mathsf{CT})
        \end{array}
        \right] \leq \mathsf{negl}(\lambda).
        \]
    \end{description}
    \item[$\bullet$] \textbf{Certified Deletion Security.}\label{def:otske-cd-security} The notion of one-time certified deletion (\textsf{OT-CD)} security for a \textsf{SKE-CD} scheme is defined via the security experiment $\mathsf{Exp}^{\mathsf{OT-CD}}_{\textsf{SKE-CD},~\mathcal{A}}(\lambda, b)$, which models the interaction between the challenger and a \textsf{QPT} adversary~$\mathcal{A}$.\vspace{0.2cm} 

    \begin{description}\setlength\itemsep{0.5em}
        \item[$-$] \textbf{KeyGen Phase:} The challenger generates a secret key $\mathsf{sk} \leftarrow \mathsf{KeyGen}(1^{\lambda})$.

        \item[$-$] \textbf{Challenge:} $\mathcal{A}$ submits two messages $(m_0, m_1) \in \mathcal{M}^2$ to the challenger. It samples a bit $b \in \{0,1\}$ and computes the ciphertext $\mathsf{CT}_b \leftarrow \mathsf{Encrypt}(\mathsf{sk}, m_b)$, which is then sent to $\mathcal{A}$.

        \item[$-$] \textbf{Certificate Verification:} $\mathcal{A}$ sends $\mathsf{cert}$ to the challenger. The challenger computes $\mathsf{Verify}(\mathsf{sk}, \mathsf{cert})$. If the output is $\perp$, the challenger sends $\perp$ to $\mathcal{A}$. If the output is $\top$, the challenger sends $\mathsf{sk}$ to $\mathcal{A}$.

        \item[$-$] \textbf{Guess:} $\mathcal{A}$ outputs a guess $b' \in \{0,1\}$.
    \end{description}\vspace{0.2cm}
\end{description}
    \noindent The advantage of $\mathcal{A}$ in breaking certified deletion is defined as
    \[
    \mathsf{Adv}^{\mathsf{OT-CD}}_{\textsf{SKE-CD},~\mathcal{A}}(\lambda) := 
    \left| \Pr\left[\mathsf{Exp}^{\mathsf{OT-CD}}_{\textsf{SKE-CD},~\mathcal{A}}(\lambda, 0) = 1\right] - 
    \Pr\left[\mathsf{Exp}^{\mathsf{OT-CD}}_{\textsf{SKE-CD},~\mathcal{A}}(\lambda, 1) = 1\right] \right|.
    \]\vspace{0.2cm}

    \noindent We say that a \textsf{SKE-CD} scheme $\prod_{\textsf{SKE-CD}}$ achieves \textsf{OT-CD} security if for all \textsf{QPT} adversaries $\mathcal{A},$ the advantage is negligible in $\lambda,$ \textit{i.e.,} the following holds:
    
    $$\mathsf{Adv}^{\mathsf{OT-CD}}_{\textsf{SKE-CD},~\mathcal{A}}(\lambda) \leq \mathsf{negl}(\lambda).$$
    
     \noindent A secret key encryption scheme satisfying the \textsf{OT-CD} security property is referred to as an \textsf{OTSKE-CD} scheme.

\begin{remark}[Impossibility of Reverse Certified Deletion]\label{rem:cd-intuition}
The notion of one-time certified deletion security ensures that once a valid deletion certificate has been produced, the corresponding ciphertext can no longer be decrypted. A natural question arises as to whether the converse can also be enforced, that once a ciphertext is decrypted, it should no longer be possible to generate a valid deletion certificate. However, this property is impossible to achieve due to the decryption correctness. If the quantum decryption algorithm $\mathsf{Decrypt}$ correctly recovers the plaintext from a quantum ciphertext $\mathsf{CT}$ with overwhelming probability, then by the gentle measurement lemma, the ciphertext remains only negligibly perturbed by the decryption process. As a result, it remains feasible to apply the deletion algorithm and obtain a valid certificate subsequently.\vspace{0.2cm}

\noindent It is noted that in existing constructions of secret key encryption with certified deletion, the secret key is a classical bitstring, while ciphertexts are quantum states. The following result, due to Broadbent and Islam \cite{broadbent2020quantum}, establishes the unconditional existence of such schemes.
\end{remark}
\begin{theorem}[Unconditional Existence of \textsf{OT-CD} Secure \textsf{SKE}]\label{thm:otske-cd-existence}
    There exists a one-time secret key encryption scheme with certified deletion that satisfies \textsf{OT-CD} security unconditionally, where the message space is $\mathcal{M} = \{0,1\}^{\ell_m}$ and the key space is $\mathcal{K} = \{0,1\}^{\ell_k}$ for some polynomials $\ell_m, \ell_k$.
\end{theorem}

\subsection{Certified Everlasting Lemmas}\label{def:cert-ever-last-lemma}

In this section, we revisit the certified everlasting lemma of Bartusek and Khurana \cite{bartusek2023cryptography}, which provides cryptography with privately verifiable deletion. Later, Kitagawa \textit{et al.} \cite{kitagawa2023publicly} generalized this notion to achieve publicly verifiable deletion. Building on these foundations, we design a new registered attribute-based encryption scheme that supports publicly verifiable deletion. \vspace{0.2cm}

\begin{lemma}[Certified Everlasting Lemma \cite{bartusek2023cryptography}]\label{lem:certified-everlasting-bk}
Let $\{\mathcal{Z}_{\lambda}(\cdot, \cdot, \cdot)\}_{\lambda \in \mathbb{N}}$ be an efficient quantum operation taking three inputs: (i) a $\lambda$-bit string $\theta$, (ii) a single bit $\beta$, (iii) a quantum register \textcolor{gray}{\textsf{A}}. For any \textsf{QPT} $\mathcal{B}$ algorithm, define the experiment $\widetilde{\mathcal{Z}}^{\mathcal{B}}_{\lambda}(b)$ over quantum states by executing $\mathcal{B}$ according to the following procedure:

\begin{description}\setlength\itemsep{0.5em}
    \item[$-$] Sample $x, \theta \leftarrow \{0, 1\}^{\lambda}$ uniformly at random.  
    \item[$-$] Initialize $\mathcal{B}$ with the state $\mathcal{Z}_{\lambda}\left(\theta,\, b \oplus \bigoplus_{i: \theta_i = 1} x_i,\, \ket{x}_{\theta}\right).$
    
    \item[$-$] Let the output of $\mathcal{B}$ be a classical string $x' \in \{0, 1\}^{\lambda}$ along with a residual quantum state on its register  \textcolor{gray}{\textsf{B}}.
    
    \item[$-$] If $x_i = x'_i$ for all indices $i$ where $\theta_i = 0$, output \textcolor{gray}{\textsf{B}}. Otherwise, output a special failure symbol $\bot$.
\end{description}

\noindent Assume that for any \textsf{QPT} $\mathcal{A}$, any $\theta \in \{0, 1\}^{\lambda}$, $\beta \in \{0, 1\}$, and any efficiently samplable state $\ket{\psi}_{\textcolor{gray}{\textsf{A}},\textcolor{gray}{\textsf{C}}}$ over registers \textcolor{gray}{\textsf{A}} and \textcolor{gray}{\textsf{C}}, the following indistinguishability condition holds:
\[
\left|\Pr[\mathcal{A}(\mathcal{Z}_{\lambda}(\theta, \beta, \textcolor{gray}{\textsf{A}}), \textcolor{gray}{\textsf{C}}) = 1] - \Pr[\mathcal{A}(\mathcal{Z}_{\lambda}(0^{\lambda}, \beta, \textcolor{gray}{\textsf{A}}), \textcolor{gray}{\textsf{C}}) = 1]\right| \leq \mathsf{negl}(\lambda).
\]
\noindent Then, for any \textsf{QPT} $\mathcal{B}$, we have
\[
\mathsf{TD}(\widetilde{\mathcal{Z}}^{\mathcal{B}}_{\lambda}(0),\, \widetilde{\mathcal{Z}}^{\mathcal{B}}_{\lambda}(1)) \leq \mathsf{negl}(\lambda).
\]
\end{lemma}
\vspace{0.2cm}

\noindent Kitagawa \textit{et al.}~\cite{kitagawa2023publicly} upgraded the Lemma \ref{lem:certified-everlasting-bk} to a publicly verifiable version by employing a signature scheme with one-time unforgeability for BB84 states.

\begin{lemma}[Publicly Verifiable Certified Everlasting Lemma \cite{kitagawa2023publicly}]\label{lem:pv-certified-everlasting}
Let $\prod_{\mathsf{SIG}} = (\mathsf{Gen}, \mathsf{Sign}, \mathsf{Verify})$ denote a signature scheme that is one-time unforgeable for BB84 states. Consider a family of efficient quantum operations $\{\mathcal{Z}_{\lambda}(\cdot,$ $\cdot,$ $\cdot,$ $\cdot,$ $\cdot) \}_{\lambda \in \mathbb{N}}$, each taking five inputs: (i) a verification key $\mathsf{vk},$ (ii) a signing key $\mathsf{sigk},$ (iii) a $\lambda$-bit string $\theta,$ (iv) a single bit $\beta,$ (v) a quantum register \textcolor{gray}{\textsf{A}}. For any \textsf{QPT} algorithm $\mathcal{B}$, define the experiment $\widetilde{\mathcal{Z}}^{\mathcal{B}}_{\lambda}(b)$ over quantum states as the outcome of executing $\mathcal{B}$ according to the following procedure: \vspace{0.2cm}

\begin{description}\setlength\itemsep{0.5em}
    \item[$-$] Sample random strings $x, \theta \in \{0,1\}^{\lambda}$ and generate a key pair $(\mathsf{vk}, \mathsf{sigk}) \gets \mathsf{Gen}(1^{\lambda})$. 
    
    \item[$-$] Construct the quantum state $\ket{\psi}$ by applying the mapping 
    \[
    \ket{m}\ket{0 \ldots 0} \ \longmapsto\ \ket{m} \ket{\mathsf{Sign}(\mathsf{sigk}, m)}
    \] to the input $\ket{x}_{\theta} \otimes \ket{0 \ldots 0}$. 
    \item[$-$] Initialize $\mathcal{B}$ with parameters
        \[
        \mathcal{Z}_{\lambda}\left(\mathsf{vk}, \mathsf{sigk}, \theta,\, b \oplus \bigoplus_{i : \theta_i = 1} x_i,\, \ket{\psi}\right).
        \]

    \item[$-$] Parse $\mathcal{B}$’s output as a pair $(x', \sigma')$ and a residual quantum state on register \textcolor{gray}{\textsf{B}}.
    
    \item[$-$] If $\mathsf{Verify}(\mathsf{vk}, x', \sigma') = \top$, output the final state of \textcolor{gray}{\textsf{B}}; otherwise output the special symbol $\bot$.
\end{description}

\noindent Suppose that, for every \textsf{QPT} adversary $\mathcal{A}$, any key pair 
$(\mathsf{vk}, \mathsf{sigk})$ of $\prod_{\mathsf{SIG}}$, any $\theta \in \{0,1\}^{\lambda}$, 
any $\beta \in \{0,1\}$, and any efficiently samplable quantum state 
$\ket{\psi}_{\textcolor{gray}{\textsf{A},}\textcolor{gray}{\textsf{C}}}$ over registers \textcolor{gray}{\textsf{A}} and \textcolor{gray}{\textsf{C}}, the following properties hold:
\begin{align*}
&\left| \Pr\big[ \mathcal{A}(\mathcal{Z}_{\lambda}(\mathsf{vk}, \mathsf{sigk}, \theta, \beta, \textcolor{gray}{\textsf{A}}), \textcolor{gray}{\textsf{C}}) = 1 \big]
- \Pr\big[ \mathcal{A}(\mathcal{Z}_{\lambda}(\mathsf{vk}, 0^{\ell_{\mathsf{sigk}}}, \theta, \beta, \textcolor{gray}{\textsf{A}}), \textcolor{gray}{\textsf{C}}) = 1 \big] \right|
\leq \mathsf{negl}(\lambda), \\
&\left| \Pr\big[ \mathcal{A}(\mathcal{Z}_{\lambda}(\mathsf{vk}, \mathsf{sigk}, \theta, \beta, \textcolor{gray}{\textsf{A}}), \textcolor{gray}{\textsf{C}}) = 1 \big]
- \Pr\big[ \mathcal{A}(\mathcal{Z}_{\lambda}(\mathsf{vk}, \mathsf{sigk}, 0^{\lambda}, \beta, \textcolor{gray}{\textsf{A}}), \textcolor{gray}{\textsf{C}}) = 1 \big] \right|
\leq \mathsf{negl}(\lambda). 
\end{align*}

\noindent It then follows that, for every \textsf{QPT} adversary $\mathcal{B}$,
\[
\mathsf{TD}(\widetilde{\mathcal{Z}}^{\mathcal{B}}_{\lambda}(0), \widetilde{\mathcal{Z}}^{\mathcal{B}}_{\lambda}(1)) \leq \mathsf{negl}(\lambda).
\]
\end{lemma}\vspace{0.2cm}
\subsection{One-Shot Signature}\label{def:oss-framework}
In this section, we present the formal definition of a one-shot signature (\textsf{OSS}) scheme as introduced by Amos \textit{et al.}~\cite{amos2020one}. Such schemes are designed for settings where each key pair is used to sign at most one message and rely on quantum capabilities for key generation and signing.\vspace{0.2cm}

\noindent\textbf{Syntax.} A one-shot signature scheme consists of a tuple of \textsf{QPT} algorithms $\prod_{\textsf{OSS}}=(\textsf{Setup}, \textsf{KeyGen}, \textsf{Sign}, \textsf{Verify})$ over a message space $\mathcal{M}$ and signature space $\mathcal{S}$, which are described below.\vspace{0.2cm}

\begin{description}\setlength\itemsep{0.5em}
    \item[$-$] $(\textsf{crs}) \leftarrow \textsf{Setup}(1^\lambda).$ On input the security parameter $1^\lambda$, the setup algorithm outputs a classical common reference string $\textsf{crs}$.

    \item[$-$] $(\mathsf{pk}, \mathsf{sk}) \leftarrow \textsf{KeyGen}(\textsf{crs}).$  
    Given the common reference string $\textsf{crs}$, the key generation algorithm outputs a key pair consisting of a classical public key $\mathsf{pk}$ and a quantum secret key $\mathsf{sk}$.

    \item[$-$] $(\sigma) \leftarrow \textsf{Sign}(m, \mathsf{sk}).$ The signing algorithm intakes the quantum secret-key $sk$ and a message $m$ from the message space $\mathcal{M}$, which have to sign. It outputs a signature $\sigma \in \mathcal{S}$ and makes publicly available.

    \item[$-$] $(\top\vee\bot) \leftarrow \textsf{Verify}(\textsf{crs}, \mathsf{pk}, m, \sigma).$ Taking as input the common reference string $crs$, classical public-key $pk$, message $m$ and the signature $\sigma$, the algorithm can verify whether the signature is valid or not.
 	\[
        \textsf{Verify}(crs,pk,m,\sigma)=
 	\begin{cases}
 		\top &\text{if} ~~\sigma~\text{is valid.}\\
 		\bot &\text{if} ~~\sigma~\text{is invalid.}
 	\end{cases}	
 	\]
\end{description}

\noindent A one-shot signature scheme must satisfy correctness and security properties as defined below.\vspace{0.2cm}

\begin{description}
\item[$\bullet$] \textbf{Correctness.}  
We say that an \textsf{OSS} scheme is said to be correct if for all security parameters $\lambda \in \mathbb{N}$ and all messages $m \in \mathcal{M}$, the probability that a valid signature fails to verify is negligible:
\[
\Pr\left[
    \textsf{Verify}(crs, \mathsf{pk}, m, \sigma) = \bot
    \,\middle|\,
    \begin{array}{l}
        crs \leftarrow \textsf{Setup}(1^\lambda), \\
        (\mathsf{pk}, \mathsf{sk}) \leftarrow \textsf{KeyGen}(crs), \\
        \sigma \leftarrow \textsf{Sign}(m, \mathsf{sk})
    \end{array}
\right] \leq \mathsf{negl}(\lambda).
\]
\end{description}

\vspace{0.2cm}
\begin{description}
\item[$\bullet$] \textbf{Security.}  
We say that the one-shot signature scheme is secure if no \textsf{QPT} adversary $\mathcal{A}$ can produce two valid message-signature pairs under the same public key. Specifically, for all such adversaries $\mathcal{A}$ and for all $i\in \{0,1\}$, we require:
\[
\Pr\left[
    \begin{array}{l}
        \textsf{Verify}(crs, \mathsf{pk}, m_i, \sigma_i) = \top ~\wedge~ m_0 \ne m_1
    \end{array}
    \,\middle|\,
    \begin{array}{l}
        crs \leftarrow \textsf{Setup}(1^\lambda), \\
        (\mathsf{pk}, m_i, \sigma_i) \leftarrow \mathcal{A}(crs)
    \end{array}
\right] \leq \mathsf{negl}(\lambda).
\]
The one-shot nature of the signature illustrates that the signing key can be used to produce a signature on at most one message, and any attempt to produce a second distinct signed message should fail verification with overwhelming probability.
\end{description}

\subsection{Registered Attribute-Based Encryption}\label{model:RABE}

In this section, we present the formal definition of a (key-policy) registered attribute-based encryption (\textsf{RABE}) scheme that supports an attribute universe of polynomial size, adapted from the framework introduced in \cite{hohenberger2023registered}. The notion of \textsf{RABE} extends the concept of registration-based encryption (\textsf{RBE}) \cite{garg2018registration} by incorporating attribute-based access control into the registration-based paradigm. \vspace{0.2cm}

\noindent\textbf{Syntax.}
Let $\lambda$ be the security parameter. We define the attribute universe as $\mathcal{U} = \{\mathcal{U}_\tau\}_{\tau \in \mathbb{N}}$, the policy space over $\mathcal{U}$ as $\mathcal{P} = \{\mathcal{P}_\tau\}_{\tau \in \mathbb{N}}$, where each ${P}\in \mathcal{P}_\tau $ is a mapping $P:\mathcal{P}_\tau: \mathcal{U}_\tau \rightarrow\{0,1\}$ and the message space as $\mathcal{M} = \{\mathcal{M}_\lambda\}_{\lambda \in \mathbb{N}}$. A \textsf{RABE} scheme is defined as a tuple of \textsf{PPT} algorithms $\prod_{\textsf{RABE}} = (\mathsf{Setup},$ $\mathsf{KeyGen},$ $\mathsf{RegPK},$ $\mathsf{Encrypt},$ $\mathsf{Update},$ $\mathsf{Decrypt})$ with the following functionalities:

\begin{description}\setlength\itemsep{0.5em}
\item[$-$] $(\mathsf{crs}) \leftarrow \mathsf{Setup}(1^\lambda, 1^{{\tau}}).$
The setup algorithm generates a common reference string $\mathsf{crs}$ based on the security parameter $\lambda$ and the policy-family parameter $\tau$.

\item[$-$] $(\mathsf{pk}, \mathsf{sk}) \leftarrow \mathsf{KeyGen}(\mathsf{crs}, \mathsf{aux}, P).$ Taking as input the common reference string \textsf{crs}, an auxiliary state $\mathsf{aux},$ and a policy $P\in \mathcal{P}_{\tau},$ the key generation algorithm outputs a public/secret key pair $(\mathsf{pk}, \mathsf{sk})$.

\item[$-$] $(\mathsf{mpk}, \mathsf{aux}') \leftarrow \mathsf{RegPK}(\mathsf{crs}, \mathsf{aux}, \mathsf{pk}, P).$ The registration algorithm takes as input the common reference string $\mathsf{crs}$, the current auxiliary state $\mathsf{aux}$, a public key $\mathsf{pk}$, and a policy $P\in \mathcal{P}_{\tau}$. It updates the auxiliary state to $\mathsf{aux}'$ and outputs the corresponding master public key $\mathsf{mpk}$.

\item[$-$] $(\mathsf{ct}) \leftarrow \mathsf{Encrypt}(\mathsf{mpk}, X, \mu).$
Given the master public key $\mathsf{mpk}$, an attribute { $X \in \mathcal{U_{\tau}}$}, and a message $\mu \in \mathcal{M}$, the encryption algorithm outputs a ciphertext $\mathsf{ct}$.

\item[$-$] $(\mathsf{hsk}) \leftarrow \mathsf{Update}(\mathsf{crs}, \mathsf{aux}, \mathsf{pk}).$ The update algorithm, given the common reference string $\mathsf{crs}$, the auxiliary state $\mathsf{aux}$, and a public key $\mathsf{pk}$, produces a helper key $\mathsf{hsk}$ to facilitate the decryption process.

\item[$-$] $(\mu') \leftarrow \mathsf{Decrypt}(\mathsf{sk}, \mathsf{hsk},{X}, \mathsf{ct}).$
Given a secret key $\mathsf{sk}$, a helper key $\mathsf{hsk}$,an attribute { $X \in \mathcal{U_{\tau}}$} and a ciphertext $\mathsf{ct}$, the decryption algorithm outputs either a message $\mu' \in \mathcal{M}$, the failure symbol $\perp$, or the flag $\mathsf{GetUpdate}$ to indicate that re-synchronization is required.
\end{description}

\vspace{0.2cm}
\noindent\textbf{Correctness and Efficiency.}\label{def:rabe-correctness}
A \textsf{RABE} scheme $\prod_{\textsf{RABE}}$ is said to satisfy correctness if any honestly registered user can successfully decrypt ciphertexts whose associated policies are satisfied by their attribute set, regardless of the presence of adversarially registered keys. Specifically, suppose that a user registers her public key with an access policy. In that case, she must be able to decrypt all ciphertexts encrypted under the corresponding (or any subsequent) master public key, provided the attribute is satisfied by the access policy. This guarantee must hold even when malicious users register malformed keys, as long as the key curator behaves semi-honestly. Given a security parameter $\lambda$, policy parameter $\tau$ the correctness and efficiency of a \textsf{RABE} scheme is formally defined through the following interaction between an adversary $\mathcal{A}$ and a challenger.\vspace{0.2cm}

\begin{description}\setlength\itemsep{0.5em}
    \item[$\bullet$] \textbf{Setup phase:} The challenger samples the common reference string $\mathsf{crs} \gets \mathsf{Setup}(1^\lambda, 1^{\tau})$, initializes the auxiliary input $\mathsf{aux} \gets \bot$ and initial master public key $\mathsf{mpk}_0 \gets \bot$. It also sets the counters $\mathsf{ctr}_{\mathsf{reg}} \gets 0$ and $\mathsf{ctr}_{\mathsf{enc}} \gets 0$, and initializes $\mathsf{ctr}_{\mathsf{reg}}^\ast \gets \infty$ to store the index of the target key. Finally, it sends $\mathsf{crs}$ to $\mathcal{A}$.

    \item[$\bullet$] \textbf{Query phase:} The adversary $\mathcal{A}$ may adaptively issue the following queries.\vspace{0.2cm}
     
    \begin{description}\setlength\itemsep{0.5em}
        \item[$-$] \textbf{Register non-target key query:} On input of a public key \textsf{pk}, and a access policy $P\in \mathcal{P}_{\tau}$, the challenger increments $\mathsf{ctr}_{\mathsf{reg}} \gets \mathsf{ctr}_{\mathsf{reg}} + 1$ and computes $(\mathsf{mpk}_{\mathsf{ctr}_{\mathsf{reg}}}, \mathsf{aux}')$ $\gets$ $\mathsf{RegPK}(\mathsf{crs},$ $\mathsf{aux},$ $\mathsf{pk},$ $P)$. It updates $\mathsf{aux} \gets \mathsf{aux}'$ and returns $(\mathsf{ctr}_{\mathsf{reg}},\mathsf{mpk}_{\mathsf{ctr}_{\mathsf{reg}}},\mathsf{aux})$ to $\mathcal{A}$.

        \item[$-$] \textbf{Register target key query:} On input a policy $P^\ast \in \mathcal{P}_{\tau}$, if a target key has already been registered, the challenger returns $\bot$. Otherwise, it increments the counter $\mathsf{ctr}_{\mathsf{reg}}$ $\gets$ $\mathsf{ctr}_{\mathsf{reg}} + 1$ and generates $(\mathsf{pk}^\ast,$ $\mathsf{sk}^\ast)$ $\gets$ $\mathsf{KeyGen}(\mathsf{crs}, \mathsf{aux}, P^\ast)$. It then registers the key to obtain $(\mathsf{mpk}_{\mathsf{ctr}_{\mathsf{reg}}}, \mathsf{aux}')$ $\gets$ $\textsf{RegPK}(\mathsf{crs},$ $\mathsf{aux},$ $\mathsf{pk}^\ast,$ $P^\ast)$, and computes the helper decryption key $\mathsf{hsk}^\ast$ $\gets$ $\mathsf{Update}(\mathsf{crs},$ $\mathsf{aux},$ $\mathsf{pk}^\ast)$. The challenger sets $\mathsf{aux} \gets \mathsf{aux}'$, updates $\mathsf{ctr}_{\mathsf{reg}}^\ast \gets \mathsf{ctr}_{\mathsf{reg}}$, and returns$(\mathsf{ctr}_{\mathsf{reg}}, \mathsf{mpk}_{\mathsf{ctr}_{\mathsf{reg}}},\mathsf{aux}, \mathsf{pk}^\ast, \mathsf{sk}^\ast, \mathsf{hsk}^\ast)$ to $\mathcal{A}$.

        \item[$-$] \textbf{Encryption query:} On input of an public key index $\mathsf{ctr}^\ast_{\mathsf{reg}}\leq i\leq \mathsf{ctr}_{\mathsf{reg}}$, a message $\mu_{\mathsf{ctr}_{\mathsf{enc}}}$, and an attribute {$X_{\mathsf{ctr}_{\mathsf{enc}}}\in \mathcal{U_{\tau}}$}, if no target key has been registered or $P^\ast (X_{\mathsf{ctr}_{\mathsf{enc}}})=0$, the challenger returns $\bot$. Otherwise, it increments the counter $\mathsf{ctr}_{\mathsf{enc}}$ $\gets$ $\mathsf{ctr}_{\mathsf{enc}} + 1$ and generates $\mathsf{ct_{\mathsf{ctr}_{\mathsf{enc}}}} \gets \mathsf{Encrypt}(\mathsf{mpk}_i, X_{\mathsf{ctr}_{\mathsf{enc}}}, \mu_{\mathsf{ctr}_{\mathsf{enc}}})$ and returns the pair $(\mathsf{ctr}_{\mathsf{enc}},\mathsf{ct_{\mathsf{ctr}_{\mathsf{enc}}}})$.

        \item[$-$] \textbf{Decryption query:} On input of a ciphertext index $1\leq j\leq \mathsf{ctr}_{\mathsf{enc}}$, the challenger computes 
       { $m'_j \gets \mathsf{Decrypt}(\mathsf{sk}^\ast, \mathsf{hsk}^\ast, X_j,  \mathsf{ct}_j)$}. If $m'_j = \mathsf{GetUpdate}$, it updates the helper key 
        $\mathsf{hsk}^\ast \gets \mathsf{Update}(\mathsf{crs}, \mathsf{aux}, \mathsf{pk}^\ast)$ 
        and recomputes $m'_j$. If $m'_j \neq m_j$, the challenger sets $b \gets 1$ and halts.
    \end{description}
    If the adversary has finished making queries and the experiment has not halted (as a result of a decryption query), then the experiment outputs $b=0$.
\end{description}

\noindent The scheme $\prod_{\textsf{RABE}}$ is correct and efficient if for all (possibly unbounded) adversaries $\mathcal{A}$ making at most a polynomial number of queries, the following properties hold:

\begin{itemize}\setlength\itemsep{0.5em}
    \item[$\bullet$] \textbf{Correctness:} There exists a negligible function 
    $\mathsf{negl}(\cdot)$ such that for all $\lambda \in \mathbb{N}$,
    \[
        \Pr[b = 1] = \mathsf{negl}(\lambda).
    \]
    It achieves perfect correctness if $\Pr[b = 1] = 0$.

    \item[$\bullet$] \textbf{Compactness:} Let $N$ denote the number of registration queries. There exists a universal polynomial $\mathsf{poly}(\cdot, \cdot, \cdot)$ such that for all $i\in[N],$
    \[
        |\mathsf{mpk}_i| = \mathsf{poly}(\lambda, |\mathcal{U}|, \log i), \quad 
        |\mathsf{hsk}^\ast| = \mathsf{poly}(\lambda, |\mathcal{U}|, \log N).
    \]

    \item[$\bullet$] \textbf{Update efficiency:} 
    The $\mathsf{Update}$ algorithm is invoked at most $O(\log N)$ times, 
    each execution running in $\mathsf{poly}(\log N)$ time in the RAM model, which allows sublinear input access.
\end{itemize}\vspace{0.2cm}

\noindent\textbf{Security.}\label{def:rabe-security}  
The security notion for a (key-policy) registered \textsf{ABE} scheme is analogous to the standard \textsf{ABE}. 
In the security game, the adversary may register users whose policies are satisfied by the challenge attribute set, as long as their secret keys are honestly generated by the challenger and thus unknown to the adversary. Furthermore, the adversary may register arbitrary public keys for policies of its choice, provided that none satisfy the challenge attributes. \vspace{0.2cm} 

\noindent Let $\prod_{\textsf{RABE}}=$ $(\mathsf{Setup},$ $\mathsf{KeyGen},$ $\mathsf{RegPK}, \mathsf{Encrypt},$ $\mathsf{Update},$ $\mathsf{Decrypt})$ be a registered \textsf{ABE} scheme with attribute universe $\mathcal{U}$, policy space $\mathcal{P}$, and message space $\mathcal{M}$. For a security parameter $\lambda$, an adversary $\mathcal{A}$, and a bit $b \in \{0,1\}$, we define the following security game $\textsf{Exp}^{\textsf{RABE}}_{\mathcal{A}}(\lambda,b)$ between $\mathcal{A}$ and the challenger:

\begin{description}\setlength\itemsep{0.5em}
    \item[$\bullet$] \textbf{Setup:} The challenger runs $\mathsf{crs} \gets \mathsf{Setup}(1^\lambda, 1^{{\tau}})$ and initializes: $\mathsf{aux} \gets \bot$, $\mathsf{mpk} \gets \bot$, a counter $\mathsf{ctr} \gets 0$ for honest-key registrations, $\mathcal{C} \gets \emptyset$ for corrupted public keys, and a dictionary $D \gets \emptyset$ mapping public keys to attribute sets. If $\mathsf{pk} \notin D$, then $D[\mathsf{pk}] \gets \emptyset$. The challenger sends $\mathsf{crs}$ to $\mathcal{A}$.

    \item[$\bullet$] \textbf{Query phase:}  
    In this phase, $\mathcal{A}$ may adaptively issue the following queries.\vspace{0.2cm}
    
    \begin{description}\setlength\itemsep{0.5em}
        \item[$-$] \textbf{Register corrupted key query:} $\mathcal{A}$ submits a public key $\mathsf{pk},$ and a policy $ P \in \mathcal{P}_{\tau}$. The challenger computes $(\mathsf{mpk}', \mathsf{aux}')$ $\gets$ $\mathsf{RegPK}(\mathsf{crs},$ $\mathsf{aux},$ $\mathsf{pk},$ $P)$, updates $\mathsf{mpk} \gets \mathsf{mpk}'$, $\mathsf{aux} \gets \mathsf{aux}'$, adds $\mathsf{pk}$ to $\mathcal{C}$, and sets $D[\mathsf{pk}] \gets D[\mathsf{pk}] \cup \{P\}$. It returns $(\mathsf{mpk}', \mathsf{aux}')$ to $\mathcal{A}$.

        \item[$-$] \textbf{Register honest key query:} $\mathcal{A}$ specifies a policy $ P \in \mathcal{P}_{\tau}$. The challenger updates the counter $\mathsf{ctr} \gets \mathsf{ctr} + 1$ and generates a key pair $(\mathsf{pk}_{\mathsf{ctr}}, \mathsf{sk}_{\mathsf{ctr}}) \gets \mathsf{KeyGen}(\mathsf{crs}, \mathsf{aux}, P)$. Subsequently, it executes  $(\mathsf{mpk}', \mathsf{aux}')$ $\gets$ $\mathsf{RegPK}(\mathsf{crs},$ $\mathsf{aux},$ $\mathsf{pk}_{\mathsf{ctr}}, P)$, and updates $\mathsf{mpk} \gets \mathsf{mpk}'$, $\mathsf{aux} \gets \mathsf{aux}'$, and $D[\mathsf{pk}_{\mathsf{ctr}}] \gets D[\mathsf{pk}_{\mathsf{ctr}}] \cup \{P\}$. Finally, it returns $(\mathsf{ctr}, \mathsf{mpk}', \mathsf{aux}', \mathsf{pk}_{\mathsf{ctr}})$ to $\mathcal{A}$.   
        
       \item[$-$] \textbf{Corrupt honest key query:} $\mathcal{A}$ submits $i \in [\mathsf{ctr}]$. The challenger retrieves $(\mathsf{pk}_i, \mathsf{sk}_i)$, adds $\mathsf{pk}_i$ to $\mathcal{C}$, and returns $\mathsf{sk}_i$ to $\mathcal{A}$.
    \end{description}

    \item[$\bullet$] \textbf{Challenge:} $\mathcal{A}$ outputs two messages $\mu_0^\ast, \mu_1^\ast \in \mathcal{M}$ and an { attribute $X \in \mathcal{U_{\tau}}$}. The challenger samples a bit $b$ from $\{0,1\}$, computes $\mathsf{ct}^\ast \gets \mathsf{Encrypt}(\mathsf{mpk},$ $X^\ast,$ $\mu_b^\ast),$ and returns $\mathsf{ct}^\ast$ to $\mathcal{A}$.

    \item[$\bullet$] \textbf{Output:}  
    Finally, $\mathcal{A}$ outputs a bit $b' \in \{0,1\}$.
\end{description}

\noindent Let $S = \{ P' \in D[\mathsf{pk}] \mid \mathsf{pk} \in \mathcal{C} \}$ be the collection of policies of all corrupted keys. An adversary $\mathcal{A}$ is admissible if for all $P \in S$, it holds that $P(X^\ast)=0$. We say that $\prod_{\textsf{RABE}}$ scheme is secure if for all efficient admissible adversaries, there exists a negligible function $\mathsf{negl}(\cdot)$ such that for all $\lambda \in \mathbb{N}$,
\[
    \big| \Pr[\textsf{Exp}^{\textsf{RABE}}_{\mathcal{A}}(\lambda,0) = 1] - \Pr[\textsf{Exp}^{\textsf{RABE}}_{\mathcal{A}}(\lambda,1) = 1] \big| = \mathsf{negl}(\lambda).
\]

%% file: 3_RABE_CD-CED-Framework.tex
\section{\textsf{RABE} with Certified Deletion and Everlasting Security}\label{section:RABE-Cipher-Deletion}
In this section, we introduce two variants of ciphertext deletion for registered attribute-based encryption. We first formalize the Registered Attribute-Based Encryption with Certified Deletion (\textsf{RABE-CD}), specifying its syntax, correctness, and security guarantees. We then upgrade this framework to capture the notion of Registered Attribute-Based Encryption with Certified Everlasting Deletion (\textsf{RABE-CED}), together with its corresponding correctness and security definitions. For both \textsf{RABE-CD} and \textsf{RABE-CED} systems, we further distinguish between privately and publicly verifiable settings and establish their respective security guarantees. \vspace{0.2cm}

\subsection{System Model of \textsf{RABE-CD}}\label{def:RABE-CD}
Let $\lambda$ be the security parameter, and $\mathcal{U} = \{\mathcal{U}_{\tau}\}_{{\tau}\in \mathbb{N}}$ denote the attribute universe. Let $\mathcal{P} = \{\mathcal{P}_{\tau}\}_{{\tau} \in \mathbb{N}}$ represent the policy space and $\mathcal{M} = \{\mathcal{M}_\lambda\}_{\lambda \in \mathbb{N}}$ represent the message space. A \textsf{RABE-CD} scheme consists of the tuple of efficient algorithms $\prod^{\mathsf{RABE}}_{\textsf{CD}} =$ $(\mathsf{Setup},$ $\mathsf{KeyGen},$ $\mathsf{RegPK},$ $\mathsf{Encrypt},$ $\mathsf{Update},$ $\mathsf{Decrypt},$ $\mathsf{Delete},$ $\mathsf{Verify})$ with the following syntax:\vspace{0.2cm}

\begin{description}\setlength\itemsep{0.5em}
\item[$-$] $(\textsf{crs}) \leftarrow \textsf{Setup}(1^{\lambda}, 1^{{\tau}}).$  
On the input of the length of the security parameter $\lambda$ and the size of the attribute universe $\mathcal{U}$, the setup algorithm outputs a common reference string $\mathsf{crs}$.
\item[$-$] {(\textsf{pk}, \textsf{sk})} $\leftarrow \mathsf{KeyGen}(\mathsf{crs}, \mathsf{aux}, P).$  
Given the common reference string \textsf{crs}, an auxiliary state $\mathsf{aux}$ (possibly empty) and a policy $P\in \mathcal{P}_{\tau}$, the key generation algorithm outputs a public key $\mathsf{pk}$ and secret key $\mathsf{sk}$.
\item[$-$] $(\mathsf{mpk}, \mathsf{aux}') \leftarrow \mathsf{RegPK}(\mathsf{crs}, \mathsf{aux}, \mathsf{pk}, P).$  The registration algorithm takes as input a common reference string \textsf{crs}, an auxiliary state \textsf{aux}, a public key $\mathsf{pk}$, and a policy $P\in \mathcal{P}_{\tau}$ to return an updated master public key $\mathsf{mpk}$ and auxiliary state $\mathsf{aux}'$.
\item[$-$] $(\mathsf{vk}, \mathsf{ct}) \leftarrow \mathsf{Encrypt}\left( \mathsf{mpk}, X, \mu \right).$ The encryption algorithm takes as input the master public key $\mathsf{mpk}$, an attribute {$X \in \mathcal{U_{\tau}}$}, and a message $\mu$ to produce a verification key $\mathsf{vk}$ together with a ciphertext $\mathsf{ct}$.

\item[$-$] $(\mathsf{hsk}) \leftarrow \mathsf{Update}(\mathsf{crs}, \mathsf{aux}, \mathsf{pk}).$ Taking as input the common reference string \textsf{crs}, an auxiliary state \textsf{aux}, and a public key \textsf{pk}, the update algorithm deterministically outputs a helper secret key $\mathsf{hsk}$.
\item[$-$] $(\mu \cup \{\perp, \mathsf{GetUpdate}\}) \leftarrow \mathsf{Decrypt}(\mathsf{sk}, \mathsf{hsk},{X}, \mathsf{ct}).$ On input of the secret key \textsf{sk}, a helper secret key \textsf{hsk}, {an attribute $X \in \mathcal{U_{\tau}}$} and a ciphertext \textsf{ct}, the decryption algorithm outputs the message $\mu$, a failure symbol $\perp$, or the flag $\mathsf{GetUpdate}$ indicating that an updated helper secret key is needed.
\item[$-$] $(\mathsf{cert}) \leftarrow \mathsf{Delete}(\mathsf{ct}).$  
The deletion algorithm generates a deletion certificate $\mathsf{cert}$ for the given ciphertext \textsf{ct}.
\item[$-$] $(b \in \{0,1\}) \leftarrow \mathsf{Verify}(\mathsf{vk}, \mathsf{cert}).$  
Given a verification key \textsf{vk} and a deletion certificate \textsf{cert}, the verification algorithm outputs {$0$ or $1$}.
\end{description}
 
\subsection{Correctness and Efficiency of \textsf{RABE-CD}}\label{def:rabe-cd-correctness}
Let $\prod^{\mathsf{RABE}}_{\textsf{CD}}$ denote a registered attribute-based encryption scheme with certified deletion protocol defined over an attribute universe $\mathcal{U}$, a policy space $\mathcal{P}$, and a message space $\mathcal{M}$. To define the correctness of $\prod^{\mathsf{RABE}}_{\textsf{CD}}$, we specify the following two experiments, Decryption Game and Verification Game, played between an adversary $\mathcal{A}$ and a challenger.\vspace{0.2cm}

\begin{description}\setlength\itemsep{0.5em}
    \item[$\bullet$] \textbf{Setup phase:} The challenger starts by sampling the common reference string $\mathsf{crs} \leftarrow \mathsf{Setup}(1^\lambda, 1^{{\tau}})$. Then it initializes the auxiliary input as $\mathsf{aux} \leftarrow \bot$ and sets the initial master public key to $\mathsf{mpk}_0 \leftarrow \bot$. Next, it initializes a counter $\mathsf{ctr}_{\mathsf{reg}} \leftarrow 0$ to record the number of registration queries and another counter $\mathsf{ctr}_{\mathsf{enc}} \leftarrow 0$ to track the number of encryption queries. In addition, it sets $\mathsf{ctr}^\ast_{\mathsf{reg}} \leftarrow \infty$ as the index of the target key, which may be updated during the game. At the end of this phase, the challenger sends $\mathsf{crs}$ to $\mathcal{A}$.
    
    \item[$\bullet$] \textbf{Query phase:} The adversary $\mathcal{A}$ may adaptively issue the following queries.\vspace{0.2cm}
    
    \begin{description}\setlength\itemsep{0.5em}
        \item[$-$] \textbf{Register non-target key query:} In this query, $\mathcal{A}$ provides a public key $\mathsf{pk}$ and a policy $P \in \mathcal{P}$. The challenger first increments the counter $\mathsf{ctr}_{\mathsf{reg}} \gets \mathsf{ctr}_{\mathsf{reg}} + 1$, then registers the key by computing $(\mathsf{mpk}_{\mathsf{ctr}_{\mathsf{reg}}}, \mathsf{aux}') \gets \mathsf{RegPK}(\mathsf{crs}, \mathsf{aux}, \mathsf{pk}, P)$. It updates the auxiliary data by setting $\mathsf{aux} \gets \mathsf{aux}'$, and replies to $\mathcal{A}$ with $(\mathsf{ctr}_{\mathsf{reg}}, \mathsf{mpk}_{\mathsf{ctr}_{\mathsf{reg}}}, \mathsf{aux})$.
        
        \item[$-$] \textbf{Register target key query:} In the register target key query, $\mathcal{A}$ specifies a policy $P^\ast \in \mathcal{P}$. If a target key has already been registered, the challenger replies with $\bot$. Otherwise, it increments the counter $\mathsf{ctr}_{\mathsf{reg}} \gets \mathsf{ctr}_{\mathsf{reg}} + 1$, samples $(\mathsf{pk}^\ast, \mathsf{sk}^\ast) \gets \mathsf{KeyGen}(\mathsf{crs}, \mathsf{aux}, P^\ast)$, and registers the key by computing $(\mathsf{mpk}_{\mathsf{ctr}_{\mathsf{reg}}}, \mathsf{aux}') \gets \mathsf{RegPK}(\mathsf{crs}, \mathsf{aux}, \mathsf{pk}^\ast, P^\ast)$. Then it computes the helper decryption key $\mathsf{hsk}^\ast \gets \mathsf{Update}(\mathsf{crs}, \mathsf{aux}, \mathsf{pk}^\ast)$, updates $\mathsf{aux} \gets \mathsf{aux}'$, and records the index of the target key as $\mathsf{ctr}^\ast_{\mathsf{reg}} \gets \mathsf{ctr}_{\mathsf{reg}}$. At the end, it responds to $\mathcal{A}$ with $(\mathsf{ctr}_{\mathsf{reg}}, \mathsf{mpk}_{\mathsf{ctr}_{\mathsf{reg}}}, \mathsf{aux}, \mathsf{pk}^\ast, \mathsf{hsk}^\ast, \mathsf{sk}^\ast)$.

        \item[$-$] \textbf{Encryption query:} In an encryption query, $\mathcal{A}$ provides an index $\mathsf{ctr}^\ast_{\mathsf{reg}} \leq i \leq \mathsf{ctr}_{\mathsf{reg}}$, a message $\mu_{\mathsf{ctr}_{\mathsf{enc}}} \in \mathcal{M}$, and  {an attribute $X_{\mathsf{ctr}_{\mathsf{enc}}} \in \mathcal{U_{\tau}}$}. If no target key has been registered or if $P(X_{\mathsf{ctr}_{\mathsf{enc}}})=0$, the challenger replies with $\bot$. Otherwise, it increments $\mathsf{ctr}_{\mathsf{enc}} \gets \mathsf{ctr}_{\mathsf{enc}} + 1$ and computes $\mathsf{ct}_{\mathsf{ctr}_{\mathsf{enc}}} \gets \mathsf{Encrypt}(\mathsf{mpk}_i, X_{\mathsf{ctr}_{\mathsf{enc}}}, \mu_{\mathsf{ctr}_{\mathsf{enc}}})$. The challenger then replies with $(\mathsf{ctr}_{\mathsf{enc}}, \mathsf{ct}_{\mathsf{ctr}_{\mathsf{enc}}})$.

        \item[$-$] \textbf{{Decryption query:}}  In a decryption query, the adversary submits a ciphertext index $1 \leq j \leq \mathsf{ctr}_{\mathsf{enc}}$. The challenger computes $m'_j \gets$ $\mathsf{Decrypt}(\mathsf{sk}^\ast,$ $\mathsf{hsk}^\ast,$ ${X}_j,$ $\mathsf{ct}_j)$. If $m'_j = \mathsf{GetUpdate}$, it updates $\mathsf{hsk}^\ast \gets \mathsf{Update}(\mathsf{crs}, \mathsf{aux}, \mathsf{pk}^\ast)$ and recomputes $m'_j \gets \mathsf{Decrypt}(\mathsf{sk}^\ast, \mathsf{hsk}^\ast, {X}_j, \mathsf{ct}_j)$. If $m'_j \neq \mu_j$, the experiment halts with output $b=0$ (Decryption Game).

        \item[$-$] \textbf{Verification query:} In a verification query, the adversary submits a ciphertext index $j$. The challenger computes $\mathsf{cert}_j \gets \mathsf{Delete}(\mathsf{ct}_j)$ and $v_j \gets \mathsf{Verify}(\mathsf{vk}_j, \mathsf{cert}_j)$. If $v_j \neq 1$, the experiment halts with output $b=0$ (Verification Game).
    \end{description}
    If the adversary finishes making queries without halting the Decryption Game, the experiment outputs $b=1$. Similarly, if the Verification Game does not halt, the output is $b=1$.
\end{description}

\noindent We say that the scheme $\prod^{\mathsf{RABE}}_{\textsf{CD}}$ is correct and efficient if for all (possibly unbounded) $\mathcal{A}$ making at most polynomially many queries, the following hold:
\begin{description}\setlength\itemsep{0.5em}
    \item[$\bullet$] \textbf{Decryption correctness:} There exist a negligible function $\mathsf{negl}$ in $\lambda$ such that $\Pr[b = 0] = \mathsf{negl}(\lambda)$ in the Decryption Game.

    \item[$\bullet$] \textbf{Verification correctness:} There exists a negligible function $\mathsf{negl}$ in $\lambda$ such that $\Pr[b = 0] = \mathsf{negl}(\lambda)$ in the Verification Game.

    \item[$\bullet$] \textbf{Compactness:} Let $N$ be the number of registration queries. There exists a polynomial $\mathsf{poly}$ such that for all $i\in[N]$,
    \[
        |\mathsf{mpk}_i| = \mathsf{poly}(\lambda, |\mathcal{U}|, \log i), \quad
        |\mathsf{hsk}^\ast| = \mathsf{poly}(\lambda, |\mathcal{U}|, \log N).
    \]

    \item[$\bullet$] \textbf{Update efficiency:} Let $N$ denote the number of registration queries issued by $\mathcal{A}$ in the above game. During the game, the challenger invokes $\mathsf{Update}$ at most $\mathcal{O}(\log N)$ times, and each invocation runs in $\mathsf{poly}(\log N)$ time in the \textsf{RAM} model. In particular, we model $\mathsf{Update}$ as a \textsf{RAM} program with random access to its input, then its running time may be smaller than the input length.
\end{description}

\subsection{Certified Deletion Security of \textsf{RABE-CD}}\label{Security-def:RABE-CD}
Let $\prod^{\mathsf{RABE}}_{\textsf{CD}}$ denote a registered attribute-based encryption with certified deletion protocol defined over an attribute universe $\mathcal{U}$, a policy space $\mathcal{P}$, and a message space $\mathcal{M}$. For a security parameter $\lambda$, a bit $b \in \{0, 1\}$, and an adversary $\mathcal{A}$, we define the following security game $\mathsf{Exp}^{\textsf{RABE-CD}}_{\mathcal{A}}(\lambda, b)$ between $\mathcal{A}$ and a challenger:

\begin{description}\setlength\itemsep{0.5em}
    \item[$\bullet$] \textbf{Setup phase:} The challenger samples the common reference string $\mathsf{crs} \leftarrow \mathsf{Setup}(1^{\lambda}, 1^{{\tau}})$. It then initializes the internal state as follows: auxiliary input $\mathsf{aux} \leftarrow \perp$, master public key $\mathsf{mpk} \leftarrow \perp$, and counter $\mathsf{ctr} \leftarrow 0$ to track the number of honest key registration queries. It also initializes empty sets $C \leftarrow \emptyset, H \leftarrow \emptyset$ to record corrupted and honest public keys, respectively, and an empty dictionary $D \leftarrow \emptyset$ to map public keys to their associated attribute sets. For notational convenience, if $\mathsf{pk} \notin D$, we define $D[\mathsf{pk}] :=\emptyset$. Finally, the challenger sends $\mathsf{crs}$ to the adversary $\mathcal{A}$.

    \item[$\bullet$] \textbf{Query phase:} The adversary $\mathcal{A}$ can issue the following queries: \vspace{0.2cm}

    \begin{description}\setlength\itemsep{0.5em}
        \item[$-$] \textbf{Register corrupted key query:} $\mathcal{A}$ submits a public key $\mathsf{pk}$ along with a policy $P \in \mathcal{P}$. The challenger then executes
        \begin{align*}
            &(\mathsf{mpk}', \mathsf{aux}') \leftarrow\mathsf{RegPK}(\mathsf{crs},\mathsf{aux},\mathsf{pk},P)
        \end{align*}
        and updates the following
        \begin{align*}
            &\mathsf{mpk} \leftarrow \mathsf{mpk}',~C \leftarrow C \cup \{\mathsf{pk}\}\\
            &\mathsf{aux} \leftarrow \mathsf{aux}',~D[\mathsf{pk}] \leftarrow D[\mathsf{pk}] \cup \{P\}.
        \end{align*}
        Finally, the challenger replies to $\mathcal{A}$ with $(\mathsf{mpk}', \mathsf{aux}')$.

        \item[$-$] \textbf{Register honest key query:} In an honest-key-registration query, the adversary specifies a policy $P \in \mathcal{P}$. The challenger increments the counter $\mathsf{ctr} \leftarrow \mathsf{ctr} + 1$, samples a key pair $(\mathsf{pk}_{\mathsf{ctr}}, \mathsf{sk}_{\mathsf{ctr}}) \leftarrow \mathsf{KeyGen}(\mathsf{crs}, \mathsf{aux}, P)$, and registers it by computing $(\mathsf{mpk}', \mathsf{aux}')$ $ \leftarrow$ $\mathsf{RegPK}(\mathsf{crs},$ $ \mathsf{aux},$ $\mathsf{pk}_{\mathsf{ctr}}, P)$. The challenger then updates its state as follows:
        \begin{align*}
            &\mathsf{mpk} \leftarrow \mathsf{mpk}',~H \leftarrow H \cup \{\mathsf{pk}_{\mathsf{ctr}}\},\\
            &\mathsf{aux} \leftarrow \mathsf{aux}',~D[\mathsf{pk}_{\mathsf{ctr}}] \leftarrow D[\mathsf{pk}_{\mathsf{ctr}}] \cup \{P\}.
        \end{align*}
        Finally, it replies to the adversary with $(\mathsf{ctr}, \mathsf{mpk}', \mathsf{aux}', \mathsf{pk}_{\mathsf{ctr}})$.

        \item[$-$] \textbf{Corrupt honest key query:} $\mathcal{A}$ specifies an index $1 \leq i \leq \mathsf{ctr}$. Let $(\mathsf{pk}_i, \mathsf{sk}_i)$ be the key pair generated during the $i$-th honest key query. The challenger updates its state by removing $\mathsf{pk}_i$ from $H$ and adding it to $C$:
        \begin{align*}
            &H \gets H \setminus \{\mathsf{pk}_i\}, ~C \gets C \cup \{\mathsf{pk}_i\}.
        \end{align*}
        It then returns $\mathsf{sk}_i$ to $\mathcal{A}$.
    \end{description}
    
    \item[$\bullet$] \textbf{Challenge phase:} In the challenge phase, $\mathcal{A}$ chooses two messages $\mu_0^*, \mu_1^* \in \mathcal{M}$ together with an {attribute $X^* \in \mathcal{U_{\tau}}$ } . Then it sends the tuple ($\mu_0^*, \mu_1^*, X^\ast$) to the challenger. The challenger first verifies that 
    \[
        P(X^\ast) = 0 \quad \text{for all } P \in \{D[\mathsf{pk}] : \mathsf{pk} \in C\}.
    \]
    Subsequently, the challenger selects a random bit $b\in\{0,1\}$ and executes the encryption algorithm
    \[
        (\mathsf{vk}^*, \mathsf{ct}_b^*) \gets \mathsf{Encrypt}(\mathsf{mpk}, X^*, \mu_b^*),
    \]
    and returns $\mathsf{ct}_b^*$ to $\mathcal{A}$. 
    
    \item[$\bullet$] \textbf{Deletion phase:} At some point, the adversary $\mathcal{A}$ produces a deletion certificate $\mathsf{cert}$ and submits it to the challenger. The challenger verifies its validity by computing
    \[
        \mathsf{valid} \gets \mathsf{Verify}(\mathsf{vk}^*, \mathsf{cert}).
    \]
    If $\mathsf{valid} = 0$, the challenger aborts the experiment and outputs $\perp$.
    
    If $\mathsf{valid} = 1$, the challenger discloses the following information:
    \begin{description}
        \item[$-$] All honest secret keys $\{\mathsf{sk}_i : \mathsf{pk}_i \in H\},$
        \item[$-$] All helper keys generated via the update algorithm $$\{\mathsf{hsk}_i :\mathsf{hsk}_i \gets \mathsf{Update}(\mathsf{crs},\mathsf{aux},\mathsf{pk}_i),~ \mathsf{pk}_i \in H\}.$$
    \end{description}
    \item[$\bullet$] \textbf{Output phase:} The adversary outputs a bit $b' \in \{0,1\}$.
\end{description}

\noindent We say that $\prod^{\mathsf{RABE}}_{\textsf{CD}}$ is secure if for all efficient and admissible \textsf{QPT} adversaries $\mathcal{A}$, there exists a negligible function $\mathsf{negl}$ such that for all $\lambda \in \mathbb{N}$
\[
\left| \Pr\left[\mathsf{Exp}^{\mathsf{RABE-CD}}_{ \mathcal{A}}(\lambda, 0) = 1\right] - \Pr\left[\mathsf{Exp}^{\mathsf{RABE-CD}}_{\mathcal{A}}(\lambda, 1) = 1\right] \right| \leq \mathsf{negl}(\lambda).
\]
\vspace{0.2cm}

\noindent Depending on the nature of the verification key, the \textsf{RABE-CD} framework can be classified into two distinct variants: registered attribute-based encryption with privately verifiable certified deletion (\textsf{RABE-PriVCD}) and registered attribute-based encryption with publicly verifiable certified deletion (\textsf{RABE-PubVCD}). \vspace{0.2cm}

\noindent\textbf{Privately Verifiable Certified Deletion.} The syntax, correctness, and security requirements of \textsf{RABE-PriVCD} follow directly from the definitions of \textsf{RABE-CD} provided in Section \ref{def:RABE-CD}. It is important to note that the verification of a deletion certificate is restricted to the sender, who is the sole possessor of the verification key. The security experiment for \textsf{RABE-PriVCD}, denoted by $\mathsf{Exp}^{\textsf{RABE-PriVCD}}{\mathcal{A}}(\lambda, b)$, is identical to the experiment $\mathsf{Exp}^{\textsf{RABE-CD}}{\mathcal{A}}(\lambda, b)$. \vspace{0.2cm}

\noindent\textbf{Publicly Verifiable Certified Deletion.} The syntax and correctness of \textsf{RABE-PubVCD} follow the general framework of \textsf{RABE-CD}, described in Section~\ref{def:RABE-CD}, with the only modification in the encryption algorithm $\mathsf{Encrypt}$. In this setting, encryption is realized through an interactive protocol between a sender $\mathsf{Sndr}$ and a receiver $\mathsf{Rcvr}$. On input $(\mathsf{mpk}, X, \mu)$ to $\mathsf{Sndr}$, where $\mathsf{mpk}$ denotes the master public key, $X \subseteq \mathcal{U}$ stands an attribute set and $\mu$ the message, and no input to $\mathsf{Rcvr}$, the protocol outputs a ciphertext and an associated verification key as follows:
\[
(\mathsf{vk}, \mathsf{ct}) \leftarrow \mathsf{Encrypt}\bigl(\mathsf{Sndr}(\mathsf{mpk}, X, \mu),\, \mathsf{Rcvr}\bigr).
\]  

\noindent It is important to note that in the publicly verifiable model, anyone can check whether a deletion certificate in the verification algorithm is valid or not. \vspace{0.2cm}

\noindent The security of \textsf{RABE-PubVCD} is evaluated using the same security game as $\mathsf{Exp}^{\textsf{RABE-CD}}_{\mathcal{A}}(\lambda, b)$. The defining characteristic of the \textsf{RABE-PubVCD} security model is its resilience against an adversary who possesses the verification key \textsf{vk}. In other words, the confidentiality of the system remains intact even if \textsf{vk} is made publicly available. In this context, the security game for \textsf{RABE-PubVCD} is referred to as $\mathsf{Exp}^{\textsf{RABE-PubVCD}}_{\mathcal{A}}(\lambda, b)$. \vspace{0.2cm}


\noindent {In the subsequent subsection, we consider the realization of ciphertext deletion functionalities within the \textsf{RABE} framework under certified everlasting security. This notion extends certified deletion by providing information-theoretic security guarantees after deletion, however, without revealing honest secret keys in the post-deletion phase.}.

\subsection{System Model for \textsf{RABE-CED}}\label{def:syntax-rabe-ce-security} 
The syntax and correctness conditions of \textsf{RABE} with Certified Everlasting Deletion (\textsf{RABE-CED}) follow identically from those of \textsf{RABE-CD}, as defined in Sections~\ref{def:RABE-CD} and~\ref{def:rabe-cd-correctness}. The only substantive difference lies in the notion of security. In \textsf{RABE-CED}, the adversary is assumed to be polynomial-time during protocol execution but may become computationally unbounded afterwards. Even with such post-execution power, any message whose deletion has been certified remains secure. In the following section, we now formalize the certified everlasting deletion security in the \textsf{RABE} framework. 

\subsection{Certified Everlasting Security for \textsf{RABE-CED}}\label{def:rabe-ce-security} 

Let $\prod^{\mathsf{RABE}}_{\textsf{CED}}$ denote a registered attribute-based encryption scheme with certified everlasting deletion framework, defined over an attribute universe $\mathcal{U}$, a policy space $\mathcal{P}$, and a message space $\mathcal{M}$. Note that the syntax is same as $\prod^{\mathsf{RABE}}_{\textsf{CD}}$. For a security parameter $\lambda$, a challenge bit $b \in \{0,1\}$, and an adversary $\mathcal{A}$, we formalize the certified everlasting security experiment $\mathsf{Exp}^{\textsf{RABE-CED}}_{\mathcal{A}}(\lambda, b)$, played between $\mathcal{A}$ and a challenger, as follows.\vspace{0.2cm}

\begin{description}\setlength\itemsep{0.5em}
    \item[$\bullet$] \textbf{\textsf{Setup:}} The challenger samples $\mathsf{crs} \leftarrow \mathsf{Setup}(1^{\lambda}, 1^{{\tau}})$, and initializes the internal state as follows:
    \begin{align*}
        &\mathsf{aux} \leftarrow \perp,~C \leftarrow \emptyset,~ D \leftarrow \emptyset,\\
        &\mathsf{mpk} \leftarrow \perp,~\mathsf{ctr} \leftarrow 0, ~H \leftarrow \emptyset,
    \end{align*}
    where $C$ maintains the set of corrupted public keys, $H$ maintains the set of honest public keys, and $D$ records the attribute sets registered to each public key. Finally, the challenger provides $\mathsf{crs}$ to the adversary $\mathcal{A}$.

    \item[$\bullet$] \textbf{\textsf{Query phase:}} The adversary $\mathcal{A}$ can issue the following queries:\vspace{0.2cm}
      
    \begin{itemize}\setlength\itemsep{0.5em}
        \item[$-$] \textbf{Register Corrupted Key:} Upon receiving a submission $(\mathsf{pk}, P \in \mathcal{P})$ from $\mathcal{A}$, the challenger computes
        \begin{align*}
            &(\mathsf{mpk}', \mathsf{aux}') \leftarrow \mathsf{RegPK}(\mathsf{crs}, \mathsf{aux}, \mathsf{pk}, P),
        \end{align*}
        updates its internal state as follows:
        \begin{align*}
            &\mathsf{mpk} \leftarrow \mathsf{mpk}', ~ D[\mathsf{pk}] \leftarrow D[\mathsf{pk}] \cup \{P\}\\
            &C \leftarrow C \cup \{\mathsf{pk}\}, ~\mathsf{aux} \leftarrow \mathsf{aux}',
        \end{align*}
        and returns $(\mathsf{mpk}', \mathsf{aux}')$ to $\mathcal{A}$.
        
        \item[$-$] \textbf{Register Honest Key:} When $\mathcal{A}$ specifies a policy $P \in \mathcal{P}$, the challenger increments the counter $\mathsf{ctr} \leftarrow \mathsf{ctr} + 1$, generates a fresh key pair
        \begin{align*}
            &(\mathsf{pk}_{\mathsf{ctr}}, \mathsf{sk}_{\mathsf{ctr}}) \leftarrow \mathsf{KeyGen}(\mathsf{crs}, \mathsf{aux}, P)
        \end{align*}
        and executes the registration procedure
        \begin{align*}
            &(\mathsf{mpk}', \mathsf{aux}') \leftarrow \mathsf{RegPK}(\mathsf{crs}, \mathsf{aux}, \mathsf{pk}_{\mathsf{ctr}}, P).
        \end{align*}
        The challenger then updates its states:
        \begin{align*}
            &\mathsf{mpk} \leftarrow \mathsf{mpk}',~D[\mathsf{pk}_{\mathsf{ctr}}] \leftarrow D[\mathsf{pk}_{\mathsf{ctr}}] \cup \{P\},\\
            &\mathsf{aux} \leftarrow \mathsf{aux}',~H \leftarrow H \cup \{\mathsf{pk}_{\mathsf{ctr}}\},
        \end{align*}
        and responds with $(\mathsf{ctr}, \mathsf{mpk}', \mathsf{aux}', \mathsf{pk}_{\mathsf{ctr}})$.
        
        \item[$-$] \textbf{Corrupt Honest Key:} On input $i \in [1,\mathsf{ctr}]$, the challenger marks $\mathsf{pk}_i$ as corrupted:
        \begin{align*}
            &C \leftarrow C \cup \{\mathsf{pk}_i\},~H \leftarrow H \setminus \{\mathsf{pk}_i\},
        \end{align*}
        and outputs the corresponding secret key $\mathsf{sk}_i$.
    \end{itemize}

    \item[$\bullet$] \textbf{\textsf{Challenge Phase:}} The adversary $\mathcal{A}$ submits a challenge  an {attribute $X^* \in \mathcal{U_{\tau}}$ } . The challenger verifies that $P(X^*) = 0$ for all $P\in S$, where $S=$ $\{P \in D[\mathsf{pk}]: \mathsf{pk} \in C\}$ and computes
    \begin{align*}
        &(\mathsf{vk}^*, \mathsf{ct_b}^*) \leftarrow \mathsf{Encrypt}(\mathsf{mpk}, X^*, b),
    \end{align*}
    and returns $\mathsf{ct_b}^*$ to the adversary.     

    \item[$\bullet$] \textbf{\textsf{Output Phase:}} The adversary $\mathcal{A}$ produces a deletion certificate $\mathsf{cert}$ together with a residual state $\rho$. The challenger verifies the certificate by computing
    \begin{align*}
        &\mathsf{valid} \leftarrow \mathsf{Verify}(\mathsf{vk}^*, \mathsf{cert})
    \end{align*}
    If $\mathsf{valid} = 0$, the experiment terminates with output $\perp$. Otherwise, if $\mathsf{valid} = 1$, the experiment concludes with output $\rho$.    
\end{description}

\noindent We say that $\prod^{\mathsf{RABE}}_{\textsf{CED}}$ is secure if for all \textsf{QPT} adversaries $\mathcal{A}$, there exists a negligible function $\mathsf{negl}(\cdot)$ such that
\begin{align*}
    &\mathsf{TD}\left(\mathsf{Exp}^{\textsf{RABE-CED}}_{\mathcal{A}}(\lambda, 0), \mathsf{Exp}^{\textsf{RABE-CED}}_{\mathcal{A}}(\lambda, 1)\right) \leq \mathsf{negl}(\lambda),
\end{align*}
where $\mathsf{TD}(\cdot,\cdot)$ denotes the trace distance.
\vspace{0.2cm}

\noindent The experiment considers the challenge messages restricted to a single bit. The definition naturally extends to multi-bit messages by encrypting each bit independently and concatenating the resulting ciphertexts, where indistinguishability is preserved through a standard hybrid argument. \vspace{0.2cm}

\noindent In line with the \textsf{RABE-CD} framework, the \textsf{RABE-CED} model also encompasses two variants, distinguished by the nature of the verification key: registered attribute-based encryption with privately verifiable certified everlasting deletion (\textsf{RABE-PriVCED}) and registered attribute-based encryption with publicly verifiable certified everlasting deletion (\textsf{RABE-PubVCED}). \vspace{0.2cm}

\noindent\textbf{Privately Verifiable Certified Everlasting Deletion.} The syntax and correctness of \textsf{RABE-PriVCED} can be directly derived from the general framework of \textsf{RABE-CD}, as specified in Section \ref{def:RABE-CD}. In addition, the security model for \textsf{RABE-PriVCED} is also identical to the experiment $\mathsf{Exp}^{\textsf{RABE-CED}}_{\mathcal{A}}(\lambda, b)$. \vspace{0.2cm}

\noindent\textbf{Publicly Verifiable Certified Everlasting Deletion.} The syntax and correctness of \textsf{RABE-PubVCED} also follow directly from the general framework of \textsf{RABE-CD}, as introduced in Section \ref{def:RABE-CD}. Its security model coincides with the experiment $\mathsf{Exp}^{\textsf{RABE-CED}}_{\mathcal{A}}(\lambda, b)$, except that during the challenge phase, the adversary additionally receives $(\mathsf{vk}^*, \mathsf{ct}_b^*)$ to enable public verifiability. \vspace{0.2cm}

\noindent To realize the \textsf{RABE-CD} scheme under private and public verification settings, we need to introduce a fundamental cryptographic primitive termed Shadow Registered Attribute-Based Encryption (\textsf{Shad-RABE}). This primitive serves as the core building block that underpins our certified deletion constructions. In the subsequent section, we formalize the syntax and security of \textsf{Shad-RABE} and present its construction. \vspace{0.2cm}

%% file: 4_Shadow-RABE.tex
\section{Shadow Registered Attribute-Based Encryption}\label{section:shad-RABE}
In this section, we present the construction of our generic (key-policy) \textsf{Shad-RABE} protocol by composing registered attribute-based encryption, zero-knowledge argument and indistinguishability obfuscation in a modular manner. We formally analyze its security to establish that the resulting construction satisfies the desired privacy and correctness properties. Furthermore, we present a lattice-based instantiation of \textsf{Shad-RABE}, ensuring post-quantum security of the framework. \vspace{0.2cm}

\subsection{System Model of \textsf{Shad-RABE}}\label{def:shadow-rabe}

A (key-policy) shadow registered attribute-based encryption (\textsf{Shad-RABE}) scheme with attribute universe $\mathcal{U}$, policy space $\mathcal{P}$, and message space $\mathcal{M}$ is a tuple of efficient algorithms $\prod_{\mathsf{Shad}\text{-}\mathsf{RABE}} =$ $ (\mathsf{Setup},$ $\mathsf{KeyGen},$ $\mathsf{RegPK},$ $\mathsf{Encrypt},$ $\mathsf{Update}, \mathsf{Decrypt},$ $\mathsf{SimKeyGen},$ $\mathsf{SimRegPK},$ $\mathsf{SimCT},$ $\mathsf{SimCorrupt},$ $\mathsf{Reveal})$ with the following syntax:

\begin{description}\setlength\itemsep{0.5em}
\item$\bm{-}$~{(\textsf{crs})} $\bm{\leftarrow}$ ${\textsf{Setup}}{(1^{\lambda}, 1^{{\tau}}).}$ On input of the length of the security parameter $\lambda$ and a {policy-family parameter $\tau$} , the setup algorithm outputs a common reference string $\mathsf{crs}$.
\item$-$~{(\textsf{pk}, \textsf{sk})} $\leftarrow$ ${\textsf{KeyGen}(\textsf{crs}, \textsf{aux}, P).}$ Given the common reference string \textsf{crs}, auxiliary state $\mathsf{aux}$, and a policy $P\in \mathcal{P},$ this algorithm produces a public/secret key pair $(\mathsf{pk}, \mathsf{sk})$.
\item$\bm{-}$~$(\mathsf{mpk}, \mathsf{aux}') \leftarrow \textsf{RegPK}(\mathsf{crs}, \mathsf{aux}, \mathsf{pk}, P).$ The registration algorithm takes as input a common reference string \textsf{crs}, an auxiliary state \textsf{aux}, a public key $\mathsf{pk},$ and a policy $P\in \mathcal{P}$ to return a master public key $\mathsf{mpk}$ and an updated auxiliary state $\mathsf{aux}'$.
\item$-$ $(\mathsf{ct})\leftarrow \textsf{Encrypt}(\mathsf{mpk}, X, \mu).$ The sender encrypts a message $\mu \in \mathcal{M}$ under the master public key $\mathsf{mpk}$ and { attribute $X \in \mathcal{U_{\tau}}$}, producing a ciphertext $\mathsf{ct}$.
\item$-$ $(\mathsf{hsk})\leftarrow \textsf{Update}(\mathsf{crs}, \mathsf{aux}, \mathsf{pk}).$ This algorithm outputs a helper secret key $\mathsf{hsk}$ given the common reference string $\mathsf{crs}$, an auxiliary state $\mathsf{aux}$, and a registered public key $\mathsf{pk}$ as inputs.
\item$-$ $(\mu \cup \{\perp, \mathsf{GetUpdate}\})\leftarrow \textsf{Decrypt}(\mathsf{ct}, \mathsf{sk}, \mathsf{hsk},{X}).$ The decryptor executes this algorithm using the ciphertext $\mathsf{ct}$, the secret key $\mathsf{sk}$, { $X \in \mathcal{U_{\tau}}$} and the helper secret key $\mathsf{hsk}$ as input. The algorithm outputs the message $\mu$, the failure symbol $\perp$, or the instruction $\mathsf{GetUpdate}$ indicating that an updated helper secret key is required.
\item$-$ $({\mathsf{pk}}, \mathsf{B'})\leftarrow \textsf{SimKeyGen}(\mathsf{crs}, \mathsf{aux}, P, \mathsf{B}).$ Given the common reference string $\mathsf{crs}$, an auxiliary state $\mathsf{aux}$, a policy $P\in\mathcal{P}$, and an ancillary dictionary $\mathsf{B}$ (possibly empty) indexed by public keys, the simulated key generation algorithm outputs a public key {${\mathsf{pk}}$ }together with an updated ancillary dictionary $\mathsf{B'}$.
\item$-$ $(\widetilde{\mathsf{mpk}}, \widetilde{\mathsf{aux}})\leftarrow \textsf{SimRegPK}(\mathsf{crs}, \mathsf{aux}, \mathsf{pk}, P, \textsf{B}).$ The simulated key registration algorithm, given the common reference string $\mathsf{crs}$, an auxiliary state $\mathsf{aux}$, a public key $\mathsf{pk},$ and a policy $P$, produces a simulated master public key $\widetilde{\mathsf{mpk}}$ and an updated auxiliary state $\widetilde{\mathsf{aux}}$.
\item$-$ $(\widetilde{\mathsf{sk}})\leftarrow \textsf{SimCorrupt}({\mathsf{crs}}, \mathsf{pk}, \mathsf{B}).$ The simulated corruption algorithm takes as input the common reference string ${\mathsf{crs}}$, the public key $\mathsf{pk}$ and an ancillary dictionary $\mathsf{B}$ to return a simulated secret key $\widetilde{\mathsf{sk}}$ corresponding to $\mathsf{pk}$.
\item$-$ $(\widetilde{\mathsf{ct}})\leftarrow \textsf{SimCT}(\widetilde{\mathsf{mpk}},  \mathsf{B}, X).$ The simulated ciphertext generation algorithm simulates the entire encryption algorithm under the {attribute  $X \in \mathcal{U_{\tau}}$ } using the simulated master public key $\widetilde{\mathsf{mpk}}$, and {an} ancillary dictionary $\mathsf{B}$, and outputs a simulated ciphertext $\widetilde{\mathsf{ct}}$.
\item$-$ $(\widetilde{\mathsf{sk}})\leftarrow \textsf{Reveal}(\mathsf{pk}, \mathsf{B}, \widetilde{\mathsf{ct}}, \mu).$ Given a public key $\mathsf{pk},$ a dictionary $\mathsf{B}$, a simulated ciphertext $\widetilde{\mathsf{ct}}$, and a target message $\mu$, this algorithm generates a simulated secret key $\widetilde{\mathsf{sk}}$.
\end{description}

\noindent\textbf{Correctness and Efficiency of \textsf{Shad-RABE}.}\label{def:shadow-rabe-properties} The \textsf{Shad-RABE} protocol preserves the same correctness and efficiency guarantees as the underlying registered \textsf{ABE} scheme. 
\vspace{0.2cm}

\noindent\textbf{Security of \textsf{Shad-RABE}.}\label{def:shadow-rabe-security} Let a (key-policy) shadow registered attribute-based encryption scheme $\prod_{\textsf{Shad}\text{-}\textsf{RABE}}$ be defined over an attribute universe $\mathcal{U}$, policy space $\mathcal{P}$ and message space $\mathcal{M}$. For a security parameter $\lambda$, a bit $b \in \{0, 1\}$, and an adversary $\mathcal{A}$, we define the following security game $\mathsf{Exp}^{\textsf{Shad-RABE}}_{\mathcal{A}}(\lambda, b)$ between $\mathcal{A}$ and a challenger:

\begin{description}\setlength\itemsep{0.5em}
   \item[$\bullet$ \textsf{Setup.}] The challenger begins by generating a common reference string $\mathsf{crs} \leftarrow \mathsf{Setup}(1^{\lambda}, 1^{|\mathcal{U}|})$. It initializes its internal variables as follows: the auxiliary state $\mathsf{aux} \leftarrow \perp$, the master public key $\mathsf{mpk} \leftarrow \perp$, and a counter $\mathsf{ctr} \leftarrow 0$ to track the number of honestly registered keys. Additionally, two empty sets $C \leftarrow \emptyset$ and $H \leftarrow \emptyset$ are initialized to store corrupted and honest public keys, respectively. The challenger also initializes empty dictionaries, $D \leftarrow \emptyset$ and $\mathsf{B}$$ \leftarrow \emptyset$, indexed by the public keys. For ease of notation, it is assumed that if $\mathsf{pk} \notin D$, then $D[\mathsf{pk}] := \emptyset$ and same for $\mathsf{B}$. Finally, the challenger sends the common reference string $\mathsf{crs}$ to the adversary $\mathcal{A}$.
    \item[$\bullet$ \textsf{Query Phase:}] During this phase, the adversary $\mathcal{A}$ is allowed to make the following types of queries:\vspace{0.2cm}
    
    \begin{itemize}
        \item[$-$] \textbf{Register corrupted key query.} Upon receiving a public key $\mathsf{pk}$ and a policy $P \in\mathcal{P}$ from $\mathcal{A}$, the challenger proceeds as follows:   \[
        (\mathsf{mpk}', \mathsf{aux}') \leftarrow 
        \begin{cases}
            \mathsf{RegPK}(\mathsf{crs}, \mathsf{aux}, \mathsf{pk}, P), & \text{if } b = 0 \\
            \mathsf{SimRegPK}(\mathsf{crs}, \mathsf{aux}, \mathsf{pk}, P, \textsf{B}), & \text{if } b = 1
        \end{cases}
        \]
        It then updates its internal state as:
        \begin{align*}
            &\mathsf{mpk} \leftarrow \mathsf{mpk}', \mathsf{aux} \leftarrow \mathsf{aux}', C \leftarrow C \cup \{\mathsf{pk}\}, D[\mathsf{pk}] \leftarrow D[\mathsf{pk}] \cup \{P\}.
        \end{align*}
        The challenger returns $(\mathsf{mpk}', \mathsf{aux}')$ to $\mathcal{A}$.
        \item[$-$] \textbf{Register honest key query:} When $\mathcal{A}$ specifying a policy $P \in \mathcal{P}$, the challenger increments its counter $\mathsf{ctr} \leftarrow \mathsf{ctr} + 1$, and performs the following based on the bit $b$:
        \[
        \left.
        \begin{aligned}
            ~~~(\mathsf{pk}_{\mathsf{ctr}}, \mathsf{sk}_{\mathsf{ctr}}) &\leftarrow \mathsf{KeyGen}(\mathsf{crs}, \mathsf{aux},P) \\
            ~~~(\mathsf{mpk}', \mathsf{aux}') &\leftarrow \mathsf{RegPK}(\mathsf{crs}, \mathsf{aux}, \mathsf{pk}_{\mathsf{ctr}}, P)
        \end{aligned}
        \right\}~~~\text{if}~~ b = 0 
        \]
        \[
        \left.
        \begin{aligned}
            (\mathsf{pk}_{\mathsf{ctr}}, \mathsf{B'}) &\leftarrow \mathsf{SimKeyGen}(\mathsf{crs}, \mathsf{aux},\mathsf{B}) \\
            (\mathsf{mpk}', \mathsf{aux}') &\leftarrow \mathsf{SimRegPK}(\mathsf{crs}, \mathsf{aux}, \mathsf{pk}_{\mathsf{ctr}}, P,\mathsf{B})
        \end{aligned}
        \right\}~~~\text{if}~~ b = 1. 
        \]

        It then updates:
        \begin{align*}
        &\mathsf{mpk} \leftarrow \mathsf{mpk}', \mathsf{aux} \leftarrow \mathsf{aux}', H \leftarrow H \cup \{\mathsf{pk}_{\mathsf{ctr}}\}, D[\mathsf{pk}_{\mathsf{ctr}}] \leftarrow D[\mathsf{pk}_{\mathsf{ctr}}] \cup \{P\}, {\mathsf{B}\leftarrow\mathsf{B'}}.
        \end{align*}
        It then replies with $(\mathsf{ctr}, \mathsf{mpk}', \mathsf{aux}', \mathsf{pk}_{\mathsf{ctr}})$.
        \item[$-$] \textbf{Corrupt honest key query:} The adversary may request the corruption of an honest key by specifying an index $i$ such that $1 \leq i \leq \mathsf{ctr}$. The challenger responds based on the bit $b$:\vspace{0.15cm}
        
        \noindent(a) If $b = 0$, it returns the actual secret key $\mathsf{sk}_i$.\vspace{0.15cm}
        
        \noindent(b) If $b = 1$, it computes a simulated secret key $\widetilde{\mathsf{sk}}_i$ $\leftarrow$ \textsf{SimCorrupt}($\mathsf{crs},$ $\mathsf{pk}_i,$ $\mathsf{B}$).\vspace{0.15cm}
        
        In either case, $\mathsf{pk}_i$ is removed from the honest set $H$ and added to the corrupted set $C$.
    \end{itemize}
\item[$\bullet$ \textsf{Challenge phase:}] The adversary $\mathcal{A}$ selects a message $\mu \in \mathcal{M}$ and an {attribute  $X^* \in \mathcal{U_{\tau}}$ } such that $P(X^\ast) = 0$ for all $P \in S,$ where $S=\{P\in D[\mathsf{pk}] \mid \mathsf{pk} \in C\}$. The challenger then prepares the challenge ciphertext as follows:
\begin{align*}
\mathsf{ct}^* &\leftarrow 
\begin{cases}
    \mathsf{Encrypt}(\mathsf{mpk}, X^*, \mu) & \text{if } b = 0, \\
    \mathsf{SimCT}(\mathsf{mpk}, \mathsf{aux}, \mathsf{B}, X^*) & \text{if } b = 1.
\end{cases}
\end{align*}

For every honest public key $\mathsf{pk}_i \in H$, the challenger generates the corresponding secret and helper keys as follows:
\begin{align*}
\mathsf{sk}_i^* &\leftarrow 
\begin{cases}
    \mathsf{sk}_i & \text{if } b = 0, \\
    \mathsf{Reveal}(\mathsf{pk}_i,\mathsf{B},  \mathsf{ct}^*, \mu) & \text{if } b = 1,
\end{cases}\\
\mathsf{hsk}_i &\leftarrow \mathsf{Update}(\mathsf{crs}, \mathsf{aux}, \mathsf{pk}_i).
\end{align*}
The challenger then returns to the adversary the tuple: ($\mathsf{ct}^*$, $\{\mathsf{sk}_i^*\}_{\mathsf{pk}_i \in H}$, $\{\mathsf{hsk}_i\}_{\mathsf{pk}_i \in H}$). 
\item [$\bullet$ \textsf{Output phase:}] The adversary finally outputs a bit $b' \in \{0, 1\}$.
\end{description}
\noindent A \textsf{Shad-RABE} scheme is said to be secure if for every efficient quantum polynomial-time (\textsf{QPT}) adversary $\mathcal{A}$, there exists a negligible function $\mathsf{negl}(\cdot)$ such that for all $\lambda \in \mathbb{N}$,
\[
\left| \Pr\left[\mathsf{Exp}^{\textsf{Shad-RABE}}_{\mathcal{A}}(\lambda, 0) = 1\right] - \Pr\left[\mathsf{Exp}^{\textsf{Shad-RABE}}_{\mathcal{A}}(\lambda, 1) = 1\right] \right| \leq \mathsf{negl}(\lambda).
\]
\subsection{Generic Construction of \textsf{Shad-RABE} Protocol}\label{generic-construct:Shadow-RABE}
In this section, we construct a shadow registered attribute-based encryption (\textsf{Shad-RABE}) scheme with plaintext space $\{0,1\}^{\ell_m}$, attribute universe $\mathcal{U}$, and policy space $\mathcal{P}$. Our construction builds upon the following cryptographic components: 
\begin{description}
    \item[$-$] A registered attribute-based encryption scheme with the tuple of algorithms $\prod_{\textsf{RABE}}=\textsf{RABE.}$($\mathsf{Setup},$ $\mathsf{KeyGen},$ $\mathsf{RegPK},$ $\mathsf{Encrypt},$ $\mathsf{Update},$ $\mathsf{Decrypt})$ defined over the same $\mathcal{U}$ and $\mathcal{P}$.
    \item[$-$] An interactive proof system given by the interactive machines $\langle \mathcal{P},\mathcal{V}\rangle$ for an \textsf{NP} language $\mathcal{L}$ characterized by the relation $R$ defined as:
    \begin{align*}
        R := &\left\{ \left((\{\mathsf{mpk}_{i,0}, \mathsf{mpk}_{i,1}\}_{i \in [\ell_m]}, \{\mathsf{CT}_{i,0}, \mathsf{CT}_{i,1}\}_{i \in [\ell_m]}, X), \{(\mu[i], r_{i,0}, r_{i,1})\}_{i \in [\ell_m]} \right) \right.\\
        &\left.\mid \forall i \in [\ell_m],\ \forall b \in \{0,1\},\ \mathsf{CT}_{i,b} = \mathsf{RABE}.\mathsf{Encrypt}(\mathsf{mpk}_{i,b}, X, \mu[i]; r_{i,b}) \right\}.
    \end{align*}
    Each $r_{i,b}$ is sampled uniformly from the randomness space of the $\mathsf{RABE}.\mathsf{Encrypt}$ algorithm. 
    \item[$-$] An indistinguishability obfuscator $i\mathcal{O}$ for polynomial-size circuits.
\end{description}
 
\noindent The algorithms of \textsf{Shad-RABE} are detailed below.

\begin{description}\setlength\itemsep{0.5em}
\item[]\textsf{Setup}($1^{\lambda}, 1^{{\tau}}$):
  \leavevmode
  \begin{itemize}
    \item[$-$] For every $i \in [\ell_m]$ and $b \in \{0,1\}$, generate $\mathsf{crs}_{i,b} \leftarrow \mathsf{RABE}.\mathsf{Setup}(1^{\lambda}, 1^{{\tau}})$.
    \item[$-$] Sample $y \leftarrow \{0,1\}^{*}$ uniformly at random.
    \item[$-$] Output $\mathsf{crs} := (\{\mathsf{crs}_{i,b}\}_{i \in [\ell_m],\ b \in \{0,1\}},y)$.
    
  \end{itemize}

\item[]\textsf{KeyGen}(\textsf{crs}, \textsf{aux}, $P$):
\leavevmode
\begin{itemize}
    \item[$-$] Parse $\mathsf{crs} = (\{\mathsf{crs}_{i,b}\}_{i,b},y)$.
    \item[$-$] Sample $z \leftarrow \{0,1\}^{\ell_m}$ uniformly at random.
    \item[$-$] For all $i \in [\ell_m]$, $b \in \{0,1\}$, generate $(\mathsf{pk}_{i,b}, \mathsf{sk}_{i,b}) \leftarrow \mathsf{RABE}.\mathsf{KeyGen}(\mathsf{crs}_{i,b},$ $\mathsf{aux}_{i,b},P)$.
    \item[$-$] Set $\mathsf{pk} := (\{\mathsf{pk}_{i,b}\}_{i,b})$.
    \item[$-$] Construct the secret key $\mathsf{sk} := iO(\mathcal{D}[\{\mathsf{sk}_{i,z[i]}\}_{i}, z])$ where the circuit is defined in Fig. \ref{Figure:UpdatedCircuit-1}.
    \item[$-$] Output $(\mathsf{pk}, \mathsf{sk})$.
\end{itemize}

\item[]\textsf{RegPK}(\textsf{crs}, \textsf{aux}, \textsf{pk}, $P$):
\leavevmode
\begin{itemize}
    \item[$-$] Parse $\mathsf{crs} = (\{\mathsf{crs}_{i,b}\}_{i,b},y)$ and $\mathsf{pk} = (\{\mathsf{pk}_{i,b}\}_{i,b})$.
    \item[$-$] If $\mathsf{aux} = \perp$, initialize $\mathsf{aux}_{i,b} := \perp$ for all $i, b$.
    \item[$-$] For each $i, b$, compute $(\mathsf{mpk}_{i,b}, \mathsf{aux}'_{i,b}) \leftarrow \mathsf{RABE}.\mathsf{RegPK}(\mathsf{crs}_{i,b},$ $\mathsf{aux}_{i,b},$ $\mathsf{pk}_{i,b},$ $P)$.
  
    \item[$-$] Set $\mathsf{mpk} := (\{\mathsf{mpk}_{i,b}\}_{i,b},y)$ and $\mathsf{aux}' := \{\mathsf{aux}'_{i,b}\}_{i,b}$.
    \item[$-$] Output $(\mathsf{mpk}, \mathsf{aux}')$.
\end{itemize}

\item[]\textsf{Encrypt}(\textsf{mpk}, X, $\mu$):
\leavevmode
\begin{itemize}
    \item[$-$] Parse $\mathsf{mpk} = (\{\mathsf{mpk}_{i,b}\}_{i,b},y)$ 
    \item[$-$] For each $i, b$, compute $\mathsf{CT}_{i,b} \leftarrow \mathsf{RABE}.\mathsf{Encrypt}(\mathsf{mpk}_{i,b}, X, \mu[i]; r_{i,b})$.
    \item[$-$] Let $d := (\{\mathsf{mpk}_{i,b}, \mathsf{CT}_{i,b}\}_{i,b}, X)$ and $w := (\mu, \{r_{i,b}\}_{i,b})$.
    \item[$-$] Compute $\pi \leftarrow \langle \mathcal{P}(w),\mathcal{V}(y)\rangle(d)$.
    \item[$-$] Output $\mathsf{ct} := (\{\mathsf{CT}_{i,b}\}_{i,b}, \pi)$.
\end{itemize}

\item[]\textsf{Update}(\textsf{crs}, \textsf{aux}, \textsf{pk}):
\leavevmode
\begin{itemize}
    \item[$-$] Parse $\mathsf{crs} = (\{\mathsf{crs}_{i,b}\}_{i,b},y)$, $\mathsf{aux} = \{\mathsf{aux}_{i,b}\}_{i,b}$, and $\mathsf{pk} = \{\mathsf{pk}_{i,b}\}_{i,b}$.
    \item[$-$] For each $i \in [\ell_m],$ $b \in \{0,1\}$, compute $\mathsf{hsk}_{i,b}$ $ \leftarrow$ $\mathsf{RABE}.\mathsf{Update}(\mathsf{crs}_{i,b},$ $ \mathsf{aux}_{i,b},$ $ \mathsf{pk}_{i,b})$.
    \item[$-$] Output $\mathsf{hsk} := \{\mathsf{hsk}_{i,b}\}_{i,b}$.
\end{itemize}

\item[]\textsf{Decrypt}(\textsf{sk}, \textsf{hsk}, {X}, \textsf{ct}):
\leavevmode
\begin{itemize}
    \item[$-$] Parse $\mathsf{sk} = \mathcal{D}_e$ and $\mathsf{hsk} = \{\mathsf{hsk}_{i,b}\}_{i,b}$.
    \item[$-$] Compute and return $\mu := \mathcal{D}_e(\mathsf{hsk}, \mathsf{ct})$.
\end{itemize}

\begin{center}
\begin{tcolorbox}[
  colback=white, title=Left or Right Decryption Circuit $\mathcal{D}$,
  arc=4mm, boxrule=0.8pt, width=0.9\linewidth]

\textbf{Input:} \textsf{hsk}, \textsf{CT}

\textbf{Hardcoded:} $\{\mathsf{sk}_{i,z[i]}\}_{i \in [\ell_m]}$, $z \in \{0,1\}^{\ell_m}$

\vspace{2mm}
\begin{enumerate}
    \item Parse $\mathsf{CT} = \left( \{ \mathsf{CT}_{i,0}, \mathsf{CT}_{i,1} \}_{i \in [\ell_m]}, \pi \right)$
    \item Parse $\mathsf{hsk} = \left( \{ \mathsf{hsk}_{i,0}, \mathsf{hsk}_{i,1} \}_{i \in [\ell_m]}  \right)$
    
    \item If $\pi\neq 1$, then output $\perp$
    \item For each $i \in [\ell_m]$, compute:
    \[
    m[i] \leftarrow \mathsf{RABE}.\mathsf{Decrypt}(\mathsf{sk}_{i,z[i]}, \mathsf{hsk}_{i, z[i]},{X}, \mathsf{CT}_{i, z[i]})
    \]
    \item Output: $m := m[1] \parallel m[2] \parallel \cdots \parallel m[\ell_m]$
\end{enumerate}
\end{tcolorbox}

\vspace{2mm}
\captionof{figure}{ Left or Right Decryption Circuit}
\label{Figure:UpdatedCircuit-1}
\end{center}
\item[]\textsf{SimKeyGen}(\textsf{crs}, \textsf{aux}, $P,\mathsf{B}$):
\leavevmode
\begin{itemize}
    \item[$-$] Parse $\mathsf{crs} = (\{\mathsf{crs}_{i,b}\}_{i,b},y)$.
    
    \item[$-$] For all $i \in [\ell_m], b \in \{0,1\}$, compute $(\mathsf{pk}_{i,b}, \mathsf{sk}_{i,b})$ $ \leftarrow $ $\mathsf{RABE}.\mathsf{KeyGen}(\mathsf{crs}_{i,b},$ $ \mathsf{aux}_{i,b},$ $P)$.
    
    \item[$-$] Sample $z^*_{{\mathsf{{pk}}}} \leftarrow \{0,1\}^{\ell_m}$ uniformly at random.

    \item[$-$] Set ${\mathsf{pk}} := (\{\mathsf{pk}_{i,b}\}_{i,b})$ and update $\mathsf{B}[\mathsf{{pk}}] = (\{\mathsf{sk}_{i,b}\}_{i,b}, z^*_{{\mathsf{{pk}}}})$.

    \item[$-$] Output $(\widetilde{\mathsf{pk}}, \mathsf{B})$.
\end{itemize}

\item[]\textsf{SimRegPK}(\textsf{crs}, \textsf{aux}, ${\textsf{pk}},$ $P$):
\leavevmode
\begin{itemize}
    \item[$-$] Parse $\mathsf{crs} = (\{\mathsf{crs}_{i,b}\}_{i,b},y)$ and ${\mathsf{pk}} = (\{\mathsf{pk}_{i,b}\}_{i,b})$.
    
    \item[$-$] For each $i, b$, compute $(\widetilde{\mathsf{mpk}}_{i,b}, \widetilde{\mathsf{aux}}_{i,b}) $ $\leftarrow$ $\mathsf{RABE}.\mathsf{RegPK}(\mathsf{crs}_{i,b},$ $ \mathsf{aux}_{i,b},$ $ \mathsf{pk}_{i,b},$ $ P)$.  
    
    \item[$-$] Set $\widetilde{\mathsf{mpk}} := (\{\widetilde{\mathsf{mpk}}_{i,b}\}_{i,b}, y)$ and $\widetilde{\mathsf{aux}} := \{\widetilde{\mathsf{aux}}_{i,b}\}_{i,b}$.
    \item[$-$] Output $(\widetilde{\mathsf{mpk}}, \widetilde{\mathsf{aux}})$.
\end{itemize}

\item[]\textsf{SimCorrupt}(\textsf{crs}, $\textsf{pk}$, $\mathsf{B}$):
\leavevmode
\begin{itemize}
    \item[$-$] Parse $\mathsf{pk} = (\{\mathsf{pk}_{i,b}\}_{i,b})$ and $\mathsf{B}[\mathsf{pk}]= (\{\mathsf{sk}_{i,b}\}_{i,b}, z^*_{\mathsf{pk}})$.
    
    \item[$-$] Compute $\widetilde{\mathsf{sk}} := iO(\mathcal{D}_0[\{\mathsf{sk}_{i,0}\}_{i}]),$ where the circuit is defined in Fig. \ref{Figure:UpdatedCircuit-2}.
    \item[$-$] Output $\widetilde{\mathsf{sk}}$.
\end{itemize}

\begin{center}
\begin{tcolorbox}[
  colback=white, title=Left Decryption Circuit $\mathcal{D}_0$,
  arc=4mm, boxrule=0.8pt, width=0.9\linewidth]

\textbf{Input:} \textsf{hsk}, \textsf{CT}

\textbf{Hardcoded:} $\{\mathsf{sk}_{i,0}\}_{i \in [\ell_m]}$

\vspace{2mm}
\begin{enumerate}
    \item Parse $\mathsf{CT} = \left( \{ (\mathsf{CT}_{i,0}, \mathsf{CT}_{i,1}) \}_{i \in [\ell_m]}, \pi  \right)$
    
    \item Parse $\mathsf{hsk} = \left( \{ \mathsf{hsk}_{i,0}, \mathsf{hsk}_{i,1} \}_{i \in [\ell_m]} \right)$
    
    \item If $\pi \neq 1$, output $\perp$
    
    \item For each $i \in [\ell_m]$, compute:
    $$m[i] \leftarrow \mathsf{RABE}.\mathsf{Decrypt}(\mathsf{sk}_{i,0}, \mathsf{hsk}_{i,0}, {X}, \mathsf{CT}_{i,0})$$
    
    \item Output: $m = m[1] \parallel m[2] \parallel \cdots \parallel m[\ell_m]$
\end{enumerate}
\end{tcolorbox}

\vspace{1mm}
\captionof{figure}{Left Decryption Circuit}
\label{Figure:UpdatedCircuit-2}
\end{center}

\item[]\textsf{SimCT}($\widetilde{\textsf{mpk}}$, $\mathsf{B}$, $X$):
\leavevmode
\begin{itemize}
    \item[$-$] Parse $\widetilde{\mathsf{mpk}} = (\{\widetilde{\mathsf{mpk}}_{i,b}\}_{i,b},y)$.
    
    \item[$-$] Compute $z^* := \bigoplus\limits_{\mathsf{pk} \in \mathsf{B}} z^*_{\mathsf{pk}}$ where $\mathsf{B}[\mathsf{pk}] = (\{\mathsf{sk}_{i,b}\}_{i,b}, z^*_{\mathsf{pk}})$.
    \item[$-$] For each $i \in [\ell_m]$, compute:
    \begin{align*}
        \mathsf{CT}^*_{i,z^*[i]} &\leftarrow \mathsf{RABE}.\mathsf{Encrypt}(\widetilde{\mathsf{mpk}}_{i,z^*[i]}, X, 0), \\
        \mathsf{CT}^*_{i,1 - z^*[i]} &\leftarrow \mathsf{RABE}.\mathsf{Encrypt}(\widetilde{\mathsf{mpk}}_{i,1 - z^*[i]}, X, 1).
    \end{align*}
    \item[$-$] Define $d^* := (\{\widetilde{\mathsf{mpk}}_{i,0}, \widetilde{\mathsf{mpk}}_{i,1}, \mathsf{CT}^*_{i,0}, \mathsf{CT}^*_{i,1}\}_{i}, X)$.
    \item[$-$] Compute $\widetilde{\pi} \leftarrow \langle \mathsf{Sim}\rangle(d^*,y)$.
    \item[$-$] Output $\widetilde{\mathsf{ct}} := (\{\mathsf{CT}^*_{i,b}\}_{i,b}, \widetilde{\pi})$.
\end{itemize}

\item[]\textsf{Reveal}(\textsf{pk}, $\mathsf{B}$, $\widetilde{\textsf{ct}}$, $\mu$):
\leavevmode
\begin{itemize}
    \item[$-$] Parse $\mathsf{B}[\mathsf{pk}] = (\{\mathsf{sk}_{i,b}\}_{i,b}, z^*)$.
    \item[$-$] Parse $\mu = \mu[1] \parallel \mu[2] \parallel \cdots \parallel \mu[\ell_m]$.
    \item[$-$] Compute $\widetilde{\mathsf{sk}} := iO(\mathcal{D}[\{\mathsf{sk}_{i,z^*[i] \oplus \mu[i]}\}_{i}, z^* \oplus \mu])$.
    \item[$-$] Output $\widetilde{\mathsf{sk}}$.
\end{itemize}
\end{description}
\begin{theorem}\label{thm:shadow-security}
If the registered attribute-based encryption scheme $\prod_{\textsf{RABE}}$ is secure and assuming the existence of an indistinguishability obfuscator $i\mathcal{O}$ for \textsf{P/poly}, and an interactive zero-knowledge argument system $\prod_{\mathsf{ZKA}}$ for \textsf{NP} language, then the \textsf{Shad-RABE} construction is secure.
\end{theorem}
\begin{proof}
We define a sequence of hybrid experiments to transition between real and ideal security games.\vspace{0.2cm}
\begin{itemize}\setlength\itemsep{1em}
\item[$\bullet$] $\textsf{Hyb}_0$: This corresponds precisely to the security experiment $\mathsf{Exp}^{\textsf{Shad-RABE}}_{\mathcal{A}}(\lambda, 0)$. Throughout the analysis, $x^*$ denotes the challenge attribute set and $\mu$ refers to the target message chosen by the adversary.

\item \textbf{$\mathsf{Hyb}_1$}: This hybrid modifies $\mathsf{Hyb}_0$ by altering the generation of secret keys during the query phase. Instead of using the left-or-right decryption circuit
\begin{align*}
    &~~~\mathcal{D}[\{\mathsf{sk}_{i,z[i]}\}_{i \in [\ell_m]}, z], 
\end{align*}
the challenger now employs the left-only decryption circuit
\begin{align*}
    &\mathcal{D}_0[\{\mathsf{sk}_{i,0}\}_{i \in [\ell_m]}].
\end{align*}
Consequently, for each key-generation query, it provides
\begin{align*}
&\widetilde{\mathsf{sk}} = iO\left(\mathcal{D}_0[\{\mathsf{sk}_{i,0}\}_{i \in [\ell_m]}]\right),
\end{align*}
instead of
\begin{align*}
&~~~~\mathsf{sk} = iO\left(\mathcal{D}[\{\mathsf{sk}_{i,z[i]}\}_{i \in [\ell_m]}, z]\right).
\end{align*}
As a result, the secret keys in this hybrid no longer depend on the random selection bit string $z$.\vspace{0.2cm}

\item \textbf{$\mathsf{Hyb}_2$}: This game is the same as $\mathsf{Hyb}_1$ apart from the fact that it transitions from real to simulated machine $\langle Sim \rangle$ of  an interactive proof system given by the interactive machines $\langle \mathcal{P},\mathcal{V}\rangle$ for an \textsf{NP} language $\mathcal{L}$ characterized by the relation $R$ in the ciphertext.

\item \textbf{$\mathsf{Hyb}_3$}: This game is the same as $\mathsf{Hyb}_2$ except that it introduces asymmetric encryption patterns in the challenge ciphertext. Rather than encrypting the same message bit under both master public keys, the challenger now generates position-dependent encryptions. For the combined randomness $z:= \bigoplus z_{\text{user}}$, where $z_{\text{user}}$ is sampled in $\textsf{KeyGen}$ for every user  and each bit position $i \in [\ell_m]$, it computes:
\begin{align*}
\mathsf{CT}^*_{i,z[i]} &\leftarrow \mathsf{RABE}.\mathsf{Encrypt}(\widetilde{\mathsf{mpk}}_{i,z[i]}, X^*, \mu[i]), \\
\mathsf{CT}^*_{i,1-z[i]} &\leftarrow \mathsf{RABE}.\mathsf{Encrypt}(\widetilde{\mathsf{mpk}}_{i,1-z[i]}, X^*, 1-\mu[i]).
\end{align*}

\item \textbf{$\mathsf{Hyb}_4$}: This game is the same as $\mathsf{Hyb}_3$, except that it is completely independent of the target message $\mu$. The challenger uses the same combined randomness $z^*:= \bigoplus z^*_{\text{user}}$, where $z^*_{\text{user}}$ is sampled in $\textsf{KeyGen}$ for every user and generates challenge ciphertexts according to:
\begin{align*}
\mathsf{CT}^*_{i,z^*[i]} &\leftarrow \mathsf{RABE}.\mathsf{Encrypt}(\widetilde{\mathsf{mpk}}_{i,z^*[i]}, X^*, 0), \\
\mathsf{CT}^*_{i,1-z^*[i]} &\leftarrow \mathsf{RABE}.\mathsf{Encrypt}(\widetilde{\mathsf{mpk}}_{i,1-z^*[i]}, X^*, 1).
\end{align*}

When revealing secret keys in the final phase, the challenger computes for each honest user with randomness $z^*_{\text{user}}$:
$$\mathsf{sk} := iO(\mathcal{D}[\{\mathsf{sk}_{i,z^*_{\text{user}}[i] \oplus \mu[i]}\}_{i \in [\ell_m]}, z^*_{\text{user}} \oplus \mu]).$$

Note that the game $\mathsf{Hyb}_4$ is identical to the simulated experiment $\mathsf{Exp}^{\textsf{Shad-RABE}}_{\mathcal{A}}$ $(\lambda,$ $ 1)$.
\end{itemize}
\begin{table}[htbp]
    \centering
    \renewcommand{\arraystretch}{2.3} 
    \setlength{\tabcolsep}{6pt} 

    \resizebox{\textwidth}{!}{
    \begin{tabular}{c|c|c|c|c}
        \hline
        {\textbf{Hybrid}}  
        & \multicolumn{1}{c|}{\textbf{Secret Keys}} 
        & \multicolumn{1}{c|}{\textbf{uNIZK Proofs}} 
        & \multicolumn{1}{c|}{\textbf{Challenge Ciphertext}} 
        & \multicolumn{1}{c}{\textbf{Dependence on $\mu$}} \\
        \hline

        \rule{0pt}{3.8ex}$\textsf{Hyb}_0$ 
        & \fbox{$iO(\mathcal{D}[\{\mathsf{sk}_{i,z[i]}\}_{i \in [\ell_m]}, z])$}
        & \fbox{Real proofs} 
        & \fbox{$\mathsf{CT}^*_{i,0}, \mathsf{CT}^*_{i,1}$ encrypt $\mu[i]$} 
        & \fbox{Yes (depends on $\mu$)} \\
        
        \rule{0pt}{3.8ex}$\textsf{Hyb}_1$ 
        & \colorbox{black!20}{$iO(\mathcal{D}_0[\{\mathsf{sk}_{i,0}\}_{i \in [\ell_m]}])$}
        & $\bm{\downarrow}$ 
        & $\bm{\downarrow}$ 
        & $\bm{\downarrow}$ \\

        \rule{0pt}{3.8ex}$\textsf{Hyb}_2$ 
        & $\bm{\downarrow}$ 
        & \colorbox{black!20}{Simulated proofs using $\langle Sim\rangle$} 
        & $\bm{\downarrow}$ 
        & $\bm{\downarrow}$ \\

        \rule{0pt}{3.8ex}$\textsf{Hyb}_3$ 
        & $\bm{\downarrow}$ 
        & $\bm{\downarrow}$ 
        & \colorbox{black!20}{Asymmetric encryptions: $\mathsf{CT}^*_{i,z[i]}, \mathsf{CT}^*_{i,1-z[i]}$} 
        & $\bm{\downarrow}$ \\

        \rule{0pt}{3.8ex}$\textsf{Hyb}_4$ 
        & \colorbox{black!20}{$iO(\mathcal{D}[\{\mathsf{sk}_{i,z^*_{\text{user}}[i]\oplus \mu[i]}\}, z^*_{\text{user}}\oplus \mu])$}
        & $\bm{\downarrow}$ 
        & \colorbox{black!20}{Encrypt fixed bits $0/1$ under $z^*$ randomness} 
        & \colorbox{black!20}{No (independent of $\mu$)} \\
        \hline
    \end{tabular}
    }
    
    \vspace{0.2cm}
    \caption{\justifying \noindent Summary of the hybrid sequence used in the security proof of \textsf{Shad-RABE}. The decryption circuit column shows the obfuscated program used to answer key queries. Simulated proof $\langle Sim \rangle$ appear in $\mathsf{Hyb}_2$, and asymmetric challenge ciphertexts are introduced in $\mathsf{Hyb}_3$. In $\mathsf{Hyb}_4$, the ciphertext becomes entirely independent of the challenge message $\mu$. Shaded entries mark the first point of deviation from the previous hybrid. Arrows ($\downarrow$) indicate no change from the previous hybrid.}
    \label{Summary-hybrid-shadRABE}
\end{table}

\noindent Let the distribution of the output of an execution of $\textsf{Hyb}_j$, with adversary $\mathcal{A}$, be denoted as $\textsf{Hyb}_j(\mathcal{A})$. We establish computational indistinguishability between consecutive hybrid games through the following sequence of reductions, overview shown in Fig. \ref{fig:shad-rabe-hybrids-arranged}.
\begin{figure}[htbp]
\begin{center}
\resizebox{1\linewidth}{!}{%
\begin{tikzpicture}[
    node distance=1.5cm and 1.5cm,
    box/.style={draw, rectangle, rounded corners=2mm, minimum height=1cm, minimum width=2.3cm, thick, align=center},
    dashedbox/.style={dashed, draw, thick, rounded corners=4mm, inner sep=5mm},
    arrow/.style={
      thick,
      postaction={decorate},
      decoration={markings, mark=at position 0.6 with {\arrow{Latex[length=3mm,width=2mm]}}}
    },
    parrow/.style={<->, thick},
    dots/.style={very thick, dotted}
]

\node[box] (H0) {\textsf{Hyb}$_0$($\mathcal{A}$)};

\node[box, right=2cm of H0] (S00) {\textsf{SubHyb}$_0^{\,0}$($\mathcal{A}$)};
\node[box, right=1.5cm of S00] (S01) {\textsf{SubHyb}$_0^{\,1}$($\mathcal{A}$)};
\node (S0mid) [right=1.5cm of S01] {};
\node[box, right=3cm of S01] (S0q) {\textsf{SubHyb}$_0^{\,q}$($\mathcal{A}$)};
\node[box, right=2cm of S0q] (H1) {\textsf{Hyb}$_1$($\mathcal{A}$)};

\node[box, below=2cm of H1] (H2) {\textsf{Hyb}$_2$($\mathcal{A}$)};

\node[box, left=2cm of H2] (S20) {\textsf{SubHyb}$_2^{\,0}$($\mathcal{A}$)};
\node[box, left=1.5cm of S20] (S21) {\textsf{SubHyb}$_2^{\,1}$($\mathcal{A}$)};
\node (S2mid) [left=1.5cm of S21] {};
\node[box, left=3cm of S21] (S2m) {\textsf{SubHyb}$_2^{\,\ell_m}$($\mathcal{A}$)};

\node[box, left=2cm of S2m] (H3) {\textsf{Hyb}$_3$($\mathcal{A}$)};
\node[box, below=1.5cm of H3] (H4) {\textsf{Hyb}$_4$($\mathcal{A}$)};

\draw[dots] (S01.east) -- (S0mid.west);
\draw[dots] (S21.west) -- (S2mid.east);

\node[dashedbox, fit=(S00)(S01)(S0mid)(S0q), name=boxA] {};
\node[above=2mm of boxA] {\textbf{Lemma~\ref{lem:hyb0-hyb1-updated} (key-query hybrids)}};

\node[dashedbox, fit=(S20)(S21)(S2mid)(S2m), name=boxB] {};
\node[below=2mm of boxB] {\textbf{Lemma~\ref{lem:hyb2-hyb3} (position-wise hybrids)}};

\draw[doublearrow] (H0) -- (S00) node[midway,above=3pt]{\textbf{Lemma~\ref{lem:hyb0-hyb1-updated}}} node[midway,below=3pt,align=center]{\textbf{Identical}};
\draw[arrow] (S00) -- (S01) node[midway,above=3pt]{\textbf{Claim 1}} node[midway,below=3pt,align=center]{\textbf{$i\mathcal{O}$}\\ \textbf{Security}};
\draw[arrow] (S01) -- (S0q) node[midway,above=3pt]{\textbf{Claim 1}} node[midway,below=3pt,align=center]{\textbf{$i\mathcal{O}$}\\ \textbf{Security}};
\draw[doublearrow] (S0q) -- (H1) node[midway,above=3pt]{\textbf{Lemma~\ref{lem:hyb0-hyb1-updated}}} node[midway,below=3pt,align=center]{\textbf{Identical}};

\draw[arrow] (H1) -- (H2) node[midway,right]{\textbf{Lemma~\ref{lem:hyb1-hyb2-updated}}} node[midway,left=3pt,align=center]{\textbf{\textsf{ZKA}}\\ \textbf{Security}};

\draw[doublearrow] (H2) -- (S20) node[midway,above=3pt]{\textbf{Lemma~\ref{lem:hyb1-hyb2-updated}}} node[midway,below=3pt,align=center]{\textbf{Identical}};
\draw[arrow] (S20) -- (S21) node[midway,above=3pt]{\textbf{Claim 2}} node[midway,below=3pt,align=center]{\textbf{\textsf{RABE}}\\ \textbf{Security}};
\draw[arrow] (S21) -- (S2m) node[midway,above=3pt]{\textbf{Claim 2}} node[midway,below=3pt,align=center]{\textbf{\textsf{RABE}}\\ \textbf{Security}} ;
\draw[doublearrow] (S2m) -- (H3) node[midway,above=3pt]{\textbf{Lemma~\ref{lem:hyb1-hyb2-updated}}} node[midway,below=3pt,align=center]{\textbf{Identical}};
\draw[arrow] (H3) -- (H4) node[midway,left=3pt]{\textbf{Lemma~\ref{lem:hyb3-hyb4}}} node[midway,right=3pt,align=center]{\textbf{$i\mathcal{O}$+\textsf{RABE}}\\ \textbf{Security}} ;

\end{tikzpicture}
}
\end{center}
\captionsetup{font=scriptsize}
\caption {\scriptsize A single arrow `\protect\tikz[baseline=-0.5ex]{\protect\draw[arrow] (0,0) -- (1,0);}' represents statistical/computational indistinguishability. The double arrow `\protect\tikz[baseline=-0.5ex]{\protect\draw[doublearrow] (0,0) -- (1,0);}' indicates that for both $b \in \{0,1\}$, the adversary is in the same experiment. Above each arrow, we write the lemma/claim that proves indistinguishability, while below, we state the underlying security assumption.}
\label{fig:shad-rabe-hybrids-arranged}
\end{figure}
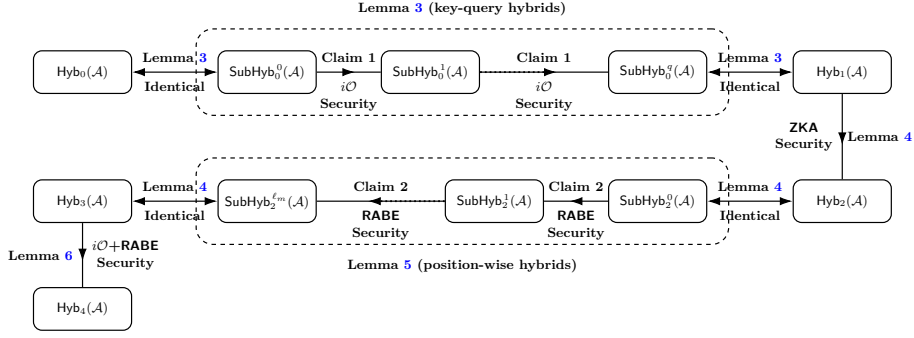

\begin{lemma}\label{lem:hyb0-hyb1-updated}
Suppose that $\prod_{\textsf{RABE}}$ satisfies perfect correctness, $\prod_{\mathsf{ZKA}}$ achieves soundness property, and the indistinguishability obfuscator $i\mathcal{O}$ ensures security under indistinguishability. Then, the following advantage in distinguishing between hybrids is negligible,
\begin{align*}
    &\left|\Pr[\mathsf{Hyb}_0(\mathcal{A}) = 1] - \Pr[\mathsf{Hyb}_1(\mathcal{A}) = 1]\right| \leq \mathsf{negl}(\lambda).
\end{align*}
\end{lemma}
\begin{proof}
We employ a hybrid argument over the adversary's key-generation queries. Let $q$ denote the total number of such queries made by the adversary $\mathcal{A}$.\vspace{0.15cm}

\noindent Now, we define a sequence of hybrid sub-experiments $\textsf{SubHyb}_{0}^j$ for each $j \in [0, q]$, where the challenger responds to the adversary's key queries as follows:
\begin{description}
    \item[$-$]\justifying For queries $k \in [j+1, q]$, the challenger returns $\mathsf{sk} =$ $i\mathcal{O}\left(\mathcal{D}\left[\{\mathsf{sk}_{i,z[i]}\}_{i \in [\ell_m]},\right.\right.$ $\left.\left. z\right]\right).$
    \item[$-$]\justifying For queries $k \in [1, j]$, the challenger returns $\widetilde{\mathsf{sk}} = i\mathcal{O}\left(\mathcal{D}_0\left[\{\mathsf{sk}_{i,0}\}_{i \in [\ell_m]}\right]\right).$
\end{description}

\noindent We observe that $\mathsf{SubHyb}_0^0(\mathcal{A}) \equiv \mathsf{Hyb}_0(\mathcal{A}),$ and $\mathsf{SubHyb}_0^q(\mathcal{A}) \equiv \mathsf{Hyb}_1(\mathcal{A})$.\vspace{0.2cm}

\noindent Let $\mathsf{Invalid}$ denote the event that there exists a statement $d^\dagger \notin \mathcal{L}$  such that for $\langle \mathcal{P}(w),\mathcal{V}(y)\rangle(d)=1$. By the soundness of the interactive argument system $\prod_{\textsf{ZKA}}$, it holds that
\[ 
\Pr[\mathsf{Invalid}] \leq \mathsf{negl}(\lambda). 
\]

\noindent Conditioned on the event $\neg\mathsf{Invalid}$, the circuits $\mathcal{D}$ and $\mathcal{D}_0$ are functionally equivalent on all valid inputs. This equivalence follows from their construction and the perfect correctness of the underlying \textsf{RABE} scheme. We now formalize the indistinguishability of adjacent subgames.\vspace{0.2cm}

\noindent\textbf{Claim 1.} For every $j\in[1, q]$, the distinguishing advantage between hybrid subexperiments $\mathsf{SubHyb}_0^{j-1}(\mathcal{A})$ and $\mathsf{SubHyb}_0^{j}(\mathcal{A})$ is negligible, \textit{i.e.,}
    \begin{align*}
        \left| \Pr[\mathsf{SubHyb}_0^{j-1}(\mathcal{A}) = 1] - \Pr[\mathsf{SubHyb}_0^{j}(\mathcal{A}) = 1] \right|.
    \end{align*}

\begin{proof}
\noindent Consider the transition between hybrid sub-experiments $\mathsf{SubHyb}_0^{j-1}$ and $\mathsf{SubHyb}_0^{j}$, which differ only in the $j$-th key query. In $\mathsf{SubHyb}_0^{j-1}$, the response is
\[ 
\mathsf{sk} = i\mathcal{O}\left(\mathcal{D}\left[\{\mathsf{sk}_{i,z[i]}\}_{i \in [\ell_m]}, z\right]\right), 
\]
while in \( \mathsf{SubHyb}_0^j \), the response is
\[ 
\widetilde{\mathsf{sk}} = i\mathcal{O}\left(\mathcal{D}_0\left[\{\mathsf{sk}_{i,0}\}_{i \in [\ell_m]}\right]\right). 
\]
\noindent Under the assumption $\neg\mathsf{Invalid}$, the two circuits $\mathcal{D}$ and $\mathcal{D}_0$ are functionally equivalent. Hence, by the security of indistinguishability obfuscator $i\mathcal{O}$, any \textsf{QPT} distinguisher can distinguish between these responses with at most negligible probability.\vspace{0.2cm}
\end{proof}

\noindent Applying the hybrid sub-experiments argument across the $q$ key queries and invoking the above claim repeatedly, we conclude that
\[ 
\left| \Pr[\mathsf{Hyb}_0(\mathcal{A}) = 1] - \Pr[\mathsf{Hyb}_1(\mathcal{A}) = 1] \right| \leq q \cdot \mathsf{negl}(\lambda) = \mathsf{negl}(\lambda), 
\]
which completes the proof of the lemma.
\end{proof}

\begin{lemma}\label{lem:hyb1-hyb2-updated}
If the interactive argument system $\prod_{\textsf{ZKA}}$ satisfies the zero-knowledge property, then
\begin{align*}
    \left|\Pr[\mathsf{Hyb}_1(\mathcal{A}) = 1] - \Pr[\mathsf{Hyb}_2(\mathcal{A}) = 1]\right| \leq \mathsf{negl}(\lambda).
\end{align*}
\end{lemma}
\begin{proof}
The hybrid experiments $\mathsf{Hyb}_1$ and $\mathsf{Hyb}_2$ differ solely in the generation of the ciphertext. Suppose that $\left|\Pr[\mathsf{Hyb}_1(\mathcal{A}) = 1] - \Pr[\mathsf{Hyb}_2(\mathcal{A}) = 1]\right|$ is non-negligible. Then there exists a distinguisher that separates the ciphertext distributions used in $\mathsf{Hyb}_1$ and $\mathsf{Hyb}_2$ with non-negligible advantage. \vspace{0.2cm} 

\noindent This would further imply that the distributions of the real and simulated proof, namely
\[
\langle \mathcal{P}(w), \mathcal{V}(y) \rangle(d)
\quad \text{and} \quad
\langle \mathsf{Sim} \rangle(d^*, y),
\]
are computationally distinguishable. Such a distinguisher contradicts the zero-knowledge property of $\prod_{\textsf{ZKA}}$, since it would allow one to efficiently distinguish the real view of the verifier from a simulated one for instances of size $|y|$ with non-negligible advantage.  \vspace{0.2cm}

\noindent Hence, by the zero-knowledge property, the advantage $\big|\Pr[\mathsf{Hyb}_1(\mathcal{A}) = 1]$ $ -$ $ \Pr[\mathsf{Hyb}_2(\mathcal{A}) $ $= 1]\big|$ must be negligible.
\end{proof}


\begin{lemma}\label{lem:hyb2-hyb3}
If the underlying \textsf{RABE} scheme $\prod_{\textsf{RABE}}$ is secure, as described in Section \ref{model:RABE}, then
\begin{align*}
    \left|\Pr[\mathsf{Hyb}_2(\mathcal{A}) = 1] - \Pr[\mathsf{Hyb}_3(\mathcal{A}) = 1]\right| \leq \mathsf{negl}(\lambda).
\end{align*}
\end{lemma}
\begin{proof}
We employ a position-by-position hybrid argument across the message bits. For each position $i \in [\ell_m]$, we demonstrate that modifying the encryption pattern at position $i$ introduces at most a negligible distinguishing advantage.\vspace{0.2cm}

\noindent We now define the hybrid sub-experiments games $\mathsf{SubHyb}_2^j$ for $j \in [0, \ell_m],$ where the challenger responds as follows:
\begin{description}
    \item[$-$] For positions $i \in [j+1, \ell_m]$, both ciphertexts encrypt the same message bit.
    \begin{align*}
        &\mathsf{CT}^*_{i,b} \leftarrow \mathsf{RABE}.\mathsf{Encrypt}(\widetilde{\mathsf{mpk}}_{i,b}, X^*, \mu[i])~\text{for}~b \in \{0, 1\}
    \end{align*}
    \item[$-$] For positions $i \in [1, j]$, the ciphertexts are generated using the following computation.
    \begin{align*}
    &\mathsf{CT}^*_{i,1-z[i]} \leftarrow \mathsf{RABE}.\mathsf{Encrypt}(\widetilde{\mathsf{mpk}}_{i,1-z[i]}, X^*, 1-\mu[i]),\\
    &\mathsf{CT}^*_{i,z[i]} \leftarrow \mathsf{RABE}.\mathsf{Encrypt}(\widetilde{\mathsf{mpk}}_{i,z[i]}, X^*, \mu[i])
    \end{align*}
\end{description}
By construction, $\mathsf{SubHyb}_2^0(\mathcal{A}) \equiv \mathsf{Hyb}_2(\mathcal{A})$ and $\mathsf{SubHyb}_2^{\ell_m}(\mathcal{A}) \equiv \mathsf{Hyb}_3(\mathcal{A})$. To complete the proof, we introduce the following claim.\vspace{0.2cm}

\noindent\textbf{Claim 2.} For every $j \in [1, \ell_m]$, the hybrids $\mathsf{SubHyb}_2^{j-1}(\mathcal{A})$ and $\mathsf{SubHyb}_2^{j}(\mathcal{A})$ are computationally indistinguishable under the security of the \textsf{RABE} scheme, \textit{i.e.,}
    \begin{align*}
        \left| \Pr[\mathsf{SubHyb}_2^{j-1}(\mathcal{A}) = 1] - \Pr[\mathsf{SubHyb}_2^{j}(\mathcal{A}) = 1] \right|.
    \end{align*}

\begin{proof}
To show that adjacent games $\mathsf{SubHyb}_2^{j-1}$ and $\mathsf{SubHyb}_2^j$ for $i \in [1, \ell_m],$ are computationally indistinguishable, we define a reduction $\mathcal{B}_{\mathsf{rabe}}$ against the security of the underlying \textsf{RABE} scheme.\vspace{0.2cm}

\noindent $\mathcal{B}_{\mathsf{rabe}}$ receives master public key $\widetilde{\mathsf{mpk}}$ from the \textsf{RABE} challenger of $\mathsf{Exp}^{\textsf{RABE}}_{\mathcal{A}}(\lambda, b^{'})$ for $b^{'}\in\{0,1\}$ and sets $\widetilde{\mathsf{mpk}}_{j,1-z[j]}$ $ :=$ $ \widetilde{\mathsf{mpk}}$. For all other public keys, $(i,b) \neq (j,1 - z[j])$, it generates $\widetilde{\mathsf{mpk}}_{i,b}$ itself. It then compiles the overall master public key as 
  \begin{align*}
      &\widetilde{\mathsf{mpk}} := \left( \{\widetilde{\mathsf{mpk}}_{i,b} \}_{i \in [\ell_m], b \in \{0,1\}}, y \right),
  \end{align*}
which is given to the adversary $\mathcal{A}$. Similar to before, $\mathcal{B}_{\mathsf{rabe}}$ receives $\mathsf{sk}_{j,1-z[j]}$ from the \textsf{RABE} challenger, and locally generates all other secret keys. Since these keys 
are never used in any decryption or response, the simulation remains indistinguishable. $\mathcal{B}_{\mathsf{rabe}}$ then constructs
\begin{align*}
    &\widetilde{\mathsf{sk}} := iO\left( \mathcal{D}_0[\{ \mathsf{sk}_{i,0} \}_{i \in [\ell_m]}] \right).
\end{align*}
At the challenge phase, $\mathcal{A}$ submits $(X^*, \mu)$ to the $\mathcal{B}_{\mathsf{rabe}}$ and it sends the message pair $(\mu[j], 1 - \mu[j])$ to the \textsf{RABE} challenger. It then computes $\mathsf{CT}^*_{j,z[j]}$ $\leftarrow \mathsf{RABE}.\mathsf{Encrypt}$ $(\widetilde{\mathsf{mpk}}_{j,z[j]},$ $X^*, \mu[j])$, and all remaining ciphertexts $\mathsf{CT}^*_{i,b}$ for $i\neq j,$ as in the hybrids $\mathsf{SubHyb}_2^{j-1}$ and $\mathsf{SubHyb}_2^j$. Finally, it outputs whatever $\mathcal{A}$ outputs.\vspace{0.2cm}

\noindent The reduction proceeds as follows:
\begin{description}\setlength\itemsep{0.50em}
    \item[$-$] If $b' = 0$, then $\mathsf{CT}^*_{j,1 - z[j]} \leftarrow \mathsf{RABE}.\mathsf{Encrypt}(\widetilde{\mathsf{mpk}}_{j,1-z[j]}, X^*, \mu[j]),$
  which implies that $\mathcal{B}_{\mathsf{rabe}}$ perfect simulation of the hybrid sub-experiment $\mathsf{SubHyb}_2^{j-1}$.
  \item[$-$] If $b' = 1$, then $\mathsf{CT}^*_{j,1 - z[j]} \leftarrow \mathsf{RABE}.\mathsf{Encrypt}(\widetilde{\mathsf{mpk}}_{j,1-z[j]}, X^*, 1-\mu[j]),$
    which implies that $\mathcal{B}_{\mathsf{rabe}}$ perfect simulation of the hybrid sub-experiment $\mathsf{SubHyb}_2^{j-1}$.
\end{description}
Hence, any non-negligible advantage by $\mathcal{A}$ in distinguishing between $\mathsf{SubHyb}_2^{j-1}$ and $\mathsf{SubHyb}_2^{j}$ can be used to break the security of the \textsf{RABE} scheme. Therefore, the hybrid sub-experiments are computationally indistinguishable.
\end{proof}

\noindent Applying the hybrid sub-experiment argument to all $\ell_m$ positions and using the above claim iteratively, we conclude that
\[
\left| \Pr[\mathsf{Hyb}_2(\mathcal{A}) = 1] - \Pr[\mathsf{Hyb}_3(\mathcal{A}) = 1] \right| \leq \ell_m \cdot \mathsf{negl}(\lambda) = \mathsf{negl}(\lambda).
\]
This concludes the proof of Lemma~\ref{lem:hyb2-hyb3}.
\end{proof}

\begin{lemma}\label{lem:hyb3-hyb4}
The hybrid games satisfy perfect indistinguishability:
\begin{align*}
    &\Pr\left[\mathsf{Hyb}_3(\mathcal{A}) = 1\right] = \Pr\left[\mathsf{Hyb}_4(\mathcal{A}) = 1\right].
\end{align*}
\end{lemma}
\begin{proof}
This hybrid transition is a purely formal change that preserves the adversary's view. Specifically, the two games become identical under the variable substitution $z := z^{*} \oplus \mu$.\vspace{0.2cm}

\noindent In game $\mathsf{Hyb}_3$, the challenge ciphertexts are given by
\begin{align*}
&~~~~~~\mathsf{CT}^*_{i,1 - z[i]} ~\text{encrypts}~ 1 - \mu[i], \\
&~~~~~~\mathsf{CT}^*_{i,z[i]}     ~\text{encrypts}~ \mu[i].
\end{align*}

\noindent After applying the substitution $z := z^* \oplus \mu$, the ciphertexts in $\mathsf{Hyb}_4$ take the form
\begin{align*}
&\mathsf{CT}^*_{i,z^*[i]}     ~\text{encrypts}~ 0, \\
&\mathsf{CT}^*_{i,1 - z^*[i]} ~\text{encrypts}~ 1.
\end{align*}

\noindent Since the secret keys are generated independently of the selection string $z$ in both hybrids, this change of variables does not affect the underlying distributions. Consequently, the adversary's view is identical across the two hybrids, and we conclude that
\[
\Pr[\mathsf{Hyb}_3(\mathcal{A}) = 1] = \Pr[\mathsf{Hyb}_4(\mathcal{A}) = 1].
\]

\end{proof}

\noindent Consequently, by combining Lemmas~\ref{lem:hyb0-hyb1-updated}, \ref{lem:hyb1-hyb2-updated}, \ref{lem:hyb2-hyb3}, and \ref{lem:hyb3-hyb4}, we conclude that the \textsf{Shad-RABE} construction satisfies security against all quantum polynomial-time adversaries.
\end{proof}

\subsection{Instatiation of \textsf{Shad-RABE} from Lattices}\label{Inst:Shadow-RABE}
Our generic construction of \textsf{Shad-RABE}, presented in Section \ref{generic-construct:Shadow-RABE} is fully modular and can be instantiated from any secure (key-policy) registered \textsf{ABE} scheme, zero-knowledge argument system, and indistinguishability obfuscation. In this section, we give a concrete lattice-based instantiation of each component, resulting in a quantum-secure \textsf{Shad-RABE} scheme. \vspace{0.2cm}

\noindent We instantiate our generic framework using the following state-of-the-art lattice-based primitives:

\begin{description}\setlength\itemsep{0.5em}
    \item[$\bullet$] \textbf{Registered \textsf{ABE} Component:} We employ the key-policy registered \textsf{ABE} scheme of Champion \textit{et al.} \cite{champion2025registered}, which supports arbitrary bounded-depth circuit policies. This scheme is proven secure under the falsifiable $\ell$-succinct \textsf{LWE} assumption in the random oracle model and provides succinct ciphertexts whose size remains independent of the attribute length.\vspace{0.2cm}

    \item[$\bullet$] \textbf{\textsf{ZKA} Component:} We utilize the post-quantum secure zero-knowledge argument system of Bitansky \textit{et al.} \cite{bitansky2020post} for arbitrary \textsf{NP} languages. Their construction achieves soundness and zero-knowledge under the plain \textsf{LWE} assumption against quantum adversaries.

    \item[$\bullet$] \textbf{Indistinguishability Obfuscation Component:} We utilize the lattice-based $i\mathcal{O}$ construction of Cini \textit{et al.} \cite{cini2025lattice}, which provides post-quantum security under the equivocal \textsf{LWE} assumption in conjunction with standard \textsf{NTRU} lattice assumptions. This represents the first plausibly post-quantum secure $i\mathcal{O}$  construction from well-studied lattice problems.
\end{description}

\noindent Adapting the above-mentioned component as a building block, we can construct a post-quantum secure \textsf{Shad-RABE} protocol from lattices.

%% file: 5-6_RABE-CD.tex
\section{\textsf{RABE} with Privately Verifiable Certified Deletion}\label{subsection-RABE-PriVCD}
This section first presents a generic construction of registered attribute-based encryption with privately verifiable certified deletion (\textsf{RABE-PriVCD}). The design builds upon the \textsf{Shad-RABE} protocol (\textit{cf.} Section \ref{generic-construct:Shadow-RABE}) and incorporates a symmetric-key encryption scheme with certified deletion (\textsf{SKE-CD}) that guarantees one-time certified deletion (\textsf{OT-CD}) security (\textit{cf.} Section \ref{subsection-enc-cd}). Following the generic model, we then propose an instantiation of the \textsf{RABE-PriVCD} scheme, which relies on standard hardness assumptions over lattices.\vspace{0.2cm}


\subsection{Generic Construction of \textsf{RABE-PriVCD}.}\label{generic-construct-RABE-PriVCD}

\noindent We construct a registered attribute-based encryption scheme with 
privately verifiable certified deletion, denoted by 
$\prod^{\mathsf{RABE}}_{\textsf{PriVCD}}$ $=(\mathsf{Setup},$ $\mathsf{KeyGen},$ $\mathsf{RegPK},$ $\mathsf{Encrypt},$ $\mathsf{Update},$ $\mathsf{Decrypt},$ $\mathsf{Delete},$ $\mathsf{Verify})$ 
from the following building blocks: a \textsf{Shad-RABE} scheme, specified by the algorithms, $\prod_{\textsf{Shad-RABE}}=$ \textsf{Shad-RABE}.($\mathsf{Setup},$ $\mathsf{KeyGen},$ $\mathsf{RegPK}, \mathsf{Encrypt},$ $\mathsf{Update},$ $\mathsf{Decrypt},$ $\mathsf{SimKeyGen},$ $\mathsf{SimRegPK},$ $\mathsf{SimCorrupt},$ $\mathsf{SimCT},$ $\mathsf{Reveal}),$ and a symmetric-key encryption scheme with certified deletion, comprising algorithms, $\prod_{\textsf{SKE-CD}}=\textsf{SKE-CD}.(\mathsf{KeyGen},$ $\mathsf{Encrypt},$ $\mathsf{Decrypt},$ $\mathsf{Delete},$ $\mathsf{Verify}).$ The detailed algorithms of \textsf{RABE-PriVCD} protocol are specified below.\vspace{0.2cm}

\begin{description}\setlength\itemsep{0.5em}

\item[]\textsf{Setup}$(1^\lambda, 1^{{\tau}})$:
\leavevmode
\begin{itemize}
  \item[$-$] Generate $\mathsf{crs} \leftarrow \mathsf{Shad\text{-}RABE}.\mathsf{Setup}(1^\lambda, 1^{{\tau}})$.
  \item[$-$] Output $\mathsf{crs}$.
\end{itemize}

\item[]\textsf{KeyGen}$(\mathsf{crs}, \mathsf{aux}, P)$:
\leavevmode
\begin{itemize}
  \item[$-$] Generate $(\mathsf{pk}, \mathsf{sk}) \leftarrow \mathsf{Shad\text{-}RABE}.\mathsf{KeyGen}(\mathsf{crs}, \mathsf{aux}, P)$.
  \item[$-$] Output $(\mathsf{pk}, \mathsf{sk})$.
\end{itemize}

\item[]\textsf{RegPK}$(\mathsf{crs}, \mathsf{aux}, \mathsf{pk}, P)$:
\leavevmode
\begin{itemize}
  \item[$-$] Generate $(\mathsf{mpk}, \mathsf{aux}') \leftarrow \mathsf{Shad\text{-}RABE}.\mathsf{RegPK}(\mathsf{crs}, \mathsf{aux}, \mathsf{pk}, P)$.
  \item[$-$] Output $(\mathsf{mpk}, \mathsf{aux}')$.
\end{itemize}

\item[]\textsf{Encrypt}$(\mathsf{mpk}, X, \mu)$:
\leavevmode
\begin{itemize}
 \item[$-$] Parse $\mathsf{mpk}$.
  \item[$-$] Generate $\mathsf{ske.sk} \leftarrow \mathsf{SKE\text{-}CD}.\mathsf{KeyGen}(1^\lambda)$.
  \item[$-$] Compute $\mathsf{srabe.ct} \leftarrow \mathsf{Shad\text{-}RABE}.\mathsf{Encrypt}(\mathsf{mpk}, X, \mathsf{ske.sk})$.
  \item[$-$] Compute $\mathsf{ske.ct} \leftarrow \mathsf{SKE\text{-}CD}.\mathsf{Encrypt}(\mathsf{ske.sk}, \mu)$.
  \item[$-$] Output ciphertext $\mathsf{ct}:= (\mathsf{srabe.ct}, \mathsf{ske.ct})$ and verification key $\mathsf{vk}:=\mathsf{ske.sk}$.
\end{itemize}

\item[]\textsf{Update}$(\mathsf{crs}, \mathsf{aux}, \mathsf{pk})$:
\leavevmode
\begin{itemize}
  \item[$-$] Generate $\mathsf{hsk} \leftarrow \mathsf{Shad\text{-}RABE}.\mathsf{Update}(\mathsf{crs}, \mathsf{aux}, \mathsf{pk})$.
  \item[$-$] Output $\mathsf{hsk}$.
\end{itemize}

\item[]\textsf{Decrypt}$(\mathsf{sk}, \mathsf{hsk},{X}, \mathsf{ct})$:
\leavevmode
\begin{itemize}
  \item[$-$] Parse $\mathsf{ct} = (\mathsf{srabe.ct}, \mathsf{ske.ct})$.
  \item[$-$] Compute $\mathsf{sk}' \leftarrow \mathsf{Shad\text{-}RABE}.\mathsf{Decrypt}(\mathsf{sk}, \mathsf{hsk},{X}, \mathsf{srabe.ct})$.
  \item[$-$] If $\mathsf{sk}' = \bot$ or $\mathsf{sk}' = \mathsf{GetUpdate}$, output $\mathsf{sk}'$.
  \item[$-$] Otherwise, compute and output $\mu' \leftarrow \mathsf{SKE\text{-}CD}.\mathsf{Decrypt}(\mathsf{sk}', \mathsf{ske.ct})$.
\end{itemize}

\item[]\textsf{Delete}$(\mathsf{ct})$:
\leavevmode
\begin{itemize}
  \item[$-$] Parse $\mathsf{ct} = (\mathsf{srabe.ct}, \mathsf{ske.ct})$.
  \item[$-$] Generate $\mathsf{ske.cert} \leftarrow \mathsf{SKE\text{-}CD}.\mathsf{Delete}(\mathsf{ske.ct})$.
  \item[$-$] Output certificate $\mathsf{cert} := \mathsf{ske.cert}$.
\end{itemize}

\item[]\textsf{Verify}$(\mathsf{vk}, \mathsf{cert})$:
\leavevmode
\begin{itemize}
  \item[$-$] Parse $\mathsf{vk} = \mathsf{ske.sk}$ and $\mathsf{cert} = \mathsf{ske.cert}$.
  \item[$-$] Output $b \leftarrow \mathsf{SKE\text{-}CD}.\mathsf{Verify}(\mathsf{ske.sk}, \mathsf{ske.cert})$.
\end{itemize}
\end{description}
\noindent\textbf{Correctness of \textsf{RABE-PrivCD}.} The correctness of the \textsf{RABE-PrivCD} protocol follows directly from the correctness of $\prod_{\textsf{Shad-RABE}}$ and $\prod_{\textsf{SKE-CD}}$ protocols.

\subsection{Proof of Security}
This section presents a provable theorem model for the certified deletion security of the $\prod^{\textsf{RABE}}_{\textsf{PriVCD}}$ protocol. 

\begin{theorem}\label{thm:rabecd-private security}
If the scheme $\prod_{\textsf{Shad-RABE}}$ is secure and $\prod_{\mathsf{SKE\text{-}CD}}$ satisfies the \textsf{OT-CD} security property, then the protocol $\prod^{\textsf{RABE}}_{\textsf{PriVCD}}$ is secure.\end{theorem}
\begin{proof}

Let $\mathcal{A}$ be a \textsf{QPT} adversary and $b$ be a bit. We define the following hybrid game $\textsf{Hyb}(b)$.

\begin{description}\setlength\itemsep{0.5em}

\item[$\bullet$] $\textsf{Hyb}(b)$: This corresponds to a modification of the original experiment $\mathsf{Exp}^{\textsf{RABE-PriVCD}}_{\mathcal{A}}$ $(\lambda, b)$, where actual arguments are replaced with simulated ones. The simulation is carried out as follows:\vspace{0.2cm}

\begin{description}\setlength\itemsep{0.5em}
    \item[$-$] \textbf{Setup:} The challenger samples \(\mathsf{crs} \leftarrow \mathsf{Shad\text{-}RABE}.\mathsf{Setup}(1^{\lambda},1^{\tau})\) and initializes the internal state as  
    \begin{align*}
        & \mathsf{aux} \leftarrow \bot ,~\mathsf{mpk} \leftarrow \bot ,~D \leftarrow \emptyset,\\
        &\mathsf{ctr} \leftarrow 0,~C \leftarrow \emptyset,~H \leftarrow \emptyset.
    \end{align*}

    \item[$-$] \textbf{Query Phase:} The challenger responds to the adversary's queries using the simulation algorithms:\vspace{0.2cm}
    
    \begin{itemize}[leftmargin=*]
        \item[(a)] For a corrupted key query, upon receiving $(\mathsf{pk},P)$ from the adversary $\mathcal{A}$, the challenger computes
        \begin{align*}
            &(\mathsf{mpk}',\mathsf{aux}') \leftarrow \mathsf{Shad\text{-}RABE}.\mathsf{SimRegPK}(\mathsf{crs},\mathsf{aux},\mathsf{pk},P, \textsf{B}),
        \end{align*}
        and updates the states as
        \begin{align*}
            &\mathsf{aux}\leftarrow\mathsf{aux}',~D[\mathsf{pk}] \leftarrow D[\mathsf{pk}] \cup \{P\},\\
            &\mathsf{mpk}\leftarrow\mathsf{mpk}',~C\leftarrow C\cup\{\mathsf{pk}\}.
        \end{align*}
        
        \item[(b)] For an honest key query with $P \in \mathcal{P}$, the challenger increments $\mathsf{ctr}\leftarrow\mathsf{ctr}+1$, samples 
        \[
            (\mathsf{pk}_{\mathsf{ctr}},\mathsf{B}) \leftarrow 
            \mathsf{Shad\text{-}RABE}.\mathsf{SimKeyGen}(\mathsf{crs},\mathsf{aux}, P, \mathsf{B}),
        \]
        and computes 
        \begin{align*}
            &(\mathsf{mpk}',\mathsf{aux}') \leftarrow 
            \mathsf{Shad\text{-}RABE}.\mathsf{SimRegPK}(\mathsf{crs},\mathsf{aux},\mathsf{pk},P, \textsf{B}).
        \end{align*}
        The states are then updated as follows:
        \begin{align*}
            &\mathsf{mpk}\leftarrow\mathsf{mpk}',~H\leftarrow H \cup \{\mathsf{pk}_{\mathsf{ctr}}\}\\
            &\mathsf{aux}\leftarrow\mathsf{aux}',~D[\mathsf{pk}_{\mathsf{ctr}}]\leftarrow D[\mathsf{pk}_{\mathsf{ctr}}]\cup\{P\}.
        \end{align*}
        
        \item[(c)] On corruption of an honest index $i \in [1,\mathsf{ctr}]$ with $\mathsf{pk}_i \in H$, the challenger computes
        \begin{align*}
            &\mathsf{sk}_i \leftarrow \mathsf{Shad\text{-}RABE}.\mathsf{SimCorrupt}(\mathsf{crs}, \mathsf{pk}_i, \textsf{B}),
        \end{align*}
        updates
        \begin{align*}
            &H \leftarrow H \setminus \{\mathsf{pk}_i\},~C \leftarrow C \cup \{\mathsf{pk}_i\},
        \end{align*}
        and returns $\mathsf{sk}_i$ to $\mathcal{A}$.
    \end{itemize}

    \item[$-$] \textbf{Challenge:} On receiving $(\mu_0,\mu_1,X^*)$ from $\mathcal{A}$, the challenger generates
    \begin{align*}
        &\mathsf{ske.sk} \leftarrow \mathsf{SKE\text{-}CD}.\mathsf{KeyGen}(1^\lambda),
    \end{align*}
    and computes
    \begin{align*}
       \mathsf{srabe.ct}^* &\leftarrow \mathsf{Shad\text{-}RABE}.\mathsf{SimCT}(\mathsf{mpk},\mathsf{B}, X^*),\\ 
       \mathsf{ske.ct}^* &\leftarrow \mathsf{SKE\text{-}CD}.\mathsf{Encrypt}(\mathsf{ske.sk},\mu_0). 
    \end{align*}
    It then outputs the challenge pair
    \begin{align*}
        &\mathsf{ct}^* := (\mathsf{srabe.ct}^*, \mathsf{ske.ct}^*),~\mathsf{vk}^* := \mathsf{ske.sk}
    \end{align*}
    to the adversary.
    \item[$-$] \textbf{Deletion:} After ciphertext deletion, for each uncorrupted $\mathsf{pk}_i \in H$, the challenger computes
    \begin{align*}
        \mathsf{sk}_i &\leftarrow \mathsf{Shad\text{-}RABE}.\mathsf{Reveal}(\mathsf{pk}_i,\textsf{B}, \mathsf{srabe.ct}^*,\textsf{ske.sk}),\\
        \mathsf{hsk}_i &\leftarrow \mathsf{Shad\text{-}RABE}.\mathsf{Update}(\mathsf{crs},\mathsf{aux},\mathsf{pk}_i),
    \end{align*}
    and returns both values to $\mathcal{A}$. 
\end{description}

\end{description}

\noindent We now state two lemmas that complete the proof. \vspace{0.2cm}

\begin{lemma}\label{lem:real-to-hybrid}
If \(\textsf{Shad-RABE}\) is secure, then for every \textsf{QPT} adversary \(\mathcal{A}\), there is a negligible
\(\mathsf{negl}\) such that
\begin{align*}
&\bigl|\Pr[\mathsf{Exp}^{\textsf{RABE-PriVCD}}_{\mathcal{A}}(\lambda,b)=1]
      - \Pr[\mathsf{Hyb}(b)=1]\bigr| \le \mathsf{negl}(\lambda).
\end{align*}
\end{lemma}

\begin{proof}
Suppose that $\mathcal{A}$ distinguishes the two experiments with a non-negligible advantage. We construct a reduction $\mathcal{B}$ that breaks the security of $\textsf{Shad-RABE}$. The challenger of $\mathsf{Exp}^{\textsf{Shad-RABE}}_{\mathcal{B}}(\lambda,b')$ for $b'\in \{0,1\}$ provides $\mathsf{crs}$ to $\mathcal{B}$, which forwards it to $\mathcal{A}$ and initializes $C, H, D$ as in the security game. \vspace{0.2cm}

\noindent All queries of $\mathcal{A}$ are answered by relaying them to the $\textsf{Shad-RABE}$ challenger through its registration and corruption interfaces, with $\mathcal{B}$ maintaining $C, H, D$ consistently.\vspace{0.2cm}

\noindent Upon receiving the challenge input $(\mu_0,\mu_1,X^*)$ from $\mathcal{A}$, $\mathcal{B}$ samples $\mathsf{ske.sk} \leftarrow \mathsf{SKE\text{-}CD}.\mathsf{KeyGen}(1^\lambda)$ and submits $(X^*,\mathsf{ske.sk})$ to its own challenger. It obtains $(\mathsf{srabe.ct}^*, \{\mathsf{sk}_i^*\}_{\mathsf{pk}_i\in H}, \{\mathsf{hsk}_i\}_{\mathsf{pk}_i\in H})$ from the challenger, 
then computes 
\begin{align*}
    &\mathsf{ske.ct}^* \leftarrow \mathsf{SKE\text{-}CD}.\mathsf{Encrypt}(\mathsf{ske.sk},\mu_b),
\end{align*}
and returns $(\mathsf{srabe.ct}^*,\mathsf{ske.ct}^*)$ together with $\mathsf{vk}^* = \mathsf{ske.sk}$ to $\mathcal{A}$.\vspace{0.2cm}

\noindent If $\mathcal{A}$ later produces a certificate $\mathsf{cert}$, $\mathcal{B}$ verifies it using $\mathsf{SKE\text{-}CD}.\mathsf{Verify}(\mathsf{ske.sk},\mathsf{cert})$. On success, $\mathcal{B}$ forwards 
$\{\mathsf{sk}_i^*\}_{\mathsf{pk}_i\in H}$ and $\{\mathsf{hsk}_i\}_{\mathsf{pk}_i\in H}$ to $\mathcal{A}$; otherwise it returns $\perp$. Finally, $\mathcal{B}$ outputs whatever $\mathcal{A}$ outputs.\vspace{0.2cm}

\noindent If $b'=0$, the simulation is perfect and matches $\mathsf{Exp}^{\textsf{RABE-CD}}_{\mathcal{A}}(\lambda,b)$. If $b'=1$, the challenger runs the simulation algorithms, and the view of $\mathcal{A}$ is identical to $\mathsf{Hyb}(b)$. Thus, any non-negligible distinguishing advantage of $\mathcal{A}$ implies a break of $\textsf{Shad-RABE}$, completing the proof.
\end{proof}
\vspace{0.2cm}

\begin{lemma}\label{lem:hybrid-indistinguishability}
If $\mathsf{SKE\text{-}CD}$ satisfies one-time certified deletion, then for every \textsf{QPT} adversary $\mathcal{A},$ there is a negligible \(\mathsf{negl}\) such that
\[
\bigl|\Pr[\mathsf{Hyb}(0)=1]-\Pr[\mathsf{Hyb}(1)=1]\bigr| \le \mathsf{negl}(\lambda).
\]
\end{lemma}

 \begin{proof}
 Suppose \(\mathcal{A}\) distinguishes \(\mathsf{Hyb}(0)\) from \(\mathsf{Hyb}(1)\). We construct \(\mathcal{B}\) against one-time certified deletion in the experiment $\mathsf{Exp}^{\mathsf{OT-CD}}_{\textsf{SKE-CD}, \mathcal{B}}(\lambda, b')$ for $b'\in\{0,1\}.$

\noindent The reduction generates \(\mathsf{crs}\leftarrow \mathsf{Setup}(1^\lambda,1^{|\mathcal{U}|})\), forwards it to \(\mathcal{A}\), and answers all registration and corruption queries using the \(\textsf{Shad-RABE}\) simulation algorithms, keeping the same state as in \(\mathsf{Hyb}(\cdot)\). On receiving the challenge $(\mu_0,\mu_1,X^*)$ from $\mathcal{A}$, it submits $(\mu_0,\mu_1)$ to its own challenger and obtains $\mathsf{ske.ct}^*$ (encrypting $\mu_{b'}$). It then prepares $\mathsf{ct}^*_{\textsf{Shad}} \leftarrow \textsf{Shad-RABE.SimCT}(\mathsf{mpk},\mathsf{B},X^*)$ and returns the combined challenge ciphertext $\mathsf{ct}^*=(\mathsf{srabe.ct}^*,\mathsf{ske.ct}^*)$.\vspace{0.3cm}

\noindent If $\mathcal{A}$ outputs a certificate $\mathsf{cert}$, $\mathcal{B}$ forwards it to its challenger. On success, the challenger releases $\mathsf{ske.sk}$, which $\mathcal{B}$ uses to compute 
$\mathsf{sk}_i \leftarrow$ $\textsf{Shad-RABE.Reveal}(\mathsf{pk}_i,$ $\mathsf{B},$ $\mathsf{srabe.ct}^*,$ $\mathsf{ske.sk})$ for each $\mathsf{pk}_i \in H$, and forwards these to $\mathcal{A}$; otherwise $\perp$ is returned.\vspace{0.3cm}

\noindent When $b'=0$, the reduction perfectly simulates $\mathsf{Hyb}(0)$; when $b'=1$, the view coincides with $\mathsf{Hyb}(1)$. \vspace{0.3cm}

\noindent Thus, a distinguisher between $\mathsf{Hyb}(0)$ and $\mathsf{Hyb}(1)$ breaks the one-time certified deletion property of $\mathsf{SKE\text{-}CD}$, which concludes the lemma.
\end{proof}

\noindent Combining Lemma~\ref{lem:real-to-hybrid} and Lemma~\ref{lem:hybrid-indistinguishability}, we obtain
\begin{align*}
    \big|\Pr[\mathsf{Exp}^{\textsf{RABE-PriVCD}}_{\mathcal{A}}(\lambda,0)=1]-\Pr[\mathsf{Hyb}(0)=1]\big| &\leq \mathsf{negl}(\lambda),\\
    \big|\Pr[\mathsf{Hyb}(0)=1]-\Pr[\mathsf{Hyb}(1)=1]\big| &\leq \mathsf{negl}(\lambda),\\
    \big|\Pr[\mathsf{Hyb}(1)=1]-\Pr[\mathsf{Exp}^{\textsf{RABE-PriVCD}}_{\mathcal{A}}(\lambda,1)=1]\big| &\leq \mathsf{negl}(\lambda).
\end{align*}
\noindent Together, these bounds establish Theorem \ref{thm:rabecd-private security}.
\end{proof}

\subsection{Instantiation of \textsf{RABE-PriVCD}}\label{Instant-RABE-PriVCD}

\noindent We now describe a lattice-based instantiation of our \textsf{RABE-PriVCD}. The construction follows the generic construction developed in the preceding sections and combines state-of-the-art lattice-based primitives in order to achieve post-quantum security. In particular, the instantiation of \textsf{RABE-PriVCD} relies on two fundamental building blocks, both of which can be realized from standard and well-studied lattice assumptions.  

\begin{description}\setlength\itemsep{0.5em}
    \item[$\bullet$] \textbf{\textsf{Shad-RABE} Instantiation.} We instantiate the \textsf{Shad-RABE} component using lattice-based primitive as described in Section~\ref{Inst:Shadow-RABE}. The instantiation of \textsf{Shad-RABE} is secure under the $\ell$-succinct \textsf{LWE}, plain \textsf{LWE}, and equivocal \textsf{LWE} assumptions.

    \item[$\bullet$] \textbf{Symmetric Encryption with Certified Deletion.}  
    From Theorem~\ref{thm:otske-cd-existence}, we invoke the existence of an unconditionally secure one-time symmetric-key encryption scheme with certified deletion (\(\mathsf{SKE\text{-}CD}\)) \cite{broadbent2020quantum}. This primitive guarantees one-time certified deletion security, which is essential to the design of \textsf{RABE-PriVCD}.  
\end{description}  

\noindent Putting the above components together, we obtain a lattice-based instantiation of \textsf{RABE-PriVCD}. The resulting scheme achieves post-quantum security under the $\ell$-succinct \textsf{LWE}, plain \textsf{LWE}, and equivocal \textsf{LWE} assumptions, thereby demonstrating the plausibility of secure certified-deletion functionalities in the lattice setting. \vspace{0.2cm}

\section{\textsf{RABE} with Publicly Verifiable Certified Deletion}\label{subsection-RABE-PubVCD}

\vspace{0.2cm}
In this section, we provide a generic construction of registered attribute-based encryption with publicly verifiable certified deletion (\textsf{RABE-PubVCD}) and then describe its lattice-based instantiation. Our construction integrates three core cryptographic primitives: witness encryption (\textit{cf.} Section~\ref{def:we-framework}), one-shot signatures (\textit{cf.} Section \ref{def:oss-framework}), and the \textsf{Shad-RABE} protocol (\textit{cf.} Section \ref{generic-construct:Shadow-RABE}). \vspace{0.2cm}

\subsection{Generic Construction of \textsf{RABE-PubVCD}.}\label{generic-construct-RABE-PubVCD}

\noindent We construct a \textsf{RABE-PubVCD} protocol, denoted by $\prod^{\mathsf{RABE}}_{\textsf{PubVCD}} =$ $ (\mathsf{Setup},$ $\mathsf{KeyGen},$ $\mathsf{RegPK},$ $\mathsf{Encrypt},$ $\mathsf{Update},$ $\mathsf{Decrypt},$ $\mathsf{Delete},$ $\mathsf{Verify})$, 
from the following building blocks: 
a \textsf{Shad-RABE} scheme, specified by the algorithms $\prod_{\textsf{Shad-RABE}}=$ \textsf{Shad-RABE}.($\mathsf{Setup},$ $\mathsf{KeyGen},$ $\mathsf{RegPK},$ $\mathsf{Encrypt},$ $\mathsf{Update},$ $\mathsf{Decrypt},$ $\mathsf{SimKeyGen},$ $\mathsf{SimRegPK},$ $\mathsf{SimCorrupt},$ $\mathsf{SimCT},$ $\mathsf{Reveal}$), a witness encryption scheme $\prod_{\mathsf{we}} =$ $\textsf{WE}.(\mathsf{Encrypt},$ $\mathsf{Decrypt})$, and a one-shot signature scheme $\prod_{\mathsf{oss}} =$ $\textsf{OSS}.(\mathsf{Setup},$ $\mathsf{KeyGen},$ $\mathsf{Sign},$ $\mathsf{Verify})$. The detailed algorithms of \textsf{RABE-PubVCD} are specified below.\vspace{0.2cm}

\begin{description}\setlength\itemsep{0.5em}
\item[]\textsf{Setup}$(1^{\lambda}, 1^{{\tau}})$:
\leavevmode
\begin{itemize}
  \item[$-$] Generate $\mathsf{crs} \leftarrow \textsf{Shad-RABE}.\mathsf{Setup}(1^{\lambda}, 1^{{\tau}})$.
  \item[$-$] Output $\mathsf{crs}$.
\end{itemize}

\item[]\textsf{KeyGen}$(\mathsf{crs}, \mathsf{aux}, P)$:
\leavevmode
\begin{itemize}
  \item[$-$] Generate $(\mathsf{srabe.pk}, \mathsf{srabe.sk}) \leftarrow \textsf{Shad-RABE}.\mathsf{KeyGen}(\mathsf{crs}, \mathsf{aux}, P)$.
  \item[$-$] Output $(\mathsf{pk}, \mathsf{sk}) := (\mathsf{srabe.pk}, \mathsf{srabe.sk})$.
\end{itemize}

\item[]\textsf{RegPK}$(\mathsf{crs}, \mathsf{aux}, \mathsf{pk}, P)$:
\leavevmode
\begin{itemize}
  \item[$-$] Parse $\mathsf{pk} = \mathsf{srabe.pk}$.
  \item[$-$] Compute $(\mathsf{srabe.mpk}, \mathsf{aux}') \leftarrow \textsf{Shad-RABE}.\mathsf{RegPK}(\mathsf{crs}, \mathsf{aux}, \mathsf{srabe.pk}, P)$.
  \item[$-$] Output $(\mathsf{mpk}, \mathsf{aux}') := (\mathsf{srabe.mpk}, \mathsf{aux}')$.
\end{itemize}

\item[]\textsf{Encrypt}$\left( \mathsf{Sndr}(\mathsf{mpk}, X, \mu), \mathsf{Rcvr} \right)$:
\leavevmode
\begin{itemize}
  \item[$-$] $\mathsf{Sndr}$ parses $\mathsf{mpk} = \mathsf{srabe.mpk}$.
  \item[$-$] $\mathsf{Sndr}$ generates $\mathsf{oss.crs} \leftarrow \mathsf{OSS}.\mathsf{Setup}(1^{\lambda})$ and sends it to $\mathsf{Rcvr}$.
  \item[$-$] $\mathsf{Rcvr}$ generates $(\mathsf{oss.pk}, \mathsf{oss.sk}) \leftarrow \mathsf{OSS}.\mathsf{KeyGen}(\mathsf{oss.crs})$, sends $\mathsf{oss.pk}$ to $\mathsf{Sndr}$, and stores $\mathsf{oss.sk}$.
  \item[$-$] $\mathsf{Sndr}$ computes the witness encryption ciphertext $\mathsf{we.ct} \leftarrow$ $\mathsf{WE}.\mathsf{Encrypt}(1^{\lambda},$ $x, \mu)$, where
  \[
     x := \{ \exists \sigma \; \text{s.t.} \; \mathsf{OSS}.\mathsf{Verify}(\mathsf{oss.crs}, \mathsf{oss.pk}, \sigma, 0) = \top \}.
  \]
  \item[$-$] $\mathsf{Sndr}$ computes $\mathsf{srabe.ct} \leftarrow \textsf{Shad-RABE}.\mathsf{Encrypt}(\mathsf{srabe.mpk}, X, \mathsf{we.ct})$.
  \item[$-$] $\mathsf{Sndr}$ sends $\mathsf{srabe.ct}$ to $\mathsf{Rcvr}$ and outputs $\mathsf{vk} := (\mathsf{oss.crs}, \mathsf{oss.pk})$.
  \item[$-$] $\mathsf{Rcvr}$ outputs $\mathsf{ct} := (\mathsf{srabe.ct}, \mathsf{oss.sk})$.
\end{itemize}

\item[]\textsf{Update}$(\mathsf{crs}, \mathsf{aux}, \mathsf{pk})$:
\leavevmode
\begin{itemize}
  \item[$-$] Parse $\mathsf{pk} = \mathsf{srabe.pk}$.
  \item[$-$] Generate $\mathsf{srabe.hsk} \leftarrow \textsf{Shad-RABE}.\mathsf{Update}(\mathsf{crs}, \mathsf{aux}, \mathsf{srabe.pk})$.
  \item[$-$] Output $\mathsf{hsk} := \mathsf{srabe.hsk}$.
\end{itemize}

\item[]\textsf{Decrypt}$(\mathsf{sk}, \mathsf{hsk},{X}, \mathsf{ct})$:
\leavevmode
\begin{itemize}
  \item[$-$] Parse $\mathsf{sk} = \mathsf{srabe.sk}$, $\mathsf{hsk} = \mathsf{srabe.hsk}$, and $\mathsf{ct} = (\mathsf{srabe.ct}, \mathsf{oss.sk})$.
  \item[$-$] Compute $\sigma \leftarrow \mathsf{OSS}.\mathsf{Sign}(\mathsf{oss.sk}, 0)$.
  \item[$-$] Compute $\mathsf{we.ct}' \leftarrow \textsf{Shad-RABE}.\mathsf{Decrypt}(\mathsf{srabe.sk}, \mathsf{srabe.hsk},{X}, \mathsf{srabe.ct})$.
  \item[$-$] If $\mathsf{we.ct}' \in \{\perp, \mathsf{GetUpdate}\}$, output $\mathsf{we.ct}'$.
  \item[$-$] Otherwise, compute and output $\mu \leftarrow \mathsf{WE}.\mathsf{Decrypt}(\mathsf{we.ct}', \sigma)$.
\end{itemize}

\item[]\textsf{Delete}$(\mathsf{ct})$:
\leavevmode
\begin{itemize}
  \item[$-$] Parse $\mathsf{ct} = (\mathsf{srabe.ct}, \mathsf{oss.sk})$.
  \item[$-$] Compute $\sigma \leftarrow \mathsf{OSS}.\mathsf{Sign}(\mathsf{oss.sk}, 1)$.
  \item[$-$] Output certificate $\mathsf{cert} := \sigma$.
\end{itemize}

\item[]\textsf{Verify}$(\mathsf{vk}, \mathsf{cert})$:
\leavevmode
\begin{itemize}
  \item[$-$] Parse $\mathsf{vk} = (\mathsf{oss.crs}, \mathsf{oss.pk})$ and $\mathsf{cert} = \sigma$.
  \item[$-$] Compute and output $b \leftarrow \mathsf{OSS}.\mathsf{Verify}(\mathsf{oss.crs}, \mathsf{oss.pk}, \sigma, 1)$.
\end{itemize}
\end{description}

\noindent\textbf{Correctness.} The correctness of decryption and verification follows directly from the correctness of $\prod_{\mathsf{we}}$ and $\prod_{\mathsf{oss}}$.\vspace{0.2cm}

\subsection{Proof of Security}
We prove the security of our construction by showing the following theorem.

\begin{theorem}\label{thm:rabe-cd-security}
If the \textsf{Shad-RABE} scheme $\prod_{\textsf{Shad-RABE}}$ is secure, as discussed in Section \ref{def:shadow-rabe-security}, the witness encryption scheme $\prod_{\mathsf{WE}}$ satisfies extractable security, and the one-shot signature scheme $\prod_{\mathsf{OSS}}$ is secure, then the $\prod^{\textsf{RABE}}_{\textsf{PubVCD}}$ construction achieves certified deletion security with public verification.
\end{theorem}

\begin{proof}
We establish security via a sequence of hybrid games. As a starting point, we present the real security experiment $\mathsf{Exp}^{\textsf{RABE-PubVCD}}_{\mathcal{A}}(\lambda, b)$, which we denote as $\mathsf{Game}_0^{(b)}$.

\begin{description}\setlength\itemsep{0.5em}
    \item[$\bullet$] $\mathsf{Game}_0^{(b)}:$ The experiment proceeds as follows:\vspace{0.2cm}
    
    \begin{description}\setlength\itemsep{0.5em}
        \item[$-$] The challenger runs $\textsf{Shad-RABE}.\mathsf{Setup}(1^{\lambda}, 1^{|\mathcal{U}|})$ to obtain the common reference string $\mathsf{crs}$, initializes its internal state, and forwards $\mathsf{crs}$ to the adversary $\mathcal{A}$.

        \item[$-$] The adversary $\mathcal{A}$ adaptively issues registration and corruption queries, possibly involving both honest and malicious keys. The challenger responds while updating the master public key $\mathsf{mpk}$ and its internal state \textsf{aux}.

        \item[$-$] $\mathcal{A}$ submits a pair of challenge messages $(\mu_0^*, \mu_1^*) \in \mathcal{M}^2$ and a challenge attribute set $X^* \subseteq \mathcal{U}$. Then, the challenger and adversary engage in the interactive encryption protocol:\vspace{0.2cm}
    
        \begin{itemize}\setlength\itemsep{0.5em}
            \item[(a)] The challenger samples $\mathsf{oss.crs} \leftarrow \mathsf{OSS}.\mathsf{Setup}(1^{\lambda})$ and sends it to $\mathcal{A}$.
        
            \item[(b)] The adversary generates a key pair $(\mathsf{oss.pk}, \mathsf{oss.sk})$ $\leftarrow$ $\mathsf{OSS}.\mathsf{KeyGen}$ $(\mathsf{oss.crs})$ and returns $\mathsf{oss.pk}$ to the challenger.
        
            \item[(c)] The challenger generates a ciphertext through witness encryption $\mathsf{we.ct}$ $\leftarrow$ $\mathsf{WE}.\mathsf{Encrypt}(1^{\lambda},$ $x,$ $\mu_b^*)$, where the statement \(x\) asserts the existence of a valid one-shot signature, \textit{i.e.,}
            \[
                x = \left\{\exists \sigma : \mathsf{OSS}.\mathsf{Verify}(\mathsf{oss.crs}, \mathsf{oss.pk}, \sigma, 0) = \top \right\}.
            \]
            
            \item[(d)] The challenger encrypts the witness ciphertext under the \textsf{Shad-RABE} scheme:
            \[
                \mathsf{srabe.ct} \leftarrow \textsf{Shad-RABE}.\mathsf{Encrypt}(\mathsf{mpk}, X^*, \mathsf{we.ct}).
            \]
        
            \item[(e)] The challenger sends $\mathsf{srabe.ct}$ to $\mathcal{A}$ and stores the verification key $\mathsf{vk}^* := (\mathsf{oss.crs}, \mathsf{oss.pk})$.
        \end{itemize}

        \item[$-$] $\mathcal{A}$ returns a deletion certificate $\mathsf{cert} = \sigma$ to the challenger.

        \item[$-$] The challenger verifies the certificate by computing $\mathsf{valid}$ $\leftarrow$ $\mathsf{OSS}.\mathsf{Verify}(\mathsf{oss.crs},$ $\mathsf{oss.pk},$ $\sigma,$ $1)$. If $\mathsf{valid} = 0$, it returns $\perp$ to $\mathcal{A}$ and halts. Otherwise, it provides all honest secret keys to $\mathcal{A}$.

        \item[$-$] Finally, the adversary outputs a bit $b' \in \{0, 1\}$.
    \end{description}
    
    \item[$\bullet$] $\mathsf{Game}_1^{(b)}:$ This game is identical to $\mathsf{Game}_0^{(b)}$, except that all operations involving the \textsf{Shad-ABE} scheme are performed using its simulation algorithms. Specifically:

    \begin{description}\setlength\itemsep{0.5em}
        \item[$-$] The challenger replaces the real key generation procedure with its simulator, \textit{i.e.,} honest key pairs are produced via $(\mathsf{pk}, \mathsf{B})$ $\leftarrow$ \textsf{Shad-RABE.SimKey}- $\textsf{Gen}(\mathsf{crs},$ $\mathsf{aux}, P, \textsf{B})$ instead of (\textsf{pk}, \textsf{sk}) $\leftarrow$ ${\textsf{Shad-RABE.KeyGen}(\textsf{crs}, \textsf{aux}, P)}.$

        \item[$-$] The honest public keys are registered using the simulation procedure $(\widetilde{\textsf{mpk}},$ $\widetilde{\mathsf{aux}})$ $\leftarrow$ $\textsf{Shad-RABE.SimRegPK}(\textsf{crs},$ $\textsf{aux},$ $\mathsf{pk},$ $P)$ in place of the real registration algorithm $(\mathsf{mpk}, \mathsf{aux}') \leftarrow \textsf{Shad-RABE.RegPK}(\mathsf{crs}, \mathsf{aux}, \mathsf{pk}, P).$

        \item[$-$] The challenge ciphertext is produced via the simulation procedure $\widetilde{\mathsf{ct}}$ $\leftarrow$ $\textsf{Shad-RABE.SimCT}(\widetilde{\mathsf{mpk}}, \mathsf{B}, X)$ rather than the original encryption algorithm $\mathsf{ct}\leftarrow \textsf{Shad-RABE.Encrypt}(\mathsf{mpk}, X, \mu).$

        \item[$-$] For corruption queries, the challenger answers with simulated secret keys, computed as $\widetilde{\mathsf{sk}} \leftarrow \textsf{Shad-RABE.SimCorrupt}(\mathsf{crs}, \mathsf{pk}, \mathsf{B})$, instead of disclosing the actual secret keys.

        \item[$-$] In the deletion phase, once verification succeeds, the challenger derives the honest secret keys using $\widetilde{\mathsf{sk}}$ $\leftarrow$ $\textsf{Shad-RABE.Reveal}(\mathsf{pk},$ $\mathsf{B},$ $\widetilde{\mathsf{ct}},$ $\mathsf{we.ct})$, ensuring they correspond to the correct challenge message.
    \end{description}    
\end{description}

\begin{lemma}\label{lem:rabe-cd-shadow-transition}
If $\prod_{\textsf{Shad-RABE}}$ is secure as defined in Section~\ref{def:shadow-rabe-security}, then for every bit $b \in \{0,1\}$,
\begin{align*}
    &\left|\Pr[\mathsf{Game}_0^{(b)} = 1] - \Pr[\mathsf{Game}_1^{(b)} = 1]\right| \leq \mathsf{negl}(\lambda).
\end{align*}
\end{lemma}
\begin{proof}
Assume, for the sake of contradiction, that
\[
\left| \Pr[\mathsf{Game}_0^{(b)} = 1] - \Pr[\mathsf{Game}_1^{(b)} = 1] \right|
\]
is non-negligible. We then construct a reduction adversary $\mathcal{B}_{\textsf{Shad-RABE}}$ that breaks the security of the \textsf{Shad-RABE} scheme $\prod_{\textsf{Shad-RABE}}$. Let $\mathcal{A}$ be the adversary that distinguishes between $\mathsf{Game}_0^{(b)}$ and $\mathsf{Game}_1^{(b)}$. The reduction $\mathcal{B}_{\textsf{Shad-RABE}}$ operates as follows:

\begin{description}\setlength\itemsep{0.5em}
    \item[$\bullet$] \textbf{Setup Phase:} Upon receiving a tuple $(\mathsf{pk}, P)$ from $\mathcal{A}$, $\mathcal{B}_{\textsf{Shad-RABE}}$ forwards the query to the \textsf{Shad-RABE} challenger and relays the response $(\mathsf{mpk}', \mathsf{aux}')$ to $\mathcal{A}$.

    \item[$\bullet$] \textbf{Registration Phase:} Upon receiving a tuple $(\mathsf{pk}, P)$ from $\mathcal{A}$, $\mathcal{B}_{\textsf{Shad-RABE}}$ forwards the query to the \textsf{Shad-RABE} challenger and relays the response $(\mathsf{mpk}', \mathsf{aux}')$ to $\mathcal{A}$.\vspace{0.2cm}
    
    \begin{itemize}\setlength\itemsep{0.5em}
        \item[$-$] {Corrupted key registration:} Whenever $\mathcal{A}$ submits a tuple $(\mathsf{pk}, P)$, the reduction $\mathcal{B}_{\textsf{Shad-RABE}}$ forwards this query to the \textsf{Shad-RABE} challenger and returns the challenger’s reply $(\mathsf{mpk}', \mathsf{aux}')$ to $\mathcal{A}$.
    
        \item[$-$] {Honest key registration:} Upon receiving an access policy $P$, $\mathcal{B}_{\textsf{Shad-RABE}}$ forwards $P$ to the \textsf{Shad-RABE} challenger and receives $(\mathsf{ctr}, \mathsf{mpk}', \mathsf{aux}', \mathsf{pk}_{\mathsf{ctr}})$, which it returns to $\mathcal{A}$.
    
        \item[$-$] {Corruption query:} When $\mathcal{A}$ issues a corruption query, $\mathcal{B}_{\textsf{Shad-RABE}}$ forwards the request to the \textsf{Shad-RABE} challenger and provides $\mathcal{A}$ with the simulated secret key returned.
    \end{itemize}

    \item[$\bullet$] \textbf{Challenge Phase:} $\mathcal{A}$ submits challenge messages $(\mu_0^*, \mu_1^*)$ and an attribute set $X^*$. The reduction proceeds as follows:\vspace{0.2cm}
    
    \begin{itemize}\setlength\itemsep{0.5em}
        \item[$-$] $\mathcal{B}_{\textsf{Shad-RABE}}$ executes $\mathsf{OSS}.\mathsf{Setup}(1^{\lambda})$ to obtain a common reference string $\mathsf{oss.crs}$, which is then given to $\mathcal{A}$.
    
        \item[$-$] $\mathcal{A}$ computes a key pair $(\mathsf{oss.pk}, \mathsf{oss.sk})$ via $\mathsf{OSS}.\mathsf{KeyGen}(\mathsf{oss.crs})$ and reveals only the public key $\mathsf{oss.pk}$.
    
        \item[$-$] To generate the challenge ciphertext, $\mathcal{B}_{\textsf{Shad-RABE}}$ invokes witness encryption as $\mathsf{we.ct} \gets \mathsf{WE}.\mathsf{Encrypt}(1^{\lambda}, x, \mu_b^*)$, with the statement
        \begin{align*}
            &x := \left\{ \exists \sigma \text{ s.t. } \mathsf{OSS}.\mathsf{Verify}(\mathsf{oss.crs}, \mathsf{oss.pk}, \sigma, 0) = \top \right\}.
        \end{align*}
        
        \item[$-$] $\mathcal{B}_{\textsf{Shad-RABE}}$ submits the challenge $(\mathsf{we.ct}, X^*)$ to the \textsf{Shad-RABE} challenger and obtains:
        \begin{align*}
            \mathsf{srabe.ct}^* &=
            \begin{cases}
                \textsf{Shad-RABE}.\textsf{Encrypt}(\mathsf{mpk}, X^*, \mathsf{we.ct}) &~~~~~~ \text{if } b = 0, \\[0.75em]
                \textsf{Shad-RABE}.\textsf{SimCT}(\mathsf{mpk}, \textsf{B},X^*) &~~~~~~ \text{if } b = 1,
            \end{cases}\\
            \mathsf{reveal} &=
            \begin{cases}
                \{\mathsf{sk}_i : \mathsf{pk}_i \in H\} & \text{if } b = 0, \\[0.75em]
                \textsf{Shad-RABE}.\textsf{Reveal}(\mathsf{pk}, \textsf{B}, \mathsf{srabe.ct}^*, \mathsf{we.ct}) &\text{if } b = 1.
            \end{cases}
        \end{align*}

    \item[$-$] $\mathcal{B}_{\textsf{Shad-RABE}}$ forwards $\mathsf{srabe.ct}^*$ to $\mathcal{A}$ and stores the verification key $\mathsf{vk}^* := (\mathsf{oss.crs}, \mathsf{oss.pk})$.
    \end{itemize}

    \item[$\bullet$] \textbf{Deletion Phase:} $\mathcal{A}$ submits a deletion certificate $\mathsf{cert} = \sigma$. The reduction verifies it as follows:
    \[
        \mathsf{OSS}.\mathsf{Verify}(\mathsf{oss.crs}, \mathsf{oss.pk}, \sigma, 1).
    \]
    If the verification fails, $\mathcal{B}_{\textsf{Shad-RABE}}$ returns $\perp$ to $\mathcal{A}$. Otherwise, it sends the value $\mathsf{reveal}$ obtained from the \textsf{Shad-RABE} challenger.

    \item[$\bullet$] \textbf{Output:} $\mathcal{B}_{\textsf{Shad-RABE}}$ outputs the bit returned by $\mathcal{A}$.
\end{description}

\noindent When the internal bit maintained by the \textsf{Shad-RABE} challenger is fixed to $0$, all computations are performed using the real algorithms of the scheme. In this case, $\mathcal{B}_{\textsf{Shad-RABE}}$ provides a perfect simulation of $\mathsf{Game}_0^{(b)}$ for the adversary $\mathcal{A}$. In contrast, when the internal bit is set to $1$, the challenger consistently employs the corresponding simulation procedures, and the resulting transcript is distributed identically to $\mathsf{Game}_1^{(b)}$. \vspace{0.2cm}

\noindent Since, by hypothesis, the adversary $\mathcal{A}$ can distinguish between $\mathsf{Game}_0^{(b)}$ and $\mathsf{Game}_1^{(b)}$ with non-negligible advantage, it follows that $\mathcal{B}_{\textsf{Shad-RABE}}$ can also distinguish the real from the simulated execution in the \textsf{Shad-RABE} security experiment with the same advantage. This directly contradicts the assumed security of $\prod_{\textsf{Shad-RABE}}$.
\end{proof}

\begin{lemma}\label{lem:rabe-cd-message-independence}
If the one-shot signature scheme $\prod_{\mathsf{OSS}}$ is secure and the witness encryption scheme $\prod_{\mathsf{WE}}$ achieves extractable security, then the distinguishing advantage between the two games is negligible, \textit{i.e.,}
\begin{align*}
    &\left|\Pr[\mathsf{Game}_1^{(0)} = 1] - \Pr[\mathsf{Game}_1^{(1)} = 1]\right| \leq \mathsf{negl}(\lambda).
\end{align*}
\end{lemma}
\begin{proof}
We prove this lemma by analyzing the adversary’s strategy with respect to the validity of the deletion certificate. Let $\mathsf{ValidCert}$ denote the event that adversary $\mathcal{A}$ provides a valid deletion certificate $\sigma$ such that $\mathsf{OSS}.\mathsf{Verify}(\mathsf{oss.crs},$ $\mathsf{oss.pk},$ $\sigma,$ $1) =$ $\top$. We can decompose the distinguishing advantage as:

\begin{align*}
&\left|\Pr\left[\mathsf{Game}_1^{(0)} = 1\right] - \Pr\left[\mathsf{Game}_1^{(1)} = 1\right]\right| \\
= &\left|\Pr\left[\mathsf{Game}_1^{(0)} = 1 \mid \mathsf{ValidCert}\right] \cdot \Pr\left[\mathsf{ValidCert}\right] + \Pr\left[\mathsf{Game}_1^{(0)} = 1 \mid \neg\mathsf{ValidCert}\right] \right.\\
&\left.\cdot \Pr\left[\neg\mathsf{ValidCert}\right]- \Pr\left[\mathsf{Game}_1^{(1)} = 1 \mid \mathsf{ValidCert}\right] \cdot \Pr\left[\mathsf{ValidCert}\right] \right.\\
&\left.- \Pr\left[\mathsf{Game}_1^{(1)} = 1 \mid \neg\mathsf{ValidCert}\right] \cdot \Pr\left[\neg\mathsf{ValidCert}\right]\right| \\
\leq &\left|\Pr\left[\mathsf{Game}_1^{(0)} = 1 \mid \mathsf{ValidCert}\right] - \Pr\left[\mathsf{Game}_1^{(1)} = 1 \mid \mathsf{ValidCert}\right]\right| \cdot \Pr\left[\mathsf{ValidCert}\right].
\end{align*}

\noindent The last inequality follows from the triangle inequality together with the observation that, whenever $\neg\mathsf{ValidCert}$ holds, the two games are indistinguishable since both return $\perp$ to $\mathcal{A}$.\vspace{0.2cm}

\noindent Assume that the above distinguishing advantage is non-negligible. Then there exists an infinite subset $I \subseteq \mathbb{N}$ and a polynomial $p(\cdot)$ such that for every $\lambda \in I$,

\begin{align*}
    &\left|\Pr\left[\mathsf{Game}_1^{(0)} = 1 \mid \mathsf{ValidCert}\right] - \Pr\left[\mathsf{Game}_1^{(1)} = 1 \mid \mathsf{ValidCert}\right]\right| \geq \frac{1}{p(\lambda)}\\
    &\Pr\left[\mathsf{ValidCert}\right] \geq \frac{1}{p(\lambda)}.
\end{align*}

\noindent Let $\mathsf{st}$ represent $\mathcal{A}$'s internal state after producing the valid certificate, conditioned on $\mathsf{ValidCert}$.From the extractable security of witness encryption, there exists a QPT extractor $\mathcal{E}$ and a polynomial $q(\cdot)$ such that, for every $\lambda \in I,$
\begin{align*}
    \Pr\left[\mathcal{E}(1^{\lambda}, x, \mathsf{st})=\sigma' \wedge \mathsf{OSS}.\mathsf{Verify}(\mathsf{oss.crs}, \mathsf{oss.pk}, \sigma', 0) = \top\right] \geq \frac{1}{q(\lambda)}.
\end{align*}

\noindent We construct adversary $\mathcal{B}_{\mathsf{oss}}$ that breaks the security of the one-shot signature scheme as follows. $\mathcal{B}_{\mathsf{oss}}$ is given the common reference string $\mathsf{oss.crs}$ from its challenger. It then internally runs $\mathsf{Game}_1^{(b)}$ with adversary $\mathcal{A}$, embedding $\mathsf{oss.crs}$ into the simulated encryption protocol. Once $\mathcal{A}$ outputs a certificate $\sigma$, $\mathcal{B}_{\mathsf{oss}}$ invokes the extractor $\mathcal{E}$ to derive another certificate $\sigma'$. At this point, $\mathcal{B}_{\mathsf{oss}}$ produces the forgery consisting of $(\mathsf{oss.pk}, 1, \sigma, 0, \sigma')$.\vspace{0.2cm}

\noindent The probability that $\mathcal{B}_{\mathsf{oss}}$ succeeds in this forgery is: 
\begin{align*}
    &\Pr\left[\mu_0^* \neq \mu_1^* \wedge \mathsf{OSS}.\mathsf{Verify}(\mathsf{oss.crs}, \mathsf{oss.pk}, \sigma', 0) = \top \right.\\
    &\left.~~~~~~~~~~~~~~~\wedge~\mathsf{OSS}.\mathsf{Verify}(\mathsf{oss.crs}, \mathsf{oss.pk}, \sigma, 1) = \top\right]\\
    =&\Pr\left[\mathsf{ValidCert}\right]\cdot\Pr\left[\mathcal{E}(1^{\lambda}, x, \mathsf{st})=\sigma' \wedge \mathsf{OSS}.\mathsf{Verify}(\mathsf{oss.crs}, \mathsf{oss.pk}, \sigma', 0) = \top\right]\\
    \geq &~\frac{1}{p(\lambda) q(\lambda)}.
\end{align*}
This contradicts the assumed security of $\prod_{\mathsf{oss}}$, thereby completing the proof of the lemma. \vspace{0.2cm}
\end{proof}

\noindent We thus conclude the proof of Theorem~\ref{thm:rabe-cd-security}.
\end{proof}

\subsection{Instantiation of \textsf{RABE-PubVCD}}  
We now describe a lattice instantiation of our \textsf{RABE-PubVCD}. The construction follows the generic construction developed in the preceding sections. In particular, the instantiation of \textsf{RABE-PubVCD} relies on three fundamental building blocks, \textsf{Shad-RABE} protocol, \textsf{WE} protocol, \textsf{OSS} scheme.  

\begin{description}\setlength \itemsep{0.5em}
    \item[$\bullet$] \textbf{\textsf{Shad-RABE} Instantiation:}  
    We realize the \textsf{Shad-RABE} component from lattice-based primitives, as detailed in Section~\ref{Inst:Shadow-RABE}. The resulting instantiation is secure under the $\ell$-succinct \textsf{LWE}, plain \textsf{LWE}, and equivocal \textsf{LWE} assumptions.  

    \item[$\bullet$] \textbf{\textsf{WE} Instantiation:}  
    For extractable witness encryption, candidate constructions have been proposed in the works of~\cite{vaikuntanathan2022witness,tsabary2022candidate} from the \textsf{LWE} assumption and its variants, offering significantly better efficiency than earlier $i\mathcal{O}$-based approaches.  

    \item[$\bullet$] \textbf{\textsf{OSS} Instantiation:}  
    One-shot signatures were originally introduced by Amos \textit{et al.}~\cite{amos2020one}. More recently, Shmueli \textit{et al.}~\cite{shmueli2025one} provided a standard-model construction secure under (sub-exponential) indistinguishability obfuscation ($i\mathcal{O}$) and \textsf{LWE}. Their work also establishes the first standard-model separation between classical and collapse-binding post-quantum commitments, and further introduces an oracle-based scheme with unconditional security.  
\end{description}  

\noindent Combining these ingredients yields a lattice-based instantiation of \textsf{RABE-PubVCD}. The resulting scheme enjoys post-quantum security under $\ell$-succinct \textsf{LWE}, plain \textsf{LWE}, and equivocal \textsf{LWE}, thereby demonstrating the feasibility of certified-deletion functionalities in the lattice setting. \vspace{0.2cm}

%% file: 7-8_RABE-CED.tex
\section{\textsf{RABE} with Privately Verifiable Certified Everlasting Deletion}\label{sec:RABE-PrivCED}
In this section, we propose a construction of registered attribute-based encryption with privately verifiable certified everlasting deletion (\textsf{RABE-PriVCED}). The construction extends the \textsf{RABE} framework by modifying the encryption algorithm to produce a hybrid quantum-classical ciphertext. Concretely, the encryption procedure samples random strings $x$ and $\theta$, and prepares the quantum state $\lvert x \rangle_{\theta}$, where each qubit is encoded in either the computational or the Hadamard basis according to $\theta_i$. The plaintext is concealed using a one-time pad derived from the positions where $\theta_i = 0$. Subsequently, the masked message together with the basis information $\theta$ is encrypted classically, whereas the quantum state can be irreversibly destroyed during the deletion process. In the following subsection, we provide a generic construction of \textsf{RABE-PriVCED}.


\subsection{Generic Construction of \textsf{RABE-PriVCED}}
We construct a registered attribute-based encryption scheme with privately verifiable certified everlasting deletion (\textsf{RABE-PriVCED}) protocol, denoted by 
$\prod^{\mathsf{RABE}}_{\textsf{PriVCED}} =$ (\textsf{Setup}, \textsf{KeyGen},\textsf{RegPK}, \textsf{Encrypt}, \textsf{Update}, \textsf{Decrypt}, \textsf{Delete}, \textsf{Ver}), 
from the registered attribute-based encryption scheme $\prod_{\textsf{RABE}} =$ (\textsf{Setup}, \textsf{KeyGen}, \textsf{RegPK}, \textsf{Encrypt}, \textsf{Update}, \textsf{Decrypt}). The detailed algorithms of \textsf{RABE-PriVCED} are specified below.

\begin{description}\setlength\itemsep{0.5em}
\item[]\textsf{Setup}($1^{\lambda}, 1^{{\tau}}$):
\leavevmode
\begin{itemize}
    \item[$-$] Generate $\mathsf{crs} \leftarrow \mathsf{RABE}.\mathsf{Setup}(1^{\lambda}, 1^{{\tau}})$.
    \item[$-$] Output $\mathsf{crs}$.
\end{itemize}

\item[]\textsf{KeyGen}(\textsf{crs}, \textsf{aux}):
\leavevmode
\begin{itemize}
    \item[$-$] Generate $(\mathsf{pk}, \mathsf{sk}) \leftarrow \mathsf{RABE}.\mathsf{KeyGen}(\mathsf{crs}, \mathsf{aux})$.
    \item[$-$] Output $(\mathsf{pk}, \mathsf{sk})$.
\end{itemize}

\item[]\textsf{RegPK}(\textsf{crs}, \textsf{aux}, \textsf{pk}, $P$):
\leavevmode
\begin{itemize}
    \item[$-$] Compute $(\mathsf{mpk}, \mathsf{aux}') \leftarrow \mathsf{RABE}.\mathsf{RegPK}(\mathsf{crs}, \mathsf{aux}, \mathsf{pk}, P)$.
    \item[$-$] Output $(\mathsf{mpk}, \mathsf{aux}')$.
\end{itemize}

\item[]\textsf{Update}(\textsf{crs}, \textsf{aux}, \textsf{pk}):
\leavevmode
\begin{itemize}
    \item[$-$] Compute $\mathsf{hsk} \leftarrow \mathsf{RABE}.\mathsf{Update}(\mathsf{crs}, \mathsf{aux}, \mathsf{pk})$.
    \item[$-$] Output $\mathsf{hsk}$.
\end{itemize}

\item[]\textsf{Encrypt}(\textsf{mpk}, X, $b$):
\leavevmode
\begin{itemize}
    \item[$-$] Sample $x, \theta \leftarrow \{0,1\}^{\lambda}$ uniformly at random.
    \item[$-$] Compute $\mathsf{ct}_{\textsf{rabe}} \leftarrow \mathsf{RABE}.\mathsf{Encrypt}\left(\mathsf{mpk}, X, \left(\theta, b \oplus \bigoplus_{i:~\theta_i=0} x_i\right)\right)$.
    \item[$-$] Output $\mathsf{ct} := \left(|x\rangle_{\theta}, \mathsf{ct}_{\textsf{rabe}}\right)$ and $\mathsf{vk} := (x, \theta)$.
\end{itemize}

\item[]\textsf{Decrypt}(\textsf{sk}, \textsf{hsk},{X}, \textsf{ct}):
\leavevmode
\begin{itemize}
    \item[$-$] Parse $\mathsf{ct} := (|x\rangle_{\theta}, \mathsf{ct}_{\textsf{rabe}})$.
    \item[$-$] Compute $(\theta, b') \leftarrow \mathsf{RABE}.\mathsf{Decrypt}(\mathsf{sk}, \mathsf{hsk},{X}, \mathsf{ct}_{\textsf{rabe}})$.
    \item[$-$] If $\mathsf{RABE}.\mathsf{Decrypt}$ returns $\perp$ or $\mathsf{GetUpdate}$, output the same.
    \item[$-$] Measure $|x\rangle_{\theta}$ in the $\theta$-basis to obtain $x$.
    \item[$-$] Output $b = b' \oplus \bigoplus_{i:\theta_i=0} x_i$.
\end{itemize}

\item[]\textsf{Delete}(\textsf{ct}):
\leavevmode
\begin{itemize}
    \item[$-$] Parse $\mathsf{ct} := (|x\rangle_{\theta}, \mathsf{ct}_{\textsf{rabe}})$.
    \item[$-$] Measure $|x\rangle_{\theta}$ in the Hadamard basis to obtain a string $x'$.
    \item[$-$] Output $\mathsf{cert} := x'$.
\end{itemize}

\item[]\textsf{Ver}(\textsf{vk}, \textsf{cert}):
\leavevmode
\begin{itemize}
    \item[$-$] Parse $\mathsf{vk}$ as $(x, \theta)$ and $\mathsf{cert}$ as $x'$.
    \item[$-$] Output $\top$ if and only if $x_i = x'_i$ for all $i$ such that $\theta_i = 1$.
\end{itemize}
\end{description}

\noindent\textbf{Correctness of \textsf{RABE-PriVCED}.} The correctness of the \textsf{RABE-PriVCED} protocol is ensured by the scheme’s description together with the correctness of \textsf{RABE}. \vspace{0.2cm}

\noindent It is important to note that the entire construction of \textsf{RABE-PriVCED} is based solely on the underlying \textsf{RABE} primitive, which can be instantiated from lattices. As discussed earlier, we employ the key-policy registered \textsf{ABE} scheme of Champion \textit{et al.} \cite{champion2025registered}, supporting arbitrary bounded-depth circuit policies. Consequently, \textsf{RABE-PriVCED} admits a concrete instantiation based on standard lattice assumptions.

\subsection{Proof of Security for \textsf{RABE-PriVCED}}

In what follows, we provide a detailed proof of security for our \textsf{RABE-PriVCED} construction. The argument relies on the security of the underlying \textsf{RABE} scheme under the \textsf{LWE} assumption, together with the certified everlasting property established earlier.

\begin{theorem}\label{thm:rabe-privcd-security-proof}
If the underlying \textsf{RABE} scheme is secure under the \textsf{LWE} assumption, then the proposed \textsf{RABE-PriVCED} scheme achieves certified everlasting security.
\end{theorem}

\begin{proof}
The theorem is proved via a standard argument of indistinguishability. Consider the following two experiments, parameterized by a bit $b \in \{0,1\}$, 
\[
    \mathsf{Exp}^{\textsf{RABE-PriVCED}}_{\mathcal{A}}(\lambda, b).
\]  
The adversary $\mathcal{A}$ interacts with the challenger in one of these experiments, and its goal is to distinguish whether it interacts with the $b=0$ or $b=1$ world. The security requires that these two experiments be computationally indistinguishable. \vspace{0.2cm}

\noindent Formally, it suffices to show that the statistical distance
\begin{align*}
    &\mathsf{TD}\left(\mathsf{Exp}^{\textsf{RABE-PriVCED}}_{\mathcal{A}}(\lambda, 0), ~\mathsf{Exp}^{\textsf{RABE-PriVCED}}_{\mathcal{A}}(\lambda, 1)
    \right) \leq \mathsf{negl}(\lambda),
\end{align*}
for every non-uniform \textsf{QPT} adversary $\mathcal{A}$. \vspace{0.2cm} 

\noindent By Lemma \ref{lem:certified-everlasting-bk}, we already know that once deletion has been verified, the adversary cannot recover information about the challenge message with a non-negligible probability. This ensures that even after the deletion phase, the system retains everlasting confidentiality.\vspace{0.2cm}

\noindent The semantic security of the underlying $\mathsf{RABE}$ construction, based on the \textsf{LWE} assumption, ensures that encryptions of two distinct challenge messages cannot be distinguished by any efficient adversary, as long as the challenge policy restrictions are satisfied. \vspace{0.2cm}

\noindent To formalize the adversarial view, we define the distribution $Z_\lambda(\theta, b', X, \mathcal{A})$ as follows:
\begin{description}\setlength\itemsep{0.5em}
    \item[$-$] The challenger first samples the common reference string  
    \begin{align*}
       &\mathsf{crs} \gets \mathsf{Setup}(1^\lambda, 1^{\tau}),
    \end{align*}
   where $\tau$ denotes the policy parameter.  
   
   \item[$-$] Using the registration procedure, the challenger generates all registered keys and derives the master public key $\mathsf{mpk}$.  

    \item[$-$] Finally, the challenger outputs  
    \begin{align*}
       &\big(\mathcal{A},\, \mathsf{RABE}.\mathsf{Encrypt}(\mathsf{mpk}, X, (\theta, b'))\big),
    \end{align*} 
   where $(\theta, b')$ encodes the challenge message and the policy $P$. 
\end{description}

\noindent This distribution captures the complete view of the adversary during the challenge phase.\vspace{0.2cm}

\noindent The class of adversaries $\mathcal{A}$ considered here consists of non-uniform quantum polynomial-time families
\begin{align*}
    \{\mathcal{A}_\lambda,\, |\psi_\lambda\rangle\}_{\lambda \in \mathbb{N}},
\end{align*} 
which may include auxiliary quantum advice states $|\psi_\lambda\rangle$. These adversaries are constrained to respect the policy restrictions defined in the registered \textsf{ABE} security experiment (\textit{i.e.,} they cannot query for decryption keys that would trivially satisfy the challenge policy).\vspace{0.2cm}

\noindent By combining the above steps, we conclude that any distinguishing advantage of $\mathcal{A}$ in the above experiment is bounded by a negligible function in $\lambda$. Hence,
\begin{align*}
    \mathsf{TD}\left(\mathsf{Exp}^{\textsf{RABE-PriVCED}}_{\mathcal{A}}(\lambda, 0),~\mathsf{Exp}^{\textsf{RABE-PriVCED}}_{\mathcal{A}}(\lambda, 1)\right) \leq \textsf{negl}(\lambda).
\end{align*} 
This completes the proof.
\end{proof}

\section{\textsf{RABE} with Publicly Verifiable Certified Everlasting Deletion}\label{sec:RABE-PubVCED}

In this section, we present a construction of registered attribute-based encryption with publicly verifiable certified everlasting deletion (\textsf{RABE-PubVCED}) scheme. Analogously to the privately verifiable setting described in the previous section, our construction is based on the \textsf{RABE} framework combined with a signature scheme that satisfies the one-time unforgeability for the BB84 states, with a modified encryption algorithm. Specifically, the encryption procedure samples random strings $x$ and $\theta$ together with a signature key pair $(\mathsf{vk}, \mathsf{sigk})$, and prepares a quantum state $\lvert \psi \rangle$ by applying a signing operation to the quantum one-time pad state $\lvert x \rangle_{\theta}$. The plaintext message is masked using the bits at positions where $\theta_i = 1$. The signing key $\mathsf{sigk}$, the basis information $\theta$, and the masked message are then encrypted using the underlying $\mathsf{RABE}$ scheme, while the public verification key $\mathsf{vk}$ allows anyone to verify deletion certificates via the signature scheme.


\subsection{Generic Construction of \textsf{RABE-PubVCED}}
We present a registered attribute-based encryption scheme with publicly verifiable certified everlasting deletion (\textsf{RABE-PubVCED}), denoted as $\prod^{\mathsf{RABE}}_{\textsf{PubVCED}} =$ \textsf{RABE}.(\textsf{Setup}, \textsf{KeyGen},\textsf{RegPK}, \textsf{Encrypt}, \textsf{Update}, \textsf{Decrypt}, \textsf{Delete}, \textsf{Verify}) defined over an attribute universe $\mathcal{U}$, policy space $\mathcal{P}$, and message space $\{0,1\}$. Our construction utilizes two core cryptographic components: the registered attribute-based encryption scheme $\prod_{\textsf{RABE}} =$ \textsf{RABE}.(\textsf{Setup}, \textsf{KeyGen}, \textsf{RegPK}, \textsf{Encrypt}, \textsf{Update}, \textsf{Decrypt}) and a deterministic signature scheme $\prod_{\mathsf{SIG}} = \textsf{SIG}.(\mathsf{Gen}, \mathsf{Sign}, \mathsf{Verify})$. The algorithms comprising \textsf{RABE-PubVCED} are described in detail below.\vspace{0.2cm}

\begin{description}\setlength\itemsep{0.5em}
\item[]\textsf{Setup}($1^{\lambda}, 1^{{\tau}}$):
\leavevmode
\begin{itemize}
    \item[$-$] Output $\mathsf{crs} \leftarrow \mathsf{RABE.Setup}(1^{\lambda}, 1^{{\tau}})$.
\end{itemize}

\item[]\textsf{KeyGen}(\textsf{crs}, \textsf{aux}, $P$):
\leavevmode
\begin{itemize}
    \item[$-$] Output $(\mathsf{pk}, \mathsf{sk}) \leftarrow \mathsf{RABE.KeyGen}(\mathsf{crs}, \mathsf{aux}, P)$.
\end{itemize}

\item[]\textsf{RegPK}(\textsf{crs}, \textsf{aux}, \textsf{pk}, $P$):
\leavevmode
\begin{itemize}
    \item[$-$] Output $(\mathsf{mpk}, \mathsf{aux}') \leftarrow \mathsf{RABE.RegPK}(\mathsf{crs}, \mathsf{aux}, \mathsf{pk}, P)$.
\end{itemize}

\item[]\textsf{Encrypt}(\textsf{mpk}, X, $\mu$):
\leavevmode
\begin{itemize}
    \item[$-$] Generate $x, \theta \leftarrow \{0, 1\}^{\lambda}$ and generate $(\mathsf{vk}, \mathsf{sigk}) \leftarrow \mathsf{SIG.Gen}(1^{\lambda})$.
    \item[$-$] Generate a quantum state $|\psi\rangle$ by applying the map $|\nu\rangle |0 \ldots 0\rangle$ $\rightarrow$ $|\nu\rangle$ $ |\mathsf{SIG.Sign}(\mathsf{sigk}, \nu)\rangle$ to $|x\rangle_{\theta} \otimes |0 \ldots 0\rangle$.
    \item[$-$] Generate $\mathsf{rabe.ct} \leftarrow \mathsf{RABE.Encrypt}(\mathsf{mpk}, X, (\mathsf{sigk}, \theta, \mu \oplus \bigoplus_{i: \theta_i = 1} x_i))$.
    \item[$-$] Output $\mathsf{vk}$ and $\mathsf{ct} := (|\psi\rangle, \mathsf{rabe.ct})$.
\end{itemize}

\item[]\textsf{Update}(\textsf{crs}, \textsf{aux}, \textsf{pk}):
\leavevmode
\begin{itemize}
    \item[$-$] Output $\mathsf{hsk} \leftarrow \mathsf{RABE.Update}(\mathsf{crs}, \mathsf{aux}, \mathsf{pk})$.
\end{itemize}

\item[]\textsf{Decrypt}(\textsf{sk}, \textsf{hsk}, {X},\textsf{ct}):
\leavevmode
\begin{itemize}
    \item[$-$] Parse $\mathsf{ct}$ into a quantum state $\rho$ and classical string $\mathsf{rabe.ct}$.
    \item[$-$] Compute $(\mathsf{sigk}, \theta, \beta) \leftarrow \mathsf{RABE.Decrypt}(\mathsf{sk}, \mathsf{hsk},{X},\mathsf{rabe.ct})$.
    \item[$-$] If $\mathsf{RABE.Decrypt}$ outputs $\perp$ or $\mathsf{GetUpdate}$, return the same output.
    \item[$-$] Apply the map $|\nu\rangle |0 \ldots 0\rangle \rightarrow |\nu\rangle |\mathsf{SIG.Sign}(\mathsf{sigk}, \nu)\rangle$ to $\rho$, measure the first $\lambda$ qubits of the resulting state in Hadamard basis, and obtain $\bar{x}$.
    \item[$-$] Output $\mu \leftarrow \beta \oplus \bigoplus_{i: \theta_i = 1} \bar{x}_i$.
\end{itemize}

\item[]\textsf{Delete}(\textsf{ct}):
\leavevmode
\begin{itemize}
    \item[$-$] Parse $\mathsf{ct}$ into a quantum state $\rho$ and a classical string $\mathsf{rabe.ct}$.
    \item[$-$] Measure $\rho$ in the computational basis and obtain $x'$ and $\sigma'$.
    \item[$-$] Output $\mathsf{cert} = (x', \sigma')$.
\end{itemize}

\item[]\textsf{Verify}(\textsf{vk}, \textsf{cert}):
\leavevmode
\begin{itemize}
    \item[$-$] Parse $(z', \sigma') \leftarrow \mathsf{cert}$.
    \item[$-$] Output the result of $\mathsf{SIG.Verify}(\mathsf{vk}, z', \sigma')$.
\end{itemize}
\end{description}
\noindent\textbf{Correctness of \textsf{RABE-PubCED}.} The correctness of the \textsf{RABE-PubCED} protocol is guaranteed by the correctness of \textsf{RABE} and \textsf{SIG}.

\subsection{Proof of Security}
\begin{theorem}\label{thm:rabe-cepubv-security-proof}
If the signature scheme $\mathsf{SIG}$ ensures one-time unforgeability for BB84 states and the underlying $\mathsf{RABE}$ scheme is secure under the \textsf{LWE} assumption, then the proposed \textsf{RABE-PubVCED} construction guarantees certified everlasting security.
\end{theorem}
\begin{proof}
We begin by formalizing the adversary’s view in the security experiment. Define $Z_{\lambda}(\mathsf{vk}, \mathsf{sigk}, \theta, \beta, \mathcal{A})$ to be the following efficient quantum process:

\begin{itemize}\setlength\itemsep{0.5em}
    \item[$-$] The challenger first generates the common reference string {$\mathsf{crs}$ $\gets$ $\mathsf{Setup}(1^{\lambda}, 1^{\tau})$}.  

    \item[$-$] Using the registration procedure, the challenger produces the master public key $\mathsf{mpk}$ and any auxiliary state $\mathsf{aux}$ that may be required for the remainder of the experiment.  

    \item[$-$] The challenger then constructs the challenge ciphertext under the attribute set $X^*$ taken from the everlasting security experiment $\mathsf{Exp}^{\textsf{RABE-PubVCED}}_{\mathcal{A}}(\lambda, b).$ Specifically, it computes
    \begin{align*}
        \mathsf{ct}^* \gets \mathsf{RABE}.\mathsf{Encrypt}(\mathsf{mpk}, X^*, (\mathsf{sigk}, \theta, \beta)).
    \end{align*}

    \item[$-$] Finally, the process outputs the tuple
    \begin{align*}
        \big(\mathsf{crs}, \mathsf{mpk}, \mathsf{vk}, \mathsf{ct}^*\big).
    \end{align*}
\end{itemize}

\noindent We now analyze the indistinguishability of the adversary’s view under different distributions of the challenge inputs. Let $\mathcal{A}$ be any non-uniform \textsf{QPT} adversary. For any key pair $(\mathsf{vk}, \mathsf{sigk})$ of the signature scheme $\mathsf{SIG}$, for any challenge string $\theta \in \{0,1\}^{\lambda}$, bit $\beta \in \{0,1\}$, and for any efficiently samplable auxiliary quantum state $\lvert \psi \rangle_{\textcolor{gray}{\textsf{A}},\textcolor{gray}{\textsf{C}}}$ over registers \textcolor{gray}{\textsf{A}} and \textcolor{gray}{\textsf{C}}, the following two indistinguishability conditions hold due to the semantic security of the underlying $\mathsf{RABE}$ scheme.\vspace{0.2cm}

\noindent First, replacing the secret signing key $\mathsf{sigk}$ with a uniformly random string of the same length $\ell_{\mathsf{sigk}}$ does not alter the adversary’s distinguishing advantage by more than a negligible function:
\begin{align*}
&\Big|\Pr\big[\mathcal{A}(Z_{\lambda}(\mathsf{vk}, \mathsf{sigk}, \theta, \beta, \textcolor{gray}{\textsf{A}}), \textcolor{gray}{\textsf{C}}) = 1\big] \\
&\qquad - \Pr\big[\mathcal{A}(Z_{\lambda}(\mathsf{vk}, 0^{\ell_{\mathsf{sigk}}}, \theta, \beta, \textcolor{gray}{\textsf{A}}), \textcolor{gray}{\textsf{C}}) = 1\big]\Big|
 \leq \mathsf{negl}(\lambda).
\end{align*}

\noindent Second, replacing the challenge string $\theta$ with the all-zero string $0^{\lambda}$ also leaves the adversary’s view unchanged up to negligible statistical distance:
\begin{align*}
&\Big|\Pr\big[\mathcal{A}(Z_{\lambda}(\mathsf{vk}, \mathsf{sigk}, \theta, \beta, \textcolor{gray}{\textsf{A}}), \textcolor{gray}{\textsf{C}}) = 1\big] \\
&\qquad - \Pr\big[\mathcal{A}(Z_{\lambda}(\mathsf{vk}, \mathsf{sigk}, 0^{\lambda}, \beta, \textcolor{gray}{\textsf{A}}), \textcolor{gray}{\textsf{C}}) = 1\big]\Big|
 \leq \mathsf{negl}(\lambda).
\end{align*}

\noindent The first bound guarantees that the adversary cannot exploit the actual signing key $\mathsf{sigk}$, since it is computationally indistinguishable from random in the ciphertext view. The second bound ensures that the particular challenge string $\theta$ likewise provides no distinguishing advantage. Together, these facts imply that the adversary’s view in the challenge experiment is independent of the challenge bit $\beta$.\vspace{0.2cm}

\noindent Since the adversary’s distinguishing advantage is bounded by a negligible function, it follows that the challenge experiments with $b=0$ and $b=1$ are computationally indistinguishable. Hence, by invoking Lemma~\ref{lem:pv-certified-everlasting}, we conclude that the construction achieves certified everlasting security.

\end{proof}

\subsection{Instantiation of \textsf{RABE-PubVCED}}\label{Instant-RABE-PubVCE}

We present a concrete lattice-based instantiation of our \textsf{RABE-PubVCED} construction that attains post-quantum security. Our instantiation builds on two fundamental lattice-based primitives, both grounded in well-established and extensively studied lattice assumptions: \vspace{0.2cm}

\begin{description}\setlength\itemsep{0.5em}
\item[$\bullet$] \textbf{Registered \textsf{ABE} Component:} As mentioned earlier, we adopt the key-policy registered \textsf{ABE} scheme of Champion \textit{et al.} \cite{champion2025registered}, which supports arbitrary bounded-depth circuit policies. This scheme achieves semantic security under the falsifiable $\ell$-succinct \textsf{LWE} assumption in the random oracle model, and produces succinct ciphertexts whose size is independent of the attribute length. 

\item[$\bullet$] \textbf{Signature Component:} For the classical deterministic signature scheme satisfying one-time unforgeability for BB84, we follow the construction framework established by Kitagawa \textit{et al.} \cite{kitagawa2023publicly}. Their approach fundamentally relies on the existence of a secure one-way function as the underlying cryptographic primitive. To achieve post-quantum security, we instantiate this one-way function using lattice-based constructions derived from cyclic lattice structures developed in \cite{micciancio2007generalized}. This construction preserves the essential one-time unforgeability properties required for our quantum deletion mechanism while grounding the security analysis in the computational hardness of lattice problems, specifically the Learning With Errors and Short Integer Solution (\textsf{SIS}) assumptions that are believed to remain secure against quantum attacks.
\end{description}

\noindent Building on these components, we present a lattice-based instantiation of \textsf{RABE-PubVCED}. The scheme achieves post-quantum security under the $\ell$-succinct \textsf{LWE} and cyclic lattice assumptions, thereby demonstrating the feasibility of certified-deletion functionalities in the lattice setting.
\vspace{0.2cm}

\section*{Acknowledgements}
This work is supported in part by the Anusandhan National Research Foundation (ANRF), Department of Science and Technology (DST), Government of India, under File No. EEQ/2023/000164 and the Information Security Education and Awareness (ISEA) Project Phase-III initiatives of the Ministry of Electronics and Information Technology (MeitY) under Grant No. F.No. L-14017/1/2022-HRD.